\documentclass[english]{article}
 
\makeatletter

\usepackage{amsthm}
\usepackage{amsmath}
\usepackage{amssymb}
\usepackage{graphicx}
\usepackage{dsfont}
\usepackage{cite}
\usepackage{array}
\usepackage{color}
\usepackage{hyperref} 
\usepackage{enumitem}
\usepackage{float}
\usepackage[ruled,vlined]{algorithm2e}
\usepackage{chngcntr}

\providecommand{\tabularnewline}{\\}
\floatstyle{ruled}
\newfloat{algorithm}{tbp}{loa}
\providecommand{\algorithmname}{Algorithm}
\floatname{algorithm}{\protect\algorithmname}

\providecommand{\tabularnewline}{\\}
\floatstyle{ruled}
\newfloat{procedure}{tbp}{loa}
\providecommand{\procedurename}{Procedure}
\floatname{procedure}{\protect\procedurename}

\makeatletter
\newcommand{\manuallabel}[2]{\def\@currentlabel{#2}\label{#1}}
\makeatother

\theoremstyle{plain}

\theoremstyle{plain}

\theoremstyle{plain}
\newtheorem{lem}{\protect\lemmaname}
\theoremstyle{plain}
\newtheorem{thm}{\protect\theoremname}
\theoremstyle{plain}
\newtheorem{cor}{\protect\corollaryname}  
\theoremstyle{definition}
\newtheorem{defn}{\protect\definitionname}
\theoremstyle{definition}

\theoremstyle{definition}
\newtheorem{rem}{\protect\remarkname}
\theoremstyle{definition}
\newtheorem{example}{\protect\examplename}

\providecommand{\claimname}{Claim}
\providecommand{\lemmaname}{Lemma}
\providecommand{\propositionname}{Proposition}
\providecommand{\theoremname}{Theorem}
\providecommand{\corollaryname}{Corollary} 
\providecommand{\definitionname}{Definition}
\providecommand{\assumptionname}{Assumption}
\providecommand{\remarkname}{Remark}
\providecommand{\examplename}{Example}

\newcommand{\overbar}{\overline}

\DeclareMathOperator*{\argmin}{arg\,min}

\newcommand{\openone}{\mathds{1}}

\newcommand{\Mc}{\mathcal{M}}

\newcommand{\Ebar}{\overline{E}}
\newcommand{\rhoKL}{\rho_{\mathrm{KL}}}
\newcommand{\rhoKLn}{\rho^n_{\mathrm{KL}}}
\newcommand{\fV}{V}

\newcommand{\Nmax}{N_{\mathrm{\max}}}
\newcommand{\Nmin}{N_{\mathrm{\min}}}
\newcommand{\thetahat}{\hat{\theta}}

\newcommand{\Vol}{\mathrm{Vol}}
\newcommand{\BB}{\mathbb{B}}

\newcommand{\nnode}{n_{\mathrm{node}}}
\newcommand{\epsp}{\epsilon_{\mathrm{p}}}
\newcommand{\epsc}{\epsilon_{\mathrm{c}}}
\newcommand{\epscn}{\epsilon_{\mathrm{c},n}}
\newcommand{\NKL}{N_{\mathrm{KL}}}
\newcommand{\NKLn}{N_{\mathrm{KL},n}}

\newcommand{\cbar}{\overline{c}}
\newcommand{\cunder}{\underline{c}}
\newcommand{\Fscv}{\mathcal{F}_{\mathrm{scv}}}
\newcommand{\Gforest}{\mathcal{G}_{\mathrm{forest}}}
\newcommand{\GER}{G_{\mathrm{ER}}}
\newcommand{\Ptn}{P_{\theta}^n}
\newcommand{\Ptxn}{P_{\theta,\Xv}^n}
\newcommand{\Ptvn}{P_{\theta_v}^n}
\newcommand{\Ptvvn}{P_{\theta_{v'}}^n}

\newcommand{\Ndist}{\mathcal{N}}

\newcommand{\missing}{\ast}

\newcommand{\TV}{\mathrm{TV}}
\newcommand{\dTV}{d_{\mathrm{TV}}}

\newcommand{\Gbar}{\overbar{G}}

\newcommand{\Bernoulli}{\mathrm{Bernoulli}}

\newcommand{\pe}{P_{\mathrm{e}}}

\newcommand{\Xv}{\mathbf{X}}
\newcommand{\yv}{\mathbf{y}}
\newcommand{\Yv}{\mathbf{Y}}

\newcommand{\Zv}{\mathbf{Z}}

\newcommand{\Ac}{\mathcal{A}}

\newcommand{\Ec}{\mathcal{E}}
\newcommand{\Fc}{\mathcal{F}}
\newcommand{\Gc}{\mathcal{G}}

\newcommand{\Tc}{\mathcal{T}}

\newcommand{\Xc}{\mathcal{X}}

\newcommand{\EE}{\mathbb{E}}
\newcommand{\PP}{\mathbb{P}}
\newcommand{\RR}{\mathbb{R}}

\newcommand{\cov}{\mathrm{Cov}}

\newcommand{\Vc}{\mathcal{V}}

\newcommand{\Lc}{\mathcal{L}}

\newcommand{\Iv}{\mathbf{I}}

\newcommand{\bzero}{\mathbf{0}}

\usepackage{setspace}

\setenumerate[1]{label=\arabic*.}

\usepackage{geometry} 
\begin{document} 

\title{An Introductory Guide to Fano's Inequality  \\ with Applications in Statistical Estimation}
\author{Jonathan Scarlett\textsuperscript{1} and Volkan Cevher\textsuperscript{2} \\ [3mm]
{\footnotesize\textsuperscript{1} Department of Computer Science \& Department of Mathematics, National University of Singapore, Singapore} \\
{\footnotesize\textsuperscript{2} Laboratory for Information and Inference Systems (LIONS), EPFL, Switzerland} \\ [2mm]
\footnotesize{ \texttt{scarlett@comp.nus.edu.sg}, ~ \texttt{volkan.cevher@epfl.ch} }}
\date{}
\maketitle

\bigskip
\bigskip

\begin{abstract}
    Information theory plays an indispensable role in the development of algorithm-independent impossibility results, both for communication problems and for seemingly distinct areas such as statistics and machine learning.  While numerous information-theoretic tools have been proposed for this purpose, the oldest one remains arguably the most versatile and widespread: Fano's inequality.  In this chapter, we provide a survey of Fano's inequality and its variants in the context of statistical estimation, adopting a versatile framework that covers a wide range of specific problems.  We present a variety of key tools and techniques used for establishing impossibility results via this approach, and provide representative examples covering group testing, graphical model selection, sparse linear regression, density estimation, and convex optimization.
\end{abstract}

%\long\def\symbolfootnote[#1]#2{\begingroup\def\thefootnote{\fnsymbol{footnote}}\footnote[#1]{#2}\endgroup}
%
%\symbolfootnote[0]{ The authors are with the Laboratory for Information and Inference Systems (LIONS), \'Ecole Polytechnique F\'ed\'erale de Lausanne (EPFL) (e-mail: \{jonathan.scarlett,volkan.cevher\}@epfl.ch).
%
%This work was supported in part by the European Commission under Grant ERC Future Proof, SNF 200021-146750 and SNF CRSII2-147633, and EPFL Fellows Horizon2020 grant 665667.}
%\vspace*{-0.5cm}

\newpage
 \tableofcontents

\newcommand{\rev}[1]{{\textcolor{blue}{#1}}}

\newpage
\section{Introduction} \label{sec:intro}

The tremendous progress in large-scale statistical inference and learning in recent years has been spurred by both practical and theoretical advances, with strong interactions between the two:  Algorithms that come with {\em a priori} performance guarantees are clearly desirable, if not crucial, in practical applications, and practical issues are indispensable in guiding the theoretical studies.  

Complementary to performance bounds for specific algorithms, a key role is also played by {algorithm-independent impossibility results}, stating conditions under which one cannot hope to achieve a certain goal.  Such results provide definitive benchmarks for practical methods, serve as certificates for near-optimality, and help guide the practical developments towards directions where the greatest improvements are possible.

Since its introduction in 1948, the field of information theory has continually provided such benefits for the problems of storing and transmitting data, and has accordingly shaped the design of practical communication systems.  In addition, recent years have seen mounting evidence that the tools and methodology of information theory reach far beyond communication problems, and can provide similar benefits {\em within the entire data processing pipeline}. % This emerging perspective is particular relevant in problems of learning from data, which can be interpreted as a {\em transfer of information}.}

% Intuitively, problems of learning from data can be interpreted as a {transfer of information}, and this perspective naturally lends itself to the analysis techniques that were originally developed for understanding communication systems.

While many information-theoretic tools have been proposed for establishing impossibility results, the oldest one remains arguably the most versatile and widespread: Fano's inequality \cite{FanoNotes}. This fundamental inequality is not only ubiquitous in studies of communication, but has been applied extensively in statistical inference and learning problems; several examples are given in Table \ref{tbl:applications}.

When applying Fano's inequality to such problems, one typically encounters a number of distinct challenges compared to those found in communication problems. The goal of this chapter is to introduce the reader to some of the key tools and techniques, explain their interactions and connections, and provide several representative examples.

% We proceed by introducing the mathematical framework that will be adopted throughout the chapter, and then give a high-level overview of the methodology.

%At a high level, this result provides a relation between two fundamental quantities concerning a discrete random variable $V$ and its estimate $\hat{V}$:
%\begin{itemize}
%    \item The {\em error probability}, defined as $\PP[\hat{V} \ne V]$;
%    \item The {\em mutual information} $I(V;\hat{V})$, which quantifies the amount of information that $\hat{V}$ reveals about $V$.
%\end{itemize}
%More precisely, unless $\hat{V}$ reveals a sufficient amount of information about $V$, the two cannot have a high probability of coinciding; hence, any upper bound on the mutual information leads to a corresponding lower bound on the error probability.  This idea led to the first proof of the converse part of Shannon's channel coding theorem \cite{FanoNotes}, and has subsequently been applied to a broad range of data science applications.  A partial list of examples is given in Table \ref{tbl:applications}. 

% The application of Fano's inequality to data science problems typically comes with a variety of distinct challenges compared to communication problems, despite having some common aspects.  In the remainder of this section, we overview the general approach adopted, and highlight the key steps involved. 

\subsection{Overview of Techniques} \label{sec:overview}

Throughout the chapter, we consider the following statistical estimation framework, which captures a broad range of problems including the majority of those listed in Table \ref{tbl:applications}:
\begin{itemize}
    \item There exists an unknown parameter $\theta$, known to lie in some set $\Theta$ (e.g., a subset of $\RR^p$), that we would like to estimate. 
    \item In the simplest case, the estimation algorithm has access to a set of {\em samples} $\Yv = (Y_1,\dotsc,Y_n)$ drawn from some joint distribution $\Ptn(\yv)$ parametrized by $\theta$.  More generally, the samples may be drawn from some joint distribution $\Ptxn(\yv)$ parametrized by $(\theta,\Xv)$, where $\Xv = (X_1,\dotsc,X_n)$ are {\em inputs} that are either known in advance or selected by the algorithm itself.
    \item Given knowledge of $\Yv$, as well as $\Xv$ if inputs are present, the algorithm forms an estimate $\hat{\theta}$ of $\theta$, with the goal of the two being ``close'' in the sense that some {\em loss function} $\ell(\theta,\hat{\theta})$ is small. When referring to this step of the estimation algorithm, we will use the terms {\em algorithm} and {\em decoder} interchangeably.
\end{itemize}
We will initially use the following simple running example to exemplify some of the key concepts, and then turn to detailed applications in Sections \ref{sec:apps_discrete} and \ref{sec:apps_cont}.

\hspace*{-1cm}
\begin{table}
    \begin{centering}
        \begin{tabular}{|>{\centering}p{4.75cm}|>{\centering}p{1.5cm}|>{\centering}p{3.75cm}|>{\centering}p{1.5cm}|}
            \hline 
            \multicolumn{2}{|>{\centering}p{5.75cm}|}{\textbf{Sparse and low rank problems}} & \multicolumn{2}{>{\centering}p{5cm}|}{\textbf{Other estimation problems}}\tabularnewline
            \hline 
            \underline{Problem} & \underline{References} & \underline{Problem} & \underline{References}\tabularnewline
            \hline 
            Group testing
            
            Compressive sensing
            
            Sparse Fourier transform
            
            Principal component analysis
            
            Matrix completion & \cite{Mal78,Ati12} %Mal80,,Sca16b} % Group testing
            
            \cite{Wai09,Can13} %Wan10a,Can13,Ree13,Aks17} % CS
            
            \cite{Has12,Cev16} % Sparse FT
            
            \cite{Ami09,Vu12} %,Cai13a,Bir13} % Sparse PCA
            
            \cite{Neg12,Dav14} %,Ma15,Sha17} % Matrix completion
            & Regression
            
            Density estimation
            
            Kernel methods
            
            Distributed estimation
            
            Local privacy 
            &  \cite{Ibr77,Yan99} %Ras11} % Regression
            
            \cite{Bir83,Yan99} %Bir86,Yan99,Yu97} % Density estimation
            
            \cite{Ras12,Yan17} %  Kernel methods 
            
            \cite{Zha13,Xu17} % Distributed setting 
            
            \cite{Duc13a} % Local privacy 
            \tabularnewline
            \hline 
            \multicolumn{2}{|c|}{\textbf{Sequential decision problems}} & \multicolumn{2}{c|}{\textbf{Other learning problems}}\tabularnewline
            \hline 
            \underline{Problem} & \underline{References} & \underline{Problem} & \underline{References}\tabularnewline
            \hline 
            Convex optimization
            
            Active learning
            
            Multi-armed bandits
            
            Bayesian optimization
            
            Communication complexity &  \cite{Rag11,Aga12} % Convex optimization
            
            \cite{Rag11a} % Active learning 
            
            \cite{Aga17}  % Bandits
            
            \cite{Sca18a} % BO
            
            \cite{Bar02} % Communication complexity
            & Graph learning
            
            Ranking
            
            Classification 
            
            Clustering 
            
            Phylogeny & 
            \cite{San12,Sha14} % Graphical model selection
            
            \cite{Sha15,Pan17} % Ranking from pairwise comparisons 
            
            \cite{Yan99a,Nok15} % Classification
            
            \cite{Maz17a} % Clustering
            
            \cite{Mos04} % Phylogeny
            \tabularnewline
            \hline 
        \end{tabular}
        \caption{Examples of applications for which impossibility results have been derived using Fano's inequality. \label{tbl:applications}}
        \par
    \end{centering}
\end{table}

\begin{example}
    ($1$-sparse linear regression)
    A vector parameter $\theta \in \RR^p$ is known to have at most one non-zero entry, and we are given $n$ linear samples of the form $\Yv = \Xv\theta + \Zv$,\footnote{Throughout the chapter, we interchange tuple-based notations such as $\Xv = (X_1,\dotsc,X_n)$, $\Yv = (Y_1,\dotsc,Y_n)$ with vector/matrix notation such as $\Xv\in\RR^{n \times p}$, $\Yv \in \RR^n$.}  where $\Xv \in \RR^{n \times p}$ is a known input matrix, and $\Zv \sim \Ndist(\bzero, \sigma^2 \Iv)$ is additive Gaussian noise.  In other words, the $i$-th sample $Y_i$ is a noisy sample of $\langle X_i, \theta \rangle$, where $X_i \in \RR^p$ is the transpose of the $i$-th row of $\Xv$.  The goal is to construct an estimate $\hat{\theta}$ such that the squared distance $\ell(\theta,\hat{\theta}) = \|\theta - \hat{\theta}\|_2^2$ is small.  
    
    This example is an extreme case of {\em $k$-sparse linear regression}, in which $\theta$ has at most $k \ll p$ non-zero entries, i.e., at most $k$ columns of $\Xv$ impact the output.  The more general $k$-sparse recovery problem will be considered in Section \ref{sec:sparse}. 
\end{example}

We seek to establish algorithm-independent impossibility results, henceforth referred to as {\em converse bounds}, in the form of lower bounds on the {\em sample complexity}, i.e., the number of samples $n$ required to achieve a certain average target loss.  
% Broadly speaking, an information-theoretic impossibly result, commonly referred to as a {\em converse bound}, states that under a certain mathematical model, it is impossible for {\em any decoder} to achieve a certain estimation goal.  This includes not only practical decoders, but also computationally intractable decoders performing methods such as exhaustive search.  Of particular of interest in the above estimation framework is {\em lower bounds on the required number of samples} for a given target average loss.
The following aspects of the problem significantly impact this goal, and their differences are highlighted throughout the chapter:
\begin{itemize}
    \item \underline{Discrete vs.~continuous}: 
    Depending on the application, the parameter set $\Theta$ may be discrete or continuous.  For instance, in the $1$-sparse linear regression example, one may consider the case that $\theta$ is known to lie in a finite set $\Theta \subseteq \RR^p$, or one may consider the general estimation of a vector in  the set
    \begin{equation}
        \Theta = \{ \theta \in \RR^p \,:\, \|\theta\|_0 \le 1 \}, \label{eq:Theta_1s}
    \end{equation}
    where $\|\theta\|_0$ is the number of non-zeros in $\theta$.  % In the discrete case, the reduction to multiple hypothesis testing is often immediate, whereas in the continuous case, a discretization argument is used.
    \item \underline{Minimax vs.~Bayesian}: 
    In the minimax setting, one seeks a decoder that attains a small loss for any given $\theta \in \Theta$, whereas in the Bayesian setting, one considers the average performance under some prior distribution on $\theta$.  Hence, these two variations respectively consider the worst-case and average-case performance with respect to $\theta$.  We focus primarily on the minimax setting throughout the chapter, and further discuss Bayesian settings in Section \ref{sec:generalizations}.
    \item \underline{Choice of target goal}: Naturally, the target goal can considerably impact the fundamental performance limits of an estimation problem.  For instance, in discrete settings, it is common to consider exact recovery, requiring that $\hat{\theta} = \theta$ (i.e., the 0-1 loss $\ell(\theta,\thetahat) = \openone\{ \thetahat \ne \theta \}$), but it is also of interest to understand to what extent approximate recovery criteria make the problem easier.  % In both discrete and continuous settings, different recovery criteria can lead to considerably different fundamental limits.
    \item \underline{Non-adaptive vs.~adaptive sampling}: 
    In settings consisting of an input $\Xv = (X_1,\dotsc,X_n)$ as introduced above, one often distinguishes between the {\em non-adaptive} setting, in which $\Xv$ is specified prior to observing any samples, and the {\em adaptive} setting, in which a given input $X_i$ can be designed based on the past inputs $(X_1,\dotsc,X_{i-1})$ and samples $(Y_1,\dotsc,Y_{i-1})$.  It is of significant interest to understand to what extent the additional freedom of adaptivity impacts the performance.
\end{itemize}
With these variations in mind, we proceed by outlining the main steps in obtaining converse bounds for statistical estimation via Fano's inequality.

% The steps shown in Procedure \ref{alg:steps} are highly inter-related; for instance, the reduction to hypothesis testing can typically done in many ways, and the specific choice considerably impacts how high the associated mutual information term is.  We proceed by discussing the steps in more detail.

\subsubsection{Step 1: Reduction to Multiple Hypothesis Testing}  \label{sec:step1}

The {\em multiple hypothesis testing} problem is defined as follows: An index $V \in \{1,\dotsc,M\}$ is drawn from a prior distribution $P_V$, and a sequence of samples $\Yv = (Y_1,\dotsc,Y_n)$ is drawn from a probability distribution $P_{\Yv|V}$ parametrized by $V$. The $M$ possible conditional distributions are known in advance, and the goal is to identify the index $V$ with high probability given the samples.

In Figure \ref{fig:Reduction}, we provide a general illustration of how an estimation problem can be reduced to multiple hypothesis testing, possibly with the added twist of including inputs $\Xv = (X_1,\dotsc,X_n)$.  Supposing for the time being that we are in the minimax setting, the idea is to construct a {hard subset} of parameters $\{\theta_1,\dotsc,\theta_M\}$ that are difficult to distinguish given the samples.  We then {lower bound the worst-case performance by the average} over this hard subset.  As a concrete example, a good choice for the 1-sparse linear regression problem is to set $M = 2p$ and consider the set of vectors of the form
\begin{equation}
    \theta = (0,\dotsc,0,\pm\epsilon,0,\dotsc,0), \label{eq:theta_1s}
\end{equation}
where $\epsilon > 0$ is a constant.  Hence, the non-zero entry of $\theta$ has a given magnitude, which can be selected to our liking for the purpose of proving a converse.

We envision an index $V \in \{1,\dotsc,M\}$ being drawn uniformly at random and used to select the corresponding parameter $\theta_V$, and the estimation algorithm being run to produce an estimate $\hat{\theta}$.  If the parameters $\{\theta_1,\dotsc,\theta_M\}$ are not too close and the algorithm successfully produces $\hat{\theta} \approx \theta_V$, then we should be able to infer the index $V$ from $\hat{\theta}$.  This entire process can be viewed as a problem of multiple hypothesis testing, where the $v$-th hypothesis is that the underlying parameter is $\theta_v$ ($v=1,\dotsc,M$).  With this reduction, we can deduce that if the algorithm performs well then the hypothesis test is successful; the contrapositive statement is then that {\em if the hypothesis test cannot be successful, then the algorithm cannot perform well}.

In the 1-sparse linear regression example, we find from \eqref{eq:theta_1s} that distinct $\theta_j,\theta_{j'}$ must satisfy $\|\theta_j - \theta_{j'}\|_2 \ge \sqrt{2} \cdot \epsilon$.  As a result, we immediately obtain from the triangle inequality that the following holds:
\begin{equation}
    \text{If}~~~~ \|\hat{\theta} - \theta_v\|_2 < \frac{\sqrt{2}}{2} \cdot \epsilon, \text{~~~~then~~~~} \argmin_{v' = 1,\dotsc,M} \|\hat{\theta} - \theta_{v'}\| = v.
\end{equation}
In other words, if the algorithm yields $\|\hat{\theta} - \theta_v\|_2^2 < \frac{\sqrt{2}}{2} \epsilon$, then $V$ can be identified as the index corresponding to the closest vector to $\hat{\theta}$.  Thus, sufficiently accurate estimation implies success in identifying $V$.

\begin{figure}
    \begin{centering}
        \includegraphics[width=0.8\columnwidth]{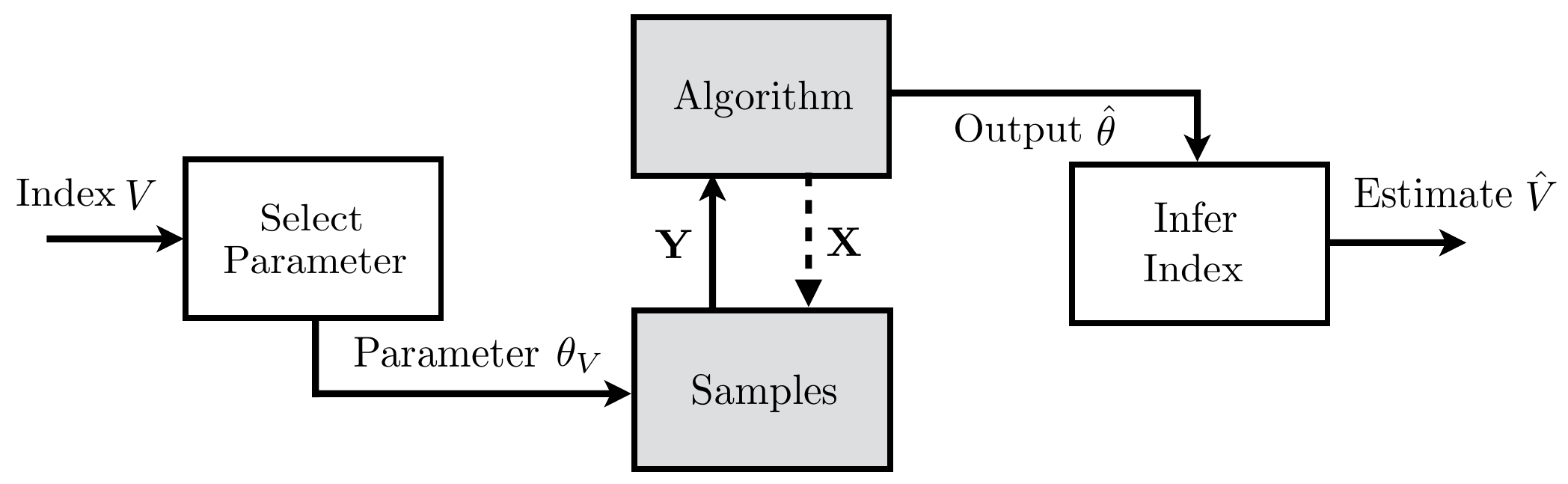}
        \par
    \end{centering}
    
    \caption{Reduction of minimax estimation to multiple hypothesis testing.  The gray boxes are fixed as part of the problem statement, whereas the white boxes are constructed to our liking for the purpose of proving a converse bound.  The dashed line marked with $\Xv$ is optional, depending on whether inputs are present. \label{fig:Reduction}}
\end{figure}

{\bf Discussion.} Selecting the hard subset $\{\theta_1,\dotsc,\theta_M\}$ of parameters is often considered somewhat of an art.  While the proofs of existing converse bounds may seem easy in hindsight when the hard subset is known, coming up with a suitable choice for a new problem usually requires some creativity and/or exploration.  % It is typically unclear {\em a priori} which choice will lead to a good converse bound, and the best approach can vary significantly from problem to problem.  
Despite this, there exist general approaches that have proved to been effective in a wide range problems, which we exemplify in Sections \ref{sec:apps_discrete} and \ref{sec:apps_cont}.

In general, selecting the hard subset requires balancing conflicting goals: Increasing $M$ so that the hypothesis test is more difficult, keeping the elements ``close'' so that they are difficult to distinguish, and keeping the elements ``sufficiently distant'' so that one can recover $V$ from $\hat{\theta}$. Typically, one of the following three approaches is adopted: (i)  explicitly construct a set whose elements are known or believed to be difficult to distinguish; (ii) prove the existence of such a set using probabilistic arguments; or (iii) consider packing as many elements as possible into the entire space.  We will provide examples of all three kinds.

In the Bayesian setting, $\theta$ is already random, so we cannot use the above-mentioned method of lower bounding the worst-case performance by the average.  Nevertheless, if $\Theta$ is discrete, we can still use the trivial reduction $V = \theta$ to form a multiple hypothesis testing problem with a possibly non-uniform prior.  In the continuous Bayesian setting, one typically requires more advanced methods not covered in this chapter; we provide further discussion in Section \ref{sec:generalizations}.

\subsubsection{Step 2: Application of Fano's Inequality}

Once a multiple hypothesis test is set up, Fano's inequality provides a lower bound on its error probability in terms of the mutual information, which is one of the most fundamental information measures in information theory.  The mutual information can often be explicitly characterized given the problem formulation, and a variety of useful properties are known for doing so, as outlined below.

We briefly state the standard form of Fano's inequality \index{subject}{Fano's inequality} for the case that $V$ is uniform on $\{1,\dotsc,M\}$ and $\hat{V}$ is some estimate of $V$:
\begin{equation}
    \PP[\hat{V}\ne V] \ge 1 - \frac{I(V;\hat{V}) + \log 2}{\log M}. \label{eq:Fano_intro}
\end{equation}
The intuition is as follows: The term $\log M$ represents the prior uncertainty (i.e., entropy) of $V$, and the mutual information $I(V;\hat{V})$ represents how much information $\hat{V}$ reveals about $V$.  In order to have a small probability of error, we require that the information revealed is close to the prior uncertainty.

Beyond the standard form of Fano's inequality \eqref{eq:Fano_intro}, it is useful to consider other variants, including approximate recovery and conditional versions.  These are the topic of Section \ref{sec:fano}, and we discuss other alternatives in Section \ref{sec:generalizations}.

\subsubsection{Step 3: Bounding the Mutual Information}

In order to make lower bounds such as \eqref{eq:Fano_intro} explicit, we need to upper bound the mutual information therein. This often consists of tedious yet routine calculations, but there are cases where it is highly non-trivial.  The mutual information depends crucially on the choice of reduction in the first step.

The joint distribution of $(V,\hat{V})$ is decoder-dependent and usually very complicated, so to simplify matters, the typical first step is to apply an upper bound known as the data processing inequality. 
In the simplest case that there is no extra input to the sampling mechanism (i.e., $\Xv$ is absent in Figure \ref{fig:Reduction}), this inequality takes the form $I(V;\hat{V}) \le I(V;\Yv)$ under the Markov chain $V \to \Yv \to \hat{V}$.  Thus, we are left to answer the question of {how much information the samples reveal about the index $V$}.  

In Section \ref{sec:mi_bounds}, we introduce several useful tools for this purpose, including:
\begin{itemize}
    \item \underline{Tensorization}:
    If the samples $\Yv = (Y_1,\dotsc,Y_n)$ are conditionally independent given $V$, we have $I(V;\Yv) \le \sum_{i=1}^n I(V;Y_i)$.  Bounds of this type simplify the mutual information containing a set of observations to simpler terms containing only a single observation.
    \item \underline{KL divergence based bounds}: 
    Straightforward bounds on the mutual information reveal that if $\{\Ptvn\}_{v = 1,\dotsc,M}$ are close in terms of KL divergence, then the mutual information is small.  Results of this type are useful, as the relevant KL divergences can often be evaluated exactly or tightly bounded.
    % \item \underline{Covering}: Continuing on the previous point, instead of insisting that {\em all} of the distributions are close in KL divergence, one can try to ``cover'' the space with auxiliary distributions $Q^n_1,\dotsc,Q^n_N$ such that each $\Ptvn$ is close to {\em at least one} $Q^n_j$.  This approach can allow greater flexibility for bounding the mutual information.
\end{itemize}
In addition to these, we introduce variations for cases that the input $\Xv$ is present in Figure \ref{fig:Reduction}, distinguishing between non-adaptive and adaptive sampling.

{\bf Toy example.} To give a simple example of how this step is combined with the previous one, consider the case that we wish to identify one of $M$ hypotheses, with the $v$-th hypothesis being that $\Yv \sim P_v(\yv)$ for some distribution $P_v$ on $\{0,1\}^n$.  That is, the $n$ observations $(Y_1,\dotsc,Y_n)$ are binary-valued.  Starting with the above-mentioned bound  $I(V;\hat{V}) \le I(V;\Yv)$, we simply write $I(V;\Yv) \le H(\Yv) \le n \log 2$, which follows since $\Yv$ takes one of at most $2^n$ values.  Substitution into \eqref{eq:Fano_intro} yields $\pe \ge 1 - \frac{n+1}{\log_2 M}$, which means that achieving $\pe \le \delta$ requires $n \ge (1-\delta)\log_2 M - 1$.  This formalizes the intuitive fact that reliably identifying one of $M \gg 1$ hypotheses requires roughly $\log_2 M$ binary observations.

%\subsection{Summary}
%
%For reference, the steps described above are summarized as follows.
%
%\begin{procedure}[h]
%    \caption{ Steps for obtaining converse bounds for estimation problems via Fano's inequality. \label{alg:steps} }
%    
%    \begin{enumerate}
%        \item Reduce the estimation problem to a multiple hypothesis testing problem;
%        \item Use Fano's inequality to lower bound the error probability in terms of the mutual information;
%        \item Upper bound the mutual information and deduce the final converse bound.
%    \end{enumerate}
%\end{procedure}

\section{Fano's Inequality and its Variants} \label{sec:fano}

In this section, we state various forms of Fano's inequality that will form the basis for the results in the remainder of the chapter.  % Here and in subsequent sections, we defer several of the proofs to the appendix, particularly those for results that are standard in introductory information theory textbooks.

% While these results pertain to the estimation of discrete random variables, they will be directly used when handling continuous problems in Section \ref{sec:continuous}.

\subsection{Standard Version}

We begin with the most simple and widely-used form of Fano's inequality.  We use the generic notation $V$ for the discrete random variable in a multiple hypothesis test, and we write its estimate as $\hat{V}$.  In typical applications, one has a Markov chain relation such as $V \to \Yv \to \hat{V}$, where $\Yv$ is the collection of samples; we will exploit this fact in Section \ref{sec:mi_bounds}, but for now, one can think of $\hat{V}$ being randomly generated by {any} means given $V$.  

The two fundamental quantities appearing in Fano's inequality are the conditional entropy $H(V|\hat{V})$, representing the uncertainty of $V$ given its estimate, and the {\em error probability}: 
\begin{equation}
    \pe = \PP[\hat{V} \ne \fV]. \label{eq:pe}
\end{equation}
Since $H(V|\hat{V}) = H(V) - I(V;\hat{V})$, the conditional entropy is closely related to the mutual information, representing how much information $\hat{V}$ reveals about $V$.

\begin{thm} \label{thm:Fano}
    {\em (Fano's inequality)} 
    For any discrete random variables $\fV$ and $\hat{V}$ on a common finite alphabet $\Vc$, we have
    \begin{equation}
        H(V | \hat{V}) \le H_2(\pe) + \pe \log\big(|\Vc| - 1\big), \label{eq:Fano1}
    \end{equation}
    where $H_2(\alpha) = \alpha\log\frac{1}{\alpha} + (1-\alpha)\log\frac{1}{1-\alpha}$ is the binary entropy function.  In particular, if $V$ is uniform on $\Vc$, we have
    \begin{equation}
        I(\fV;\hat{V}) \ge (1-\pe) \log|\Vc| - \log 2, \label{eq:Fano2}
    \end{equation}
    or equivalently,
    \begin{equation}
        \pe \ge 1 - \frac{I(V;\hat{V}) + \log 2}{\log|\Vc|}. \label{eq:Fano3}
    \end{equation}
\end{thm}
%\begin{proof}
%
%\end{proof}
% \noindent The proof is given in Appendix \ref{sec:pf_Fano}.

Since the proof of Theorem \ref{thm:Fano} is widely accessible in standard references such as \cite{Cov01}, we provide only an intuitive explanation of \eqref{eq:Fano1}: To resolve the uncertainty in $V$ given $\hat{V}$, we can first ask whether the two are equal, which bears uncertainty $H_2(\pe)$.  In the case that they differ, which only occurs a fraction $\pe$ of the time, the remaining uncertainty is at most $\log\big(|\Vc| - 1\big)$.

\begin{rem} \label{rem:FanoWeaken}
%    In typical applications of \eqref{eq:Fano1}--\eqref{eq:Fano2}, one weakens the bound in two ways: (i) If $\hat{V}$ is formed based on a set of observations $\Vc$ such that $V \to \Yv \to \Vc$ forms a Markov chain, then one can bound $I(V;\hat{V}) \le I(V;\Yv)$ by the data processing inequality.  (ii) The binary entropy function satisfies $H_2(\pe) \le \log 2$.  For example, applying these bounds to \eqref{eq:Fano2} gives
%    \begin{equation}
%        I(\fV;\Yv) \ge (1-\pe) \log|\Vc| - \log 2, \label{eq:Fano2a}
%    \end{equation}
%    or equivalently
%    \begin{equation}
%        \pe \ge 1 - \frac{I(V;\Yv) + \log 2}{\log|\Vc| }. \label{eq:Fano2b}
%    \end{equation}

    For uniform $\fV$, we obtain \eqref{eq:Fano2} by upper bounding $|\Vc| - 1 \le |\Vc|$ and $H_2(\pe) \le \log 2$ in \eqref{eq:Fano1}, and subtracting $H(V) = \log|\Vc|$ on both sides.  While these additional bounds have a minimal impact for moderate to large values of $|\Vc|$, a notable case where one should use \eqref{eq:Fano1} is the binary setting, i.e., $|\Vc| = 2$.  In this case, \eqref{eq:Fano2} is meaningless due to the right-hand side being negative, whereas \eqref{eq:Fano1} yields the following for uniform $V$:
    \begin{equation}
        I(V; \hat{V}) \ge \log 2 - H_2(\pe). \label{eq:FanoM2a}
    \end{equation}
    It follows that the error probability is lower bounded as
    \begin{equation}
        \pe \ge H_2^{-1}\big( \log 2 - I(V;\hat{V}) \big), \label{eq:FanoM2b}
    \end{equation}
      where $H_2^{-1}(\cdot) \in \big[0,\frac{1}{2}\big]$ is the inverse of $H_2(\cdot) \in [0,\log2]$ on the domain $\big[0,\frac{1}{2}\big]$.
\end{rem}

%\begin{rem} \label{rem:FanoAltProof}
%    In the case that $\fV$ is uniform on $\Vc$, we can provide an alternative proof of Fano's inequality that is more amenable to certain generalizations ({\em cf.}, Section \ref{sec:generalizations}):
%    \begin{align}
%        I(V;\hat{V}) 
%            &= D(P_{\fV\hat{V}} \| P_{\fV}\times P_{\hat{V}})  \\
%            &\ge D_2\big( P_{\fV\hat{V}}[E] \,\|\, (P_{\fV}\times P_{\hat{V}})[E] \big) \label{eq:Fano_alt2} \\
%            &= D_2\Big( \pe \,\Big\|\, 1-\frac{1}{|\Vc|} \Big) \label{eq:Fano_alt3} \\
%            &= (1-\pe) \log\big( |\Vc| (1-\pe) \big) + \pe\log\frac{\pe}{1-\frac{1}{|\Vc|}} \label{eq:Fano_alt4} \\
%            &= (1-\pe) \log |\Vc| + \pe\log\frac{1}{1 - \frac{1}{|\Vc|}} - H_2(\pe), \label{eq:Fano_alt5}
%    \end{align}
%    where \eqref{eq:Fano_alt2} holds with $d(p\|q) = p\log\frac{p}{q} + (1-p)\log\frac{1-p}{1-q}$ from the data processing inequality, \eqref{eq:Fano_alt3} follows from since if $\fV$ and $\hat{V}$ are generated independently then the probability of the two being equal is $\frac{1}{|\Vc|}$ by uniformity, and \eqref{eq:Fano_alt4}--\eqref{eq:Fano_alt5} follows from the definitions of $d(\cdot\|\cdot)$ and $H_2(\cdot)$.  From \eqref{eq:Fano_alt5}, we obtain \eqref{eq:Fano2} by writing $\pe\log\frac{1}{1-\frac{1}{|\Vc|}} \ge 0$ and $H_2(\pe) \le \log 2$, and we also obtain \eqref{eq:FanoM2a} by setting $|\Vc| = 2$.
%\end{rem}

\subsection{Approximate Recovery}

The notion of error probability considered in Theorem \ref{thm:Fano} is that of {exact recovery}, insisting that $\hat{V} = V$.  More generally, one can consider notions of {\em approximate recovery}, where one only requires $\hat{V}$ to be ``close'' to $V$ in some sense.  This is useful for at least two reasons:
\begin{itemize}
    \item Exact recovery is often a highly stringent criterion in discrete statistical estimation problems, and it is of considerable interest to understand to what extent moving to approximate recovery makes the problem easier;
    \item When we reduce continuous estimation problems to the discrete setting ({\em cf.}, Section \ref{sec:continuous}), permitting approximate recovery will provide a useful additional degree of freedom.
\end{itemize}
We consider a general setup with a random variable $V$, an estimate $\hat{V}$, and an error probability of the form 
\begin{equation}
\pe(t) = \PP\big[ d(V,\hat{V}) > t \big] \label{eq:pe_partial}
\end{equation}
for some real-valued function $d(v,\hat{v})$ and threshold $t \in \RR$.  In contrast to the exact recovery setting, there are interesting cases where $V$ and $\hat{V}$ are defined on different alphabets, so we denote these by $\Vc$ and $\hat{\Vc}$, respectively.

One can interpret \eqref{eq:pe_partial} as requiring $\hat{V}$ to be within a ``distance'' $t$ of $V$.  However, $d$ need not be a true distance function, and need not even be symmetric nor take non-negative values.  This definition of error probability in fact entails no loss of generality, since one can set $t=0$ and $d(V,\hat{V}) = \openone\{(V,\hat{V}) \in \Ec\}$ for an arbitrary set $\Ec$ containing the pairs that are considered errors.

In the following, we make use of the quantities
\begin{equation}
\Nmax(t) = \max_{\hat{v} \in \hat{\Vc}} N_{\hat{v}}(t), \qquad \Nmin(t) = \min_{\hat{v} \in \hat{\Vc}} N_{\hat{v}}(t),
\end{equation}
where
\begin{equation}
N_{\hat{v}}(t) = \sum_{v \in \Vc} \openone\{ d(v,\hat{v}) \le t \}
\end{equation}
counts the number of $v \in \Vc$ within a ``distance'' $t$ of $\hat{v} \in \hat{\Vc}$. 

\begin{thm} \label{thm:Partial}
    {\em (Fano's inequality with approximate recovery)}
    For any random variables $V,\hat{V}$ on the finite alphabets $\Vc,\hat{\Vc}$, we have
    \begin{equation}
    H(\fV|\hat{V}) \le H_2( \pe(t) ) + \pe(t)\log\frac{|\Vc| - \Nmin(t)}{\Nmax(t)} + \log \Nmax(t). \label{eq:Partial1}
    \end{equation}
    In particular, if $\fV$ is uniform on $\Vc$, then
    \begin{equation}
    I(\fV;\hat{V}) \ge (1 - \pe(t))\log\frac{|\Vc|}{\Nmax(t)} - \log 2, \label{eq:Partial2}
    \end{equation}
    or equivalently
    \begin{equation}
    \pe(t) \ge 1 - \frac{I(\fV;\hat{V}) + \log 2}{ \log\frac{|\Vc|}{\Nmax(t)} }. \label{eq:Partial3}
    \end{equation}
\end{thm}
%\begin{proof}
%
%\end{proof}
\noindent The proof is similar to that of Theorem \ref{thm:Fano}, and can be found in \cite{Duc13}.%\footnote{All omitted proofs can be found in the supplementary material, which is included for the reviewers' reference.} % The proof is given in Appendix \ref{sec:pf_Partial}.

By setting $d(v,\hat{v}) = \openone\{ v \ne \hat{v} \}$ and $t=0$, we find that Theorem \ref{thm:Partial} recovers Theorem \ref{thm:Fano} as a special case.  More generally, the bounds \eqref{eq:Partial2}--\eqref{eq:Partial3} resemble those for exact recovery in \eqref{eq:Fano2}--\eqref{eq:Fano3}, but $\log|\Vc|$ is replaced by $\log\frac{|\Vc|}{\Nmax(t)}$.  When $\Vc = \hat{\Vc}$, one can intuitively think of the approximate recovery setting as dividing the space into regions of size $\Nmax(t)$, and only requiring the correct region to be identified, thereby reducing the effective alphabet size to $\frac{|\Vc|}{\Nmax(t)}$.

%\begin{cor} \label{cor:list}
%    \emph{(List decoding)}
%    For any random variable $V$ on a finite alphabet $\Vc$, and any random list $\Lc$ consisting of $L$ elements from $\Vc$, it holds that
%    \begin{equation}
%        \PP[ V \notin \Lc ] \ge 1 - \frac{I(V;\Lc) + \log 2}{ \log\frac{|\Vc|}{L} }.
%    \end{equation}
%\end{cor}
%\begin{proof}
%    This follows from Theorem \ref{thm:Partial} with the function $d(V,\Lc) = \openone\{V \notin \Lc\}$ and threshold $t = 0$.
%\end{proof}

\subsection{Conditional Version} \label{sec:cond}

%In the following, in addition to the average error probability in \eqref{eq:pe}, we consider the maximal error probability:
%\begin{equation}
%    \pmax = \max_{v} \PP[\hat{V} \ne V \,|\, V=v].
%\end{equation}

When applying Fano's inequality, it is often useful to condition on certain random events and random variables.  The following theorem states a general variant of Theorem \ref{thm:Fano} with such conditioning.  Conditional forms for the case of approximate recovery (Theorem \ref{thm:Partial}) follow in an identical manner.

\begin{thm} \label{thm:Conditional}
    {\em (Conditional Fano inequality)} 
    For any discrete random variables $\fV$ and $\hat{V}$ on a common alphabet $\Vc$, any discrete random variable $A$ on an alphabet $\Ac$, and any subset $\Ac' \subseteq \Ac$, the error probability $\pe = \PP[\hat{V} \ne V]$ satisfies
    \begin{equation}
        \pe \ge \sum_{a \in \Ac'} \PP[A = a] \frac{H(V|\hat{V},A = a) - \log 2}{ \log\big(|\Vc_a| - 1\big) }, \label{eq:Conditional}
    \end{equation}
    where $\Vc_a = \{ v \in \Vc \,:\, \PP[V = v \,|\, A = a] > 0 \}$.  For possibly continuous $A$, the same holds true with $\sum_{a \in \Ac'} \PP[A = a] (\,\cdots)$ replaced by $\EE[ \openone\{ A \in \Ac '\} (\,\cdots)]$.
\end{thm}
\begin{proof}
    We write $\pe \ge \sum_{a \in \Ac'} \PP[A = a] \PP[\hat{V} \ne V \,|\, A = a]$, and lower bound the conditional error probability using Fano's inequality ({\em cf.}, Theorem \ref{thm:Fano}) under the joint distribution of $(V,\hat{V})$ conditioned on $A = a$.  
\end{proof}

\begin{rem} \label{rem:cond_uses}
    Our main use of Theorem \ref{thm:Conditional} will be to average over the input $\Xv$ ({\em cf.}, Figure \ref{fig:Reduction}) in the case that it is random and independent of $V$.  In such cases, by setting $A = \Xv$ in \eqref{eq:Conditional} and letting $\Ac'$ contain all possible outcomes, we simply recover Theorem \ref{thm:Fano} with conditioning on $\Xv$ in the conditional entropy and mutual information terms.  The approximate recovery version, Theorem \ref{thm:Partial}, extends in the same way.
    In Section \ref{sec:apps_discrete}, we will discuss more advanced applications of Theorem \ref{thm:Conditional}, including (i) genie arguments, in which some information about $V$ is revealed to the decoder, and (ii) typicality arguments, where we condition on $V$ falling in some high-probability set.
\end{rem}

\section{Mutual Information Bounds} \label{sec:mi_bounds}

We saw in  Section \ref{sec:fano} that the mutual information $I(V;\hat{V})$ naturally arises from Fano's inequality when $V$ is uniform.  More generally, we have $H(V|\hat{V}) = H(V) - I(V;\hat{V})$, so we can characterize the conditional entropy by characterizing both the entropy and the mutual information.  In this section, we provide some of the main useful tools for upper bounding the mutual information.  For brevity, we omit the proofs of standard results commonly found in information theory textbooks, or simple variations thereof. 

% We have seen that mutual information terms such as arise in all of the variants in for which $V$ is uniform.  In fact, the mutual information still plays a crucial role in the non-uniform variants written in terms of $H(V|\hat{V})$, since by definition we have $H(V|\hat{V}) = H(V) - I(V;\hat{V})$.  One of the main steps in any application of Fano's inequality is bounding this mutual information.  For instance, if the estimate $\hat{V}$ is based on $n$ independent samples from some distribution, and each such sample provides at most at most $C$ nats of information about $V$, then we can upper bound $I(V;\hat{V}) \le nC$, and re-arranging \eqref{eq:Fano2} immediately gives the lower bound $n \ge \frac{1}{C} \big( (1-\pe) \log |\Vc| - \log 2 \big)$ on the required number of samples we must gather to recover $V$ with probability at least $1-\pe$.  In this section, we provide a variety of general results that provide mutual information bounds of this type, and of other types.

Throughout the section, the random variables $V$ and $\hat{V}$ are assumed to be discrete, whereas the other random variables involved, including the inputs $\Xv = (X_1,\dotsc,X_n)$ and samples $\Yv = (Y_1,\dotsc,Y_n)$, may be continuous.  Hence, notation such as $P_Y(y)$ may represent either a probability mass function (PMF) or a probability density function (PDF).
%  In such cases, entropy quantities such as $H(Y_i)$ should be interpreted as being the {\em differential entropy} \cite[Ch.~8]{Cov01}.  % In contrast to the entropy, when it comes to mutual information there is no need to distinguish between the discrete and continuous cases.

\subsection{Data Processing Inequality}

Recall the random variables $V$, $\Xv$, $\Yv$, and $\hat{V}$ in the multiple hypothesis testing reduction depicted in Figure \ref{fig:Reduction}.  In nearly all cases, the first step in bounding a mutual information term such as $I(V;\hat{V})$ is to upper bound it in terms of the samples $\Yv$, and possibly the inputs $\Xv$.  By doing so, we remove the dependence on $\hat{V}$, and form a bound that is {algorithm-independent}.

The following lemma provides three variations along these lines.  The three are all essentially equivalent, but are written separately since each will be more naturally suited to certain settings, as described below.  Recall the terminology that $X \to Y \to Z$ {\em forms a Markov chain} if $X$ and $Z$ are conditionally independent given $Y$, or equivalently, $Z$ depends on $(X,Y)$ only through $Y$.

\begin{lem} \label{lem:dpi}
    {\em (Data processing inequality)}  
    
    \noindent (i) If $V \to \Yv \to \hat{V}$ forms a Markov chain, then $I(V;\hat{V}) \le I(V;\Yv)$.
    
    \noindent (ii) If $V \to \Yv \to \hat{V}$ forms a Markov chain conditioned on $\Xv$, then $I(V;\hat{V}|\Xv) \le I(V;\Yv|\Xv)$.
    
    \noindent (iii) If $V \to (\Xv,\Yv) \to \hat{V}$ forms a Markov chain, then $I(V;\hat{V}) \le I(V;\Xv,\Yv)$.
\end{lem}
% \noindent The proof is given in Appendix

\noindent We will use the first part when $\Xv$ is absent or deterministic, the second part for random non-adaptive $\Xv$, and the third when the elements of $\Xv$ can be chosen adaptively based on the past samples ({\em cf.} Section \ref{sec:overview}).

% In the literature, the terminology {\em Fano's inequality} often refers to the weakened version of Theorem \ref{thm:Fano} obtained by applying $I(V;\hat{V}) \le I(V;\Yv)$.  However, we prefer to state the data processing step separately to highlight the three different ways (at least) that it can be applied.

\subsection{Tensorization}

One of the most useful properties of mutual information is {\em tensorization}: Under suitable conditional independence assumptions, mutual information terms containing length-$n$ sequences (e.g., $\Yv = (Y_1,\dotsc,Y_n)$) can be upper bounded by a sum of $n$ mutual information terms, the $i$-th of which contains the corresponding entry of each associated vector (e.g., $Y_i$).  Thus, we can reduce a complicated mutual information term containing sequences to a sum of simpler terms containing individual elements.  The following lemma provides some of the most common scenarios in which such tensorization can be performed.

\begin{lem} \label{lem:mi_tensorization}
    {\em (Tensorization of mutual information)}
    (i) If the entries of $\Yv = (Y_1,\dotsc,Y_n)$ are conditionally independent given $\fV$, then
    \begin{equation}
        I(\fV;\Yv) \le \sum_{i=1}^n I(V; Y_i).
    \end{equation}
    (ii) If the entries of $\Yv$ are conditionally independent given $(V,\Xv)$, and $Y_i$ depends on $(V,\Xv)$ only through $(V,X_i)$, then
    \begin{equation}
        I(\fV; \Yv | \Xv) \le \sum_{i=1}^n I(V ; Y_i | X_i).
    \end{equation}
    (iii) If, in addition to the assumptions in part (ii), $Y_i$ depends on $(V,X_i)$ only through $U_i = \psi_i(V,X_i)$ for some deterministic function $\psi_i$, then
    \begin{equation}
        I(\fV; \Yv | \Xv) \le \sum_{i=1}^n I(U_i ; Y_i).
    \end{equation}
\end{lem}
\noindent The proof is based on the sub-additivity of entropy, along with the conditional independence assumptions given.  We will use the first part of the lemma when $\Xv$ is absent or deterministic, and the second and third parts for random non-adaptive $\Xv$.  When $\Xv$ can be chosen adaptively based on the past samples ({\em cf.} Section \ref{sec:overview}), the following variant is used.
%\begin{proof}
%
%\end{proof}
% \noindent The proof is given in Appendix \ref{sec:pf_tens}.

% In the second part of Lemma \ref{lem:dpi} (which combines naturally with the second part of Lemma \ref{lem:mi_tensorization}), we assume that $V \to \Yv \to \hat{V}$ conditioned on the {\em entire sequence} $\Xv$.  This condition holds in cases that $\Xv$ is fully specified {\em in advance}, i.e., before collecting the samples $\Yv$.  In many applications of interest, one has the additional freedom of {\em adaptivity}: The $i$-th element $X_i$ of $\Xv$ can be designed {\em based on the history} $X_1^{i-1} = (X_1,\dotsc,X_{i-1})$ and $Y_1^{i-1} = (Y_1,\dotsc,Y_{i-1})$.  In such cases, Lemma \ref{lem:mi_tensorization} no longer applies directly, but one has the following analogous result.

\begin{lem} \label{lem:mi_adaptive}
    {\em (Tensorization of mutual information for adaptive settings)}
    (i) If $X_i$ is a function of $(X_1^{i-1},Y_1^{i-1})$, and $Y_i$ is conditionally independent of $(X_1^{i-1},Y_1^{i-1})$ given $(V,X_i)$, then
    \begin{equation}
         I(V;\Xv,\Yv) \le  \sum_{i=1}^n I(\fV; Y_i | X_i). \label{eq:mi_adaptive}
    \end{equation}
    (ii)  If, in addition to the assumptions in part (i), $Y_i$ depends on $(V,X_i)$ only through $U_i = \psi_i(V,X_i)$ for some deterministic function $\psi_i$, then
    \begin{equation}
        I(V;\Xv,\Yv) \le \sum_{i=1}^n I(U_i;Y_i).
    \end{equation}
\end{lem}
\noindent The proof is based on the chain rule for mutual information, i.e., $I(V;\Xv,\Yv) = \sum_{i=1}^n I(X_i, Y_i; \fV \,|\, X_1^{i-1}, Y_1^{i-1})$, as well as suitable simplifications via the conditional independence assumptions.
%\begin{proof}
%    
%\end{proof}
% \noindent The proof is given in Appendix \ref{sec:pf_adaptive}.

\begin{rem}
    The mutual information bounds in Lemma \ref{lem:mi_adaptive} are analogous to those used in the problem of communication with feedback \cite[Sec.~7.12]{Cov01}.  A key difference is that in the latter setting, the channel input $X_i$ is a function of $(V,X_1^{i-1},Y_1^{i-1})$, with $V$ representing the message.  In statistical estimation problems, the quantity $V$ being estimated is typically unknown to the decision-maker,  so the input $X_i$ is only a function of $(X_1^{i-1},Y_1^{i-1})$
\end{rem}

\begin{rem}
    Lemma \ref{lem:mi_adaptive} should be applied with care, since even if $V$ is uniform on some set {\em a priori}, it may not be uniform conditioned on $X_i$.  This is because in the adaptive setting, $X_i$ depends on $Y_1^{i-1}$, which in turn depends on $V$.
\end{rem}

\subsection{KL Divergence Based Bounds}

By definition, the mutual information is the KL divergence between the joint distribution and the product of marginals, $I(V;Y) = D(P_{VY} \| P_V \times P_Y)$, and can equivalently be viewed as a conditional divergence $I(V;Y) = D(P_{Y|V} \| P_Y | P_V)$.  Viewing the mutual information in this way leads to a variety of useful bounds in terms of related KL divergence quantities, as the following lemma shows.

\begin{lem} \label{lem:mi_kl}
    {\em (KL divergence based bounds)}
    Let $P_{V}$, $P_{Y}$, and $P_{Y|V}$ be the marginal distributions corresponding to a pair $(V,Y)$, where $V$ is discrete.  For any auxiliary distribution $Q_{Y}$, we have
    \begin{align}
        I(V;Y) 
            &= \sum_{v} P_{V}(v) D\big( P_{Y|V}(\cdot\,|\,v) \,\big\|\, P_{Y}\big) \label{eq:Aux1} \\
            & \le \sum_{v} P_{V}(v) D\big( P_{Y|V}(\cdot\,|\,v) \,\big\|\, Q_{Y}\big) \label{eq:Aux2} \\
            &\le \max_{v} D\big( P_{Y|V}(\cdot\,|\,v) \,\big\|\, Q_{Y}\big), \label{eq:Aux3}
    \end{align}
    and in addition,
    \begin{align}
        I(V;Y) 
            &\le \sum_{v,v'} P_{V}(v) P_{V}(v')  D\big( P_{Y|V}(\cdot\,|\,v) \,\big\|\, P_{Y}(\cdot\,|\,v') \big) \label{eq:TwoThetas1} \\ 
            &\le \max_{v,v'} D\big( P_{Y|V}(\cdot\,|\,v) \,\big\|\, P_{Y|V}(\cdot\,|\,v') \big). \label{eq:TwoThetas2}
    \end{align}
\end{lem}
\begin{proof}
    We obtain \eqref{eq:Aux1} from the definition of mutual information, and \eqref{eq:Aux2} from the fact that $\EE\big[ \log\frac{P_{Y|V}(Y|V)}{P_{Y}(Y)}\big] = \EE\big[ \log\frac{P_{Y|V}(Y|V)}{Q_{Y}(Y)}\big] - \EE\big[ \log\frac{P_{Y}(Y)}{Q_{Y}(Y)}\big]$; the second term here is a KL divergence, and is therefore non-negative.  We obtain \eqref{eq:TwoThetas1} from \eqref{eq:Aux2} by noting that $Q_{Y}$ can be chosen to be any of the $P_{Y}(\cdot\,|\,v')$, and the remaining inequalities \eqref{eq:Aux3} and \eqref{eq:TwoThetas2} are trivial.
\end{proof}
% \noindent The proof is given in Appendix \ref{sec:pf_kl}.

\noindent The upper bounds in \eqref{eq:Aux2}--\eqref{eq:TwoThetas2} are closely related, and often essentially equivalent in the sense that they lead to very similar converse bounds.  In the authors' experience, it is usually slightly simpler to choose a suitable auxiliary distribution $Q_Y$ and apply \eqref{eq:Aux3}, rather than bounding the pairwise divergences as in \eqref{eq:TwoThetas2}.  Examples will be given in Sections \ref{sec:apps_discrete} and \ref{sec:apps_cont}.

\begin{rem} \label{rem:kl_bounds}
    We have used the generic notation $Y$ in Lemma \ref{lem:mi_kl}, but in applications this may represent either the entire vector $\Yv$, or a single one of its entries $Y_i$.  Hence, the lemma may be used to bound $I(V;\Yv)$ directly, or one may first apply tensorization and then use the lemma to bound each $I(V;Y_i)$.
\end{rem}

\begin{rem} \label{rem:kl_bounds_cond}
    Lemma \ref{lem:mi_kl} can also be used to bound {\em conditional} mutual information terms such as $I(V;Y|X)$.  Conditioned on any $X = x$, we can upper bound $I(V;Y|X=x)$ using Lemma \ref{lem:mi_kl}, with an auxiliary distribution $Q_{Y|X=x}$ that may depend on $x$.  For instance, doing this for \eqref{eq:Aux3} and then averaging over $X$, we obtain for any $Q_{Y|X}$ that
        \begin{align}
            I(V;Y|X) 
                &\le \max_{v} D\big( P_{Y|X,V}(\cdot\,|\,\cdot,v) \,\big\|\, Q_{Y|X} | P_{X}\big) \\
                &\le \max_{x,v} D\big( P_{Y|X,V}(\cdot\,|\,x,v) \,\big\|\, Q_{Y|X}(\cdot\,|\,x) \big).
        \end{align}
\end{rem}

The bound \eqref{eq:Aux3} in Lemma \ref{lem:mi_kl} is useful when there exists a single auxiliary distribution $Q_Y$ that is ``close'' to each $P_{Y|V}(\cdot|v)$ in KL divergence, i.e., $D\big( P_{Y|V}(\cdot\,|\,v) \,\big\|\, Q_{Y}\big)$ is small.   It is natural to extend this idea by introducing {multiple} auxiliary distributions, and only requiring that {any one of them} is close to a given $P_{Y|V}(\cdot|v)$.  This can be viewed as ``covering''  the conditional distributions $\{P_{Y|V}(\cdot|v)\}_{v\in\Vc}$ with ``KL divergence balls'', and we will return to this viewpoint in Section \ref{sec:global_local}.

\begin{lem} \label{lem:Covering}
    {\em (Mutual information bound via covering)} 
    Under the setup of Lemma \ref{lem:mi_kl}, suppose there exist $N$ distributions $Q_1(y),\dotsc,Q_N(y)$ such that for all $v$ and some $\epsilon > 0$, it holds that
    \begin{equation}
        \min_{j=1,\dotsc,N} D\big( P_{Y|V}(\cdot\,|\,v) \,\big\|\, Q_{j}\big) \le \epsilon. \label{eq:CoveringCond}
    \end{equation}
    Then we have
    \begin{equation}
        I(V;Y) \le \log N + \epsilon. \label{eq:CoveringBound}
    \end{equation}
\end{lem}
%\begin{proof}
%
%\end{proof}
\noindent The proof is based on applying \eqref{eq:Aux2} with $Q_{Y}(y) = \frac{1}{N}\sum_{j=1}^N Q_j(y)$, and then lower bounding this summation over $j$ by the value $j^*(v)$ achieving the minimum in \eqref{eq:CoveringCond}. %  The details are left as an exercise.
% The proof is given in Appendix \ref{sec:pf_covering}.
We observe that setting $N = 1$ in Lemma \ref{lem:Covering} simply yields \eqref{eq:Aux3}.  % Moreover, similarly to Remark \ref{rem:kl_bounds}, we can apply the lemma with the vector of samples $\Yv$ playing the role of $Y$.

% Typical applications of Lemma \ref{lem:Covering} directly bound the entire sequence of observations $I(V;\Yv)$, rather than a single observation $I(V;Y_i)$.

\subsection{Relations Between KL Divergence and Other Measures}

As evidenced above, the KL divergence plays a crucial role in applications of Fano's inequality.  In some cases, directly characterizing the KL divergence can still be difficult, and it is more convenient to bound it in terms of other divergences or distances.  The following lemma gives a few simple examples of such relations; the reader is referred to \cite{Sas16} for a more thorough treatment.

\begin{lem} \label{lem:relations}
    {\em (Relations between divergence measures)}
    Fix two distributions $P$ and $Q$, and consider the KL divergence $D(P\|Q) = \EE_P\big[ \log \frac{P(Y)}{Q(Y)}\big]$, total variation $\dTV(P,Q) = \frac{1}{2}\EE_Q\big[ \big| \frac{P(Y)}{Q(Y)} - 1 \big| \big]$, squared Hellinger distance $H^2(P,Q) = \EE_Q\big[ \big(\sqrt{\frac{P(Y)}{Q(Y)}} - 1\big)^2\big]$, and  $\chi^2$-divergence $\chi^2(P\|Q) = \EE_Q\big[ \big(\frac{P(Y)}{Q(Y)}-1\big)^2\big]$. We have:
    \begin{itemize}
        \item {\em (KL vs.~TV)} $D(P\|Q) \ge 2 \dTV(P,Q)^2$, whereas if  $P$ and $Q$ are probability mass functions and each entry of $Q$ is at least $\eta > 0$, then $D(P\|Q) \le \frac{2}{\eta} \dTV(P,Q)^2$.
        \item {\em (Hellinger vs.~TV)} $\frac{1}{2}H^2(P,Q) \le \dTV(P,Q) \le H(P,Q)\sqrt{1 - \frac{H^2(P,Q)}{4}}$;
        \item {\em (KL vs.~$\chi^2$)} $D(P\|Q) \le \log(1 + \chi^2(P\|Q)) \le \chi^2(P\|Q)$.
    \end{itemize}
\end{lem}

\section{Applications -- Discrete Settings} \label{sec:apps_discrete}

In this section, we provide two examples of statistical estimation problems in which the quantity being estimated is discrete: group testing and graphical model selection.  Our goal is not to treat these problems comprehensively, but rather, to study particular instances that permit a simple analysis while still illustrating the key ideas and tools introduced in the previous sections.   We consider the {\em high-dimensional} setting, in which the underlying number of parameters being estimated is much higher than the number of measurements.  To simplify the final results, we will often write them using the asymptotic notation $o(1)$ for asymptotically vanishing terms, but non-asymptotic variants are easily inferred from the proofs.

\subsection{Group Testing}

The group testing problem consists of determining a small subset of ``defective'' items within a larger set of items based on a number of pooled tests.  A given test contains some subset of the items, and the binary test outcome indicates, possibly in a noisy manner, whether or not {\em at least one} defective item was included in the test.  This problem has a history in medical testing \cite{Dor43}, and has regained significant attention following applications in communication protocols, pattern matching, database systems, and more.

In more detail, the setup is described as follows:
\begin{itemize}
    \item In a population of $p$ items, there are $k$ unknown {\em defective items}.  This defective set is denoted by $S \subseteq \{1,\dotsc,p\}$, and is assumed to be uniform on the set of ${p \choose k}$ subsets having cardinality $k$.  Hence, in this example, we are in the Bayesian setting with a uniform prior.  We focus on the {sparse} setting, in which $k \ll p$, i.e., defective items are rare.
    \item There are $n$ tests specified by a {\em test matrix} $\Xv \in \{0,1\}^{n \times p}$: The $(i,j)$-th entry of $\Xv$, denoted by $X_{ij}$, indicates whether item $j$ is included in test $i$.  We initially consider the {\em non-adaptive} setting, where $\Xv$ is chosen in advance.  We allow for this choice to be random; for instance, a common choice of random design is to let the entries of $\Xv$ be i.i.d.~Bernoulli random variables.
    \item To account for possible noise, we consider the following observation model:
    \begin{equation}
        Y_i = \bigg( \bigvee_{j \in S} X_{ij} \bigg) \oplus Z_i, \label{eq:gt_noisy}
    \end{equation}
    where $Z_i \sim \Bernoulli(\epsilon)$ for some $\epsilon \in \big[0,\frac{1}{2}\big)$, $\oplus$ denotes modulo-2 addition, and $\vee$ is the ``OR'' operation.  In the channel coding terminology, this corresponds to passing the noiseless test outcome $\bigvee_{j \in S} X_{ij}$ through a {binary symmetric channel}.   We assume that the noise variables $Z_i$ are independent of each other and of $\Xv$, and we define the vector of test outcomes $\Yv = (Y_1,\dotsc,Y_n)$.
    \item Given $\Xv$ and $\Yv$, a {decoder} forms an estimate $\hat{S}$ of $S$.  We initially consider the {exact recovery} criterion, in which the error probability is given by
    \begin{equation}
        \pe = \PP[\hat{S} \ne S], \label{eq:gt_exact_criterion}
    \end{equation}
    where the probability with respect to $S$, $\Xv$, and $\Yv$.
\end{itemize}
In the following subsections, we present several results and analysis techniques that are primarily drawn from \cite{Mal78,Ati12}.

\subsubsection{Exact Recovery with Non-Adaptive Testing}

Under the exact recovery criterion \eqref{eq:gt_exact_criterion}, we have the following lower bound on the required number of tests.  Recall that $H_2(\alpha) = \alpha\log\frac{1}{\alpha} + (1-\alpha)\log\frac{1}{1-\alpha}$ denotes the binary entropy function.

\begin{thm} \label{thm:gt_exact}
    {\em (Group testing with exact recovery)}
    Under the preceding noisy group testing setup, in order to achieve $\pe \le \delta$, it is necessary that
    \begin{equation}
        n \ge \frac{ k\log\frac{p}{k} }{ \log 2 - H_2(\epsilon) } (1 - \delta - o(1)) \label{eq:gt_exact}
        % n \ge \max\bigg\{ \frac{ \log{p \choose k} }{ \log 2 - H_2(\epsilon) }, \frac{ k\log (p -k + 1) }{ (1-2\epsilon)\log\frac{1-\epsilon}{\epsilon} } \bigg\} (1 - \delta - o(1)). \label{eq:gt_exact}
    \end{equation}
    as $p \to \infty$, possibly with $k \to \infty $ simultaneously.
\end{thm}

%We can translate this condition into a lower bound on $n$ by the second part of Lemma applying the first part of Lemma \ref{lem:mi_tensorization}.  Specifically, since $S \to (\Xv,\Yv) \to \hat{S}$ forms a Markov chain, applying the lemma with $(\Xv,\Yv)$ in place of $\Yv$ gives 
%\begin{align}
%    I(S;\hat{S}) 
%        \le \sum_{i=1}^n I(V;X_i,Y_i) \label{eq:gt_mi_ub1} \\
%        = n I(V;X,Y) \label{eq:gt_mi_ub2} \\
%        = n I(V;Y | X), \label{eq:gt_mi_ub3}
%\end{align}
%where \eqref{eq:gt_mi_ub1} holds with $(X,Y) = (X_i,Y_i)$ for an arbitrary fixed $i$ since the pairs $(X_i,Y_i)$ are identically distributed, and \eqref{eq:gt_mi_ub3} holds since $X$ is independent of $V$.  Substituting \eqref{eq:gt_mi_ub3} into \eqref{eq:gt_exact_mi} gives
%\begin{equation}
%    n \ge \frac{ (1-\delta)\log|\Sc| - \log 2 }{ I(S;Y|X) }. \label{eq:gt_exact_mi2}
%\end{equation}
\begin{proof}
    Since $S$ is discrete-valued, we can use the trivial reduction to multiple hypothesis testing with $V = S$.
    Applying Fano's inequality ({\em cf.}, Theorem \ref{thm:Fano}) with conditioning on $\Xv$ ({\em cf.}, Section \ref{sec:cond}), we obtain
    \begin{equation}
        I(S;\Yv|\Xv) \ge (1-\delta) \log{p \choose k} - \log 2, \label{eq:gt_exact_mi}
    \end{equation}
    where we have also upper bounded $I(S;\hat{S}|\Xv) \le I(S;\Yv|\Xv)$ using the data processing inequality ({\em cf.}, second part of Lemma \ref{lem:dpi}), which in turn uses the fact that $S \to \Yv \to \hat{S}$ conditioned on $\Xv$.
    
    Let $U_i = \bigvee_{j \in S} X_{ij}$ denote the hypothetical noiseless outcome.  Since the noise variables $\{Z_i\}_{i=1}^n$ are independent and $Y_i$ depends on $(S,\Xv)$ only through $U_i$ ({\em cf.}, \eqref{eq:gt_noisy}), we can apply tensorization ({\em cf.}, third part of Lemma \ref{lem:mi_tensorization}) to obtain
    \begin{align}
        I(S;\Yv|\Xv)
            &\le \sum_{i=1}^n I(U_i; Y_i) \label{eq:gt_exact1_mi1} \\
            &\le n\big(\log 2 - H_2(\epsilon)\big), \label{eq:gt_exact1_mi2}
    \end{align}
    where \eqref{eq:gt_exact1_mi2} follows since $Y_i$ is generated from $U_i$ according to a binary symmetric channel, which has capacity $\log 2 - H_2(\epsilon)$.  Substituting \eqref{eq:gt_exact1_mi2} and ${p \choose k} \ge \big(\frac{p}{k}\big)^k$ into \eqref{eq:gt_exact_mi} and rearranging, we obtain \eqref{eq:gt_exact}.
\end{proof}

Theorem \ref{thm:gt_exact} is known to be tight in terms of scaling laws whenever $\delta \in (0,1)$ is fixed and $k = o(p)$, and perhaps more interestingly, {tight including constant factors} as $\delta \to 0$ under the scaling $k = O(p^{\theta})$ for sufficiently small $\theta > 0$.  The matching achievability result in this regime can be proved using maximum-likelihood decoding \cite{Sca15b}.  However, achieving such a result using a computationally efficient decoder remains a challenging open problem.

\subsubsection{Approximate Recovery with Non-Adaptive Testing}

We now move to an approximate recovery criterion:  The decoder outputs a list $\Lc \subseteq \{1,\dotsc,p\}$ of cardinality $L \ge k$, and we require that at least a fraction $(1-\alpha )k$ of the defective items appear in the list, for some $\alpha  \in (0,1)$.  It follows that the error probability can be written as
\begin{equation}
    \pe(t) = \PP[d(S,\Lc) > t],
\end{equation}
where $d(S,\Lc) = |S \backslash \Lc|$, and $t = \alpha k$.  Notice that a higher value of $L$ means more non-defective items may be included in the list, whereas a higher value of $\alpha$ means more defective items may be absent.

\begin{thm} \label{thm:gt_partial}
    {\em (Group testing with approximate recovery)}
    Under the preceding noisy group testing setup with list size $L \ge k$, in order to achieve $\pe(\alpha k) \le \delta$ for some $\alpha  \in (0,1)$ (not depending on $p$), it is necessary that
    \begin{equation}
        n \ge  \frac{ (1-\alpha ) k \log\frac{p}{L} }{ \log 2 - H_2(\epsilon) } (1 - \delta - o(1)) \label{eq:gt_partial}
    \end{equation}
    as $p \to \infty$, $k \to \infty$ and $L \to \infty$ simultaneously with $L = o(p)$.
\end{thm}
\begin{proof}
    We apply the approximate recovery version of Fano's inequality ({\em cf.}, Theorem \ref{thm:Partial}) with $d(S,\Lc) = |S \backslash \Lc|$ and $t = \alpha k$ as above.  For any $\Lc$ with cardinality $L$, the number of $S$ with $d(S,\Lc) \le \alpha k$ is given by
    $\Nmax(t) = \sum_{j=0}^{\lfloor \alpha k \rfloor} {p-L \choose j}{L \choose k-j}$,
    which follows by counting the number of ways to place $k-j$ defective items in $\Lc$, and the remaining $j$ defective items in the other $p - L$ entries.  Hence, using Theorem \ref{thm:Partial} with conditioning on $\Xv$ ({\em cf.}, Section \ref{sec:cond}), and applying the data processing inequality ({\em cf.}, second part of Lemma \ref{lem:dpi}), we obtain
    \begin{equation}
        I(S;\Yv|\Xv) \ge (1-\delta) \log\frac{{p \choose k}}{\sum_{j=0}^{\lfloor \alpha k \rfloor} {p-L \choose j}{L \choose k -j}} - \log 2. \label{eq:gt_partial_init}
    \end{equation}
    By upper bounding the summation by $\lfloor \alpha k \rfloor +1$ times the maximum value, and performing some asymptotic simplifications via the assumption $L = o(p)$, we can simplify the logarithm to $\big(k \log \frac{p}{L}\big)(1+o(1))$ \cite{Sca17}.  The theorem is then established by upper bounding the conditional mutual information using \eqref{eq:gt_exact1_mi2}.
\end{proof}

Theorem \ref{thm:gt_partial} matches Theorem \ref{thm:gt_exact} up to the factor of $1-\alpha$ and the replacement of $\log\frac{p}{k}$ by $\log\frac{p}{L}$, suggesting that approximate recovery provides a minimal reduction in the number of tests even for moderate values of $\alpha$ and $L$.  However, under approximate recovery, a near-matching achievability bound is known under the scaling $k = O(p^{\theta})$ {\em for all $\theta \in (0,1)$}, rather than only sufficiently small $\theta$ \cite{Sca15b}.

% is close to that in Theorem \ref{thm:gt_exact}, particularly when $\alpha$ is small and $L$ is not too large.  In particular, in the above-mentioned scaling regimes of the form $k = O(p^{\theta})$ where Theorem \ref{thm:gt_partial} is tight, we deduce that if $L = O(k)$, then {approximate recovery provides an asymptotic gain of at most a factor $1-\alpha$}.  However, approximate recovery has been shown to help significantly more in certain regimes with higher $k$ \cite{Sca17}.

% In particular, when $\Xv$ is restricted to have i.i.d.~Bernoulli entries, scaling regimes are known for which $n = \big(k\log_2\frac{p}{k}\big) (1+o(1))$ tests are sufficient for approximate recovery with $L = k$ and any fixed $\alpha \in (0,1)$, but insufficient for exact recovery 

\subsubsection{Adaptive Testing}

Next, we discuss the {adaptive testing} setting, in which a given input vector $X_i \in \{0,1\}^p$, corresponding to a single row of $\Xv$, is allowed to depend on the previous inputs and outcomes, i.e., $X_1^{i-1} = (X_1,\dotsc,X_{i-1})$ and $Y_1^{i-1} = (Y_1,\dotsc,Y_{i-1})$.  In fact, it turns out that Theorems \ref{thm:gt_exact} and \ref{thm:gt_partial} still apply in this setting.  Establishing this simply requires making the following modifications to the above analysis:
\begin{itemize}
    \item Apply the data processing inequality in the form of the {\em third} part of Lemma \ref{lem:dpi}, yielding \eqref{eq:gt_exact_mi} and \eqref{eq:gt_partial_init} with $I(S;\Xv,\Yv)$ in place of $I(S;\Yv|\Xv)$;
    \item Apply tensorization via Lemma \ref{lem:mi_adaptive} to deduce \eqref{eq:gt_exact1_mi1}--\eqref{eq:gt_exact1_mi2} with $I(S;\Xv,\Yv)$ in place of $I(S;\Yv|\Xv)$.
\end{itemize}
In the regimes where Theorems \ref{thm:gt_exact} and/or \ref{thm:gt_partial} are known to have matching upper bounds with non-adaptive designs, we can clearly deduce that {adaptivity provides no asymptotic gain}.  However, as with approximate recovery, adaptivity can significantly broaden the conditions under which matching achievability bounds are known, at least in the noiseless setting \cite{Bal13}.

% However, even when adaptivity does not help in an information-theoretic sense, its added freedom can be helpful for designing {\em practical} algorithms.  For instance, for noiseless group testing (i.e., $\epsilon = 0$), efficient and practical adaptive algorithms are known that match the converse, whereas only computationally intractable algorithms have been proved to do so in the non-adaptive case.

\subsubsection{Discussion: General Noise Models}

The preceding analysis can easily be extended to more general group testing models in which the observations $(Y_1,\dotsc,Y_n)$ are conditionally independent given $\Xv$.  A broad class of such models can be written in the form $(Y_i|N_i) \sim P_{Y|N}$, where $N_i = \sum_{j \in S} \openone\{X_{ij} = 1\}$ denotes the number of defective items in the $i$-th test.  In such cases, the preceding results hold true more generally when $\log 2 - H_2(\epsilon)$ is replaced by the capacity $\max_{P_N} I(N;Y)$ of the ``channel'' $P_{Y|N}$.

For certain models, we can obtain a better lower bound by applying a {\em genie argument}, along with the conditional form of Fano's inequality in Theorem \ref{thm:Conditional}.  Fix $\ell \in \{1,\dotsc,k\}$,  and suppose that a uniformly random subset $S^{(1)} \subseteq S$ of cardinality $k - \ell$ is revealed to the decoder.  This extra information can only make the group testing problem easier, so any converse bound for this modified setting remains valid for the original setting.  Perhaps counter-intuitively, this idea can lead to a better final bound.

We only briefly outline the details of this more general analysis, and refer the interested reader to \cite{Ati12,Sca16b}.  Using Theorem \ref{thm:Conditional} with $A = S^{(1)}$, and applying the data processing inequality and tensorization, one can obtain
\begin{equation}
    \pe \ge 1 - \frac{\sum_{i=1}^n I(N_i^{(0)};Y_i|N_i^{(1)}) - \log 2}{ \log{p - k + \ell \choose \ell} }, \label{eq:gt_genie_init}
\end{equation}
where $N_i^{(1)} = \sum_{j \in S^{(1)}} \openone\{X_{ij} = 1\}$, and $N_i^{(0)} = N_i - N_i^{(1)}$.  The intuition is that we condition on $N_i^{(1)}$ since it is known via the genie, while the remaining information about $Y_i$ is determined by $N_i^{(0)}$.  Once \eqref{eq:gt_genie_init} is established, it only remains to simplify the mutual information terms; see \cite{Ati12,Sca16b} for further details.

% Upper bounding $I(N_i^{(0)};Y_i|N_i^{(1)})$ is somewhat more difficult for general $\ell$ than for the case $\ell = k$ that we already analyzed.  In the random testing scenario where the entries of $\Xv$ are i.i.d.~Bernoulli, $N_i^{(0)}$ and $N_i^{(1)}$ follow Binomial distributions, and in this case the mutual information can often be explicitly bounded via a direct calculation.  For a fixed deterministic design $\Xv$, the random variables $N_i^{(0)}$ and $N_i^{(1)}$ instead follow Hypergeometric distributions; these are somewhat trickier to handle, but can explicitly be approximated by Binomial distributions.  The interested reader is referred to \cite{Ati12,Sca16b}.

\subsection{Graphical Model Selection}

Graphical models provide compact representations of the conditional independence relations between random variables, and frequently arise in areas such as image processing, statistical physics, computational biology, and natural language processing.  The fundamental problem of \emph{graphical model selection} consists of recovering the graph structure given a number of independent samples from the underlying distribution.

Graphical model selection has been studied under several different families of joint distributions, and also several different graph classes.  We focus our attention on the commonly-used {\em Ising model} with binary observations, and on a simple graph class known as {\em forests}, defined to contain the graphs having no cycles. % An attractive feature of this setup is that it is known to permit practical learning algorithms with rigorous guarantees on the required number of samples \cite{Tan11,Ana12}.

Formally, the setup is described as follows:
\begin{itemize}
    \item We are given $n$ independent samples $Y_1,\dotsc,Y_n$ from a $p$-dimensional joint distribution: $Y_i = (Y_{i1},\dotsc,Y_{ip})$ for $i=1,\dotsc,n$.  This joint distribution is encoded by a graph $G = (V,E)$, where $V = \{1,\dotsc,p\}$ is the {\em vertex set}, and $E \subseteq V \times V$ is the {\em edge set}.  We use the terminology {\em vertex} and {\em node} interchangeably.  We assume that there are no edges from a vertex to itself, and that the edges are {\em undirected}: $(i,j) \in E$ and $(j,i) \in E$ are equivalent, and only count as one edge.
    \item We focus on the {\em Ising model}, in which the observations are binary-valued, and the joint distribution of a given sample, say $Y_1 = (Y_{11},\dotsc,Y_{1p}) \in \{-1,1\}^p$, is
    \begin{equation}
        P_{G}(y_1) = \frac{1}{Z}\exp\bigg( \lambda \sum_{(i,j)\in E} y_{1i} y_{1j} \bigg), \label{eq:Ising}
    \end{equation}
    where $Z$ is a normalizing constant. Here $\lambda > 0$ is a parameter to the distribution dictating the {edge strength}; a higher value means it is more likely that $Y_{1i} = Y_{1j}$ for any given edge $(i,j) \in E$.
    \item We restrict the graph $G = (V,E)$ to be the set of all {\em forests}:  
    \begin{equation}
        \Gforest = \big\{ G \,:\, G\text{ has no cycles} \},
    \end{equation}
    where a {\em cycle} is defined to be a path of distinct edges leading back to the start node, e.g., $(1,4),(4,2),(2,1)$.  A special case of a forest is a {\em tree}, which is an acyclic graph for which a path exists between any two nodes.  One can view any forest as being a disjoint union of trees, each defined on some subset of $V$.  See Figure \ref{fig:ForestTree} for an illustration.
    \item Let $\Yv \in \{-1,1\}^{n \times p}$ be the matrix whose $i$-th row contains the $p$ entries of the $i$-th sample.  Given $\Yv$, a decoder forms an estimate $\hat{G}$ of $G$, or equivalently, an estimate $\hat{E}$ of $E$.  We initially focus on the exact recovery criterion, in which the minimax error probability is given by
    \begin{equation}
        \Mc_n(\Gforest,\lambda) = \inf_{\hat{G}} \sup_{G \in \Gforest} \PP_G[\hat{G} \ne G],
    \end{equation}
    where $\PP_G$ denotes probability when the true graph is $G$, and the infimum is over all estimators.
\end{itemize}
To our knowledge, Fano's inequality has not been applied previously in this exact setup; we do so using the general tools for Ising models given in \cite{San12,Sha14,Sca16,Sca16d}.

\begin{figure}
    
    \begin{centering}
        % \subfloat[Example forest] {
            \includegraphics[height=0.2\columnwidth]{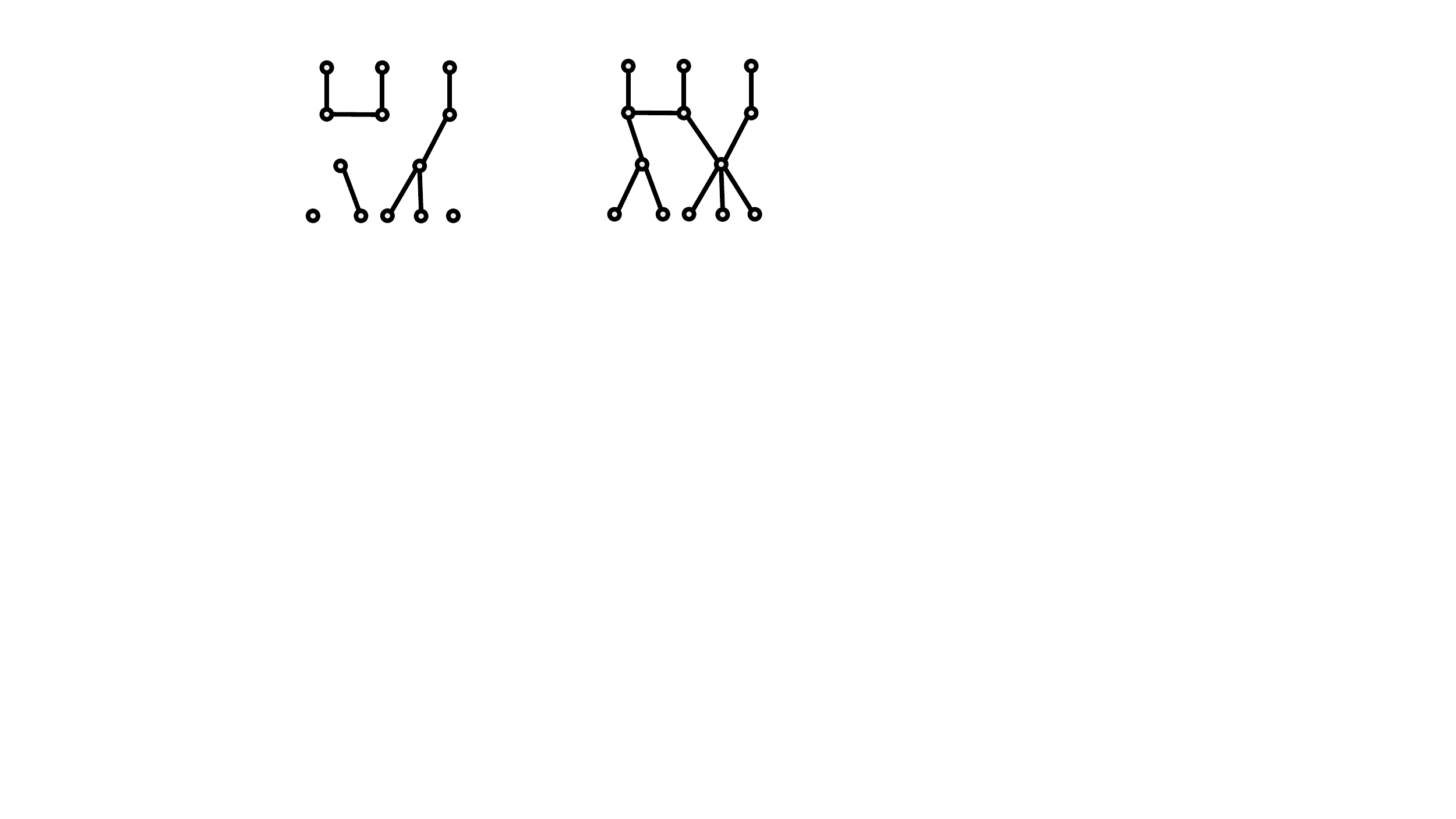}
        % }
        \qquad\qquad
       %  \subfloat[Example tree] {
            \includegraphics[height=0.2\columnwidth]{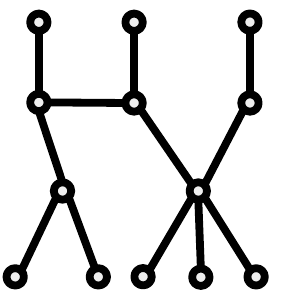}
        % }
        \par
    \end{centering}
    \caption{Two examples of graphs that are forests (i.e., acyclic graphs); the graph on the right is also a tree (i.e., a connected acyclic graph). \label{fig:ForestTree}} \vspace*{-3ex}
\end{figure}

\subsubsection{Exact Recovery}

Under the exact recovery criterion, we have the following.

\begin{thm} \label{thm:GraphExact}
    {\em (Exact recovery of forest graphical models)}
    Under the preceding Ising graphical model selection setup with a given edge parameter $\lambda > 0$, in order to achieve $\Mc_n(\Gforest,\lambda) \le \delta$, it is necessary that
    \begin{equation}
        n \ge \max \Bigg\{ \frac{\log p}{\log 2}, \frac{2 \log p}{ \lambda \tanh\lambda } \Bigg\} (1-\delta-o(1)) \label{eq:GraphExact}
    \end{equation}
    as $p \to \infty$.
\end{thm}
\begin{proof}
    Recall from Section \ref{sec:overview} that we can lower bound the worst-case error probability over $\Gforest$ by the average error probability over {any} subset of $\Gforest$.  This gives us an important degree of freedom in the reduction to multiple hypothesis testing, and corresponds to selecting a hard subset $\theta_1,\dotsc,\theta_M$ as described in Section \ref{sec:step1}.  We refer to a given subset $\Gc \subseteq \Gforest$ as a {\em graph ensemble}, and provide two choices that lead to the two terms in \eqref{eq:GraphExact}.  % The two ensembles will provide different trade-offs between the number of graphs that need to be distinguished, and how ``close'' the various joint distributions $P_G$ are from each other.
    
    For any choice of $\Gc \subseteq \Gforest$, Fano's inequality (Theorem \ref{thm:Fano}) gives
    \begin{equation}
        n \ge \frac{ (1-\delta)\log |\Gc| - \log 2  }{ I(G;Y_1) }, \label{eq:graph_init}
    \end{equation}
    for $G$ uniform on $\Gc$, where we used $I(G;\hat{G}) \le I(G;\Yv) \le nI(G;Y_1)$ by the data processing inequality and tensorization ({\em cf.} first parts of Lemmas \ref{lem:dpi} and \ref{lem:mi_tensorization}).
    
    \underline{Restricted ensemble 1}: Let $\Gc_1$ be the set of all trees.  It is well-known from graph theory that the number of trees on $p$ nodes is $|\Gc_1| = p^{p-2}$ \cite{Tan11}.  Moreover, since $Y_1$ is a length-$p$ binary sequence, we have $I(G;Y_1) \le H(Y_1) \le p \log 2$.  Hence, \eqref{eq:graph_init} yields $n \ge \frac{(1-\delta) (p-2)\log p - \log 2}{ p \log 2 }$, implying the first bound in \eqref{eq:GraphExact}.
    
    \underline{Restricted ensemble 2}: Let $\Gc_2$ be the set of graphs containing a single edge, so that $|\Gc_2| = {p \choose 2}$.  We will upper bound the mutual information using \eqref{eq:Aux3} in Lemma \ref{lem:mi_kl}, choosing the auxiliary distribution $Q_Y$ to be $P_{\Gbar}$ with $\Gbar$ being the empty graph.  Thus, we need to bound $D(P_G \| P_{\Gbar})$ for each $G \in \Gc_2$.  
    
    We first give an upper bound on $D(P_G \| P_{\Gbar})$ for {\em any} two graphs $(G,\Gbar)$.  We start with the trivial bound
    \begin{equation}
        D(P_G \| P_{\Gbar}) \le D(P_G \| P_{\Gbar}) + D(P_{\Gbar} \| P_G). \label{eq:GraphKL_init}
    \end{equation}
    Recall the definition $D(P \| Q) = \EE_{P}\big[ \log \frac{P(Y)}{Q(Y)} \big]$, and consider the substitution of $P_G$ and $P_{\Gbar}$ according to \eqref{eq:Ising}, with different normalizing constants $Z_G$ and $Z_{\Gbar}$.  We see that when we sum the two terms in \eqref{eq:GraphKL_init}, the normalizing constants inside the logarithms cancel, and we are left with 
    \begin{align}
        D(P_G \| P_{\Gbar}) \le &\sum_{(i,j) \in E \backslash \Ebar} \lambda\big( \EE_G[Y_{1i} Y_{1j}] - \EE_{\Gbar}[Y_{1i} Y_{1j}] \big) \nonumber \\ + &\sum_{(i,j) \in \Ebar \backslash E} \lambda\big( \EE_{\Gbar}[Y_{1i} Y_{1j}] - \EE_{G}[Y_{1i} Y_{1j}] \big) \label{eq:GraphKL_init2}
    \end{align}
    for $G = (V,E)$ and $\Gbar = (V,\Ebar)$.  
    
    In the case that $G$ has a single edge (i.e., $G \in \Gc_2$) and $\Gbar$ is the empty graph, we can easily compute $\EE_{\Gbar}[Y_{1i} Y_{1j}] = 0$, and \eqref{eq:GraphKL_init2} simplifies to
    \begin{equation}
        D(P_G \| P_{\Gbar}) \le \lambda \EE_G[Y_{1i} Y_{1j}], \label{eq:GraphKL_init3}
    \end{equation}
    where $(i,j)$ is the unique edge in $G$.  Since $Y_{1i}$ and $Y_{1j}$ only take values in $\{-1,1\}$, we have $\EE_G[Y_{1i} Y_{1j}] = (+1)\PP[Y_{1i} = Y_{1j}] + (-1)\PP[Y_{1i} \ne Y_{1j}] = 2\PP[Y_{1i} = Y_{1j}] - 1$, and letting $E$ have a single edge in \eqref{eq:Ising} yields $\PP_G[(Y_{1i},Y_{1j}) = (y_i,y_j)] = \frac{e^{\lambda y_i y_j}}{ 2e^{\lambda} + 2e^{-\lambda} }$, and hence $\PP_G[Y_{1i} = Y_{1j}] = \frac{e^{\lambda}}{e^{\lambda} + e^{-\lambda}}$.  Combining this with $\EE_G[Y_{1i} Y_{1j}] = 2\PP[Y_{1i} = Y_{1j}] - 1$ yields $\EE_G[Y_{1i} Y_{1j}] = \frac{2e^{\lambda}}{e^{\lambda} + e^{-\lambda}} - 1 = \tanh\lambda$.  Hence, using \eqref{eq:GraphKL_init3} along with \eqref{eq:Aux3} in Lemma \ref{lem:mi_kl}, we obtain $I(G;Y_1) \le \lambda\tanh\lambda$.  Substitution into \eqref{eq:graph_init} (with $\log |\Gc| = (2\log p) (1+o(1))$) yields the second bound in \eqref{eq:GraphExact}.
\end{proof}

Theorem \ref{thm:GraphExact} is known to be tight up to constant factors whenever $\lambda = O(1)$ \cite{Tan11,Ana12}: When $\lambda$ is constant the lower bound becomes $n = \Omega(\log p)$, whereas for asymptotically vanishing $\lambda$ it simplifies to $n = \Omega\big(\frac{1}{\lambda^2} \log p\big)$.  % However, unlike the group testing problem, it is unclear to what extent the constant factors can be improved. 

\subsubsection{Approximate Recovery}

We consider the approximate recovery of $G = (V,E)$ with respect to the {\em edit distance} $d(G,\hat{G}) = |E \backslash \hat{E}| + |\hat{E} \backslash E|$, which is the number of edge additions and removals needed to transform $G$ into $\hat{G}$ or vice versa.  Since any forest can have at most $p-1$ edges, it is natural to consider the case that an edit distance of up to $\alpha  p$ is permitted, for some $\alpha > 0$.  Hence, the minimax risk is given by
\begin{equation}
    \Mc_n(\Gforest,\lambda,\alpha ) = \inf_{\hat{G}} \sup_{G \in \Gforest} \PP_G[d(G,\hat{G}) > \alpha  p].
\end{equation}
In this setting, we have the following.

\begin{thm} \label{thm:GraphPartial}
    {\em (Approximate recovery of forest graphical models)}
    Under the preceding Ising graphical model selection setup with a given edge parameter $\lambda > 0$ and approximate recovery parameter $\alpha  \in \big(0,\frac{1}{2}\big)$ (with the latter not depending on $p$), in order to achieve $\Mc_n(\Gforest,\lambda,\alpha ) \le \delta$, it is necessary that
    \begin{equation}
        n \ge \max \Bigg\{ \frac{(1-\alpha )\log p}{\log 2}, \frac{2(1-\alpha ) \log p}{ \lambda \tanh\lambda }\Bigg\} (1-\delta-o(1)) \label{eq:GraphPartial}
    \end{equation}
    as $p \to \infty$.
\end{thm}
\begin{proof}
    For any $\Gc \subseteq \Gforest$, Theorem \ref{thm:Partial} provides the following analog of \eqref{eq:graph_init}:
    \begin{equation}
        n \ge \frac{ (1-\delta)\log\frac{|\Gc|}{\Nmax(\alpha p)} - \log 2  }{ I(G;Y_1) } \label{eq:graph_init_P}
    \end{equation}
    for $G$ uniform on $\Gc$, where $\Nmax(t) = \max_{\hat{G}} \sum_{G \in \Gc} \openone\{ d(G,\hat{G}) \le t\}$ implicitly depends on $\Gc$.  We again consider two restricted ensembles; the first is identical to the exact recovery setting, whereas the second is modified due to the fact that learning single-edge graphs with approximate recovery is trivial.
    
    \underline{Restricted ensemble 1}: Once again, let $\Gc_1$ be the set of all trees.  We have already established $|\Gc_1| = (p-2)\log p$ and $I(G;Y_1) \le n \log 2$ for this ensemble, so it only remains to characterize $\Nmax(\alpha p)$.  
    
    While the decoder may output a graph $\hat{G}$ not lying in $\Gc_1$, we can assume without loss of generality that $\hat{G}$ is always selected such that $d(\hat{G},G^*) \le \alpha p$ for some $G^* \in \Gc_1$; otherwise, an error would be guaranteed.  As a result, for any $\hat{G}$, and any $G \in \Gc_1$ such that $d(G,\hat{G}) \le \alpha p$, we have from the triangle inequality that $d(G,G^*) \le d(G,\hat{G}) + d(\hat{G},G^*) \le 2\alpha p$, which implies that 
    \begin{equation}
        \Nmax(\alpha p) \le \sum_{G \in \Gc_1} \openone\{ d(G,G^*) \le 2\alpha p\}. \label{eq:GraphNmax_init1}
    \end{equation}
    Now observe that since all graphs in $\Gc_1$ have exactly $p-1$ edges, transforming $G$ to $G^*$ requires removing $j$ edges and adding $j$ different edges, for some $j \le \alpha p$.  Hence, we have 
    \begin{equation}
        \Nmax(\alpha p) \le \sum_{j=0}^{\lfloor \alpha p \rfloor} {p-1 \choose j} { {p \choose 2} - p + 1 \choose j }. \label{eq:GraphNmax_init2}
    \end{equation}
    By upper bounding the summation by $\lfloor \alpha p \rfloor + 1$ times the maximum, and performing some asymptotic simplifications, we can show that $\log \Nmax(\alpha p) \le \big(\alpha p \log p\big) (1+o(1))$.  Substituting into \eqref{eq:graph_init_P} and recalling that $|\Gc_1| = (p-2)\log p$ and $I(G;Y_1) \le p\log 2$, we obtain the first bound in \eqref{eq:GraphPartial}.
    
    \underline{Restricted ensemble 2a}: Let $\Gc_{2a}$ be the set of all graphs on $p$ nodes containing exactly $\frac{p}{2}$ isolated edges; if $p$ is an odd number, the same analysis applies with an arbitrary single node ignored.  We proceed by characterizing $|\Gc_{2a}|$, $I(G;Y_1)$, and $\Nmax(\alpha p)$.  The number of graphs in the ensemble is $|\Gc_{2a}| = {p \choose 2}{p -2 \choose 2}\cdots{4 \choose 2}{2 \choose 2} = \frac{p!}{2^{p/2}}$, and Stirling's approximation yields $\log |\Gc_{2a}| \ge \big( p \log p\big) (1+o(1))$.
    
    % The number of graphs in the ensemble is $|\Gc_{2a}| = {p \choose 2}{p -2 \choose 2}\cdots{4 \choose 2}{2 \choose 2}$.  For any $\epsilon > 0$, this is lower bounded by ${\epsilon p \choose 2}^{\frac{(1-\epsilon)p}{2}}$, and straightforward asymptotic simplifications yield $\log |\Gc_{2a}| \ge \big((1-\epsilon) p \log p\big) (1+o(1))$ for any fixed $\epsilon> 0$.  Since $\epsilon$ may be arbitrarily small, we conclude that $\log |\Gc_{2a}| \ge \big( p \log p\big) (1+o(1))$.
    
    Since the KL divergence is additive for product distributions, and we established in the exact recovery case that the KL divergence between the distributions of a single-edge graph and an empty graph is at most $\lambda \tanh\lambda$, we deduce that $D(P_G \| P_{\Gbar}) \le \frac{p}{2} \lambda \tanh\lambda$ for any $G \in \Gc_{2a}$, where $\Gbar$ is the empty graph.  We therefore obtain from Lemma \ref{lem:mi_kl} that $I(G;Y_1) \le \frac{p}{2} \lambda \tanh\lambda$.
    
    A similar argument to that of Ensemble 1 yields $\Nmax(\alpha p) \le \sum_{j=0}^{ \lfloor \alpha p \rfloor} {\frac{p}{2} \choose j} { {p \choose 2} - \frac{p}{2} \choose j }$, in analogy with \eqref{eq:GraphNmax_init2}.  This again simplifies to $\Nmax(\alpha p) \le \big(\alpha p \log p\big) (1+o(1))$, and having established $\log |\Gc_{2a}| \ge \big( p \log p\big) (1+o(1))$ and  $I(G;Y_1) \le \frac{p}{2} \lambda \tanh\lambda$, substitution into  \eqref{eq:graph_init_P} yields the second bound in \eqref{eq:GraphPartial}. 
\end{proof}

The bound in Theorem \ref{thm:GraphPartial} matches that of Theorem \ref{thm:GraphExact} up to a multiplicative factor of $1-\alpha$, thus suggesting that approximate recovery does not significantly help in reducing the required number of samples, at least in the minimax sense, for the Ising model and forest graph class.  % Note that although we reached this conclusion for both examples in this section, this should not be taken as an indication that it is something one should expect in general.  For instance, see \cite{Ree11,Ree12} for a sparse recovery problem in which considering approximate recovery is crucial.

\subsubsection{Adaptive Sampling} \label{sec:GraphAdaptive}

We now return to the exact recovery setting, and consider a modification in which we have an added degree of freedom in the form of {\em adaptive sampling}:
\begin{itemize}
    \item The algorithm proceeds in rounds; in round $i$, the algorithm queries a {subset} of the $p$ nodes indexed by $X_i \in \{0,1\}^p$, and the corresponding sample $Y_i$ is generated as follows:
    \begin{itemize}
        \item The joint distribution of the entries of $Y_i$, corresponding to the entries where $X_i$ is one, coincides with the corresponding marginal distribution of $P_G$, with independence between rounds;
        \item The values of the entries of $Y_i$, corresponding to the entries where $X_i$ is zero, are given by $\missing$, a symbol indicating that the node was not observed.
    \end{itemize}    
    We allow $X_i$ to be selected based on the past queries and samples, namely, $X_1^{i-1} = (X_1,\dotsc,X_{i-1})$ and $Y_1^{i-1} = (Y_1,\dotsc,Y_{i-1})$.
    \item Let $n(X_i)$ denote the number of ones in $X_i$, i.e., the number of nodes observed in round $i$.  While we allow the total number of rounds to vary, we restrict the algorithm to output an estimate $\hat{G}$ after observing at most $\nnode$ nodes.  This quantity is related to $n$ in the non-adaptive setting according to $\nnode = np$, since in the non-adaptive setting we always observe all $p$ nodes in each sample.
    \item The minimax risk is given by
    \begin{equation}
        \Mc_{\nnode}(\Gforest,\lambda) = \inf_{\hat{G}} \sup_{G \in \Gforest} \PP_G[\hat{G} \ne G],
    \end{equation}
    where the infimum is over all adaptive algorithms that observe at most $\nnode$ nodes in total.
\end{itemize}

\begin{thm} \label{thm:GraphAdaptive}
    {\em (Adaptive sampling for forest graphical models)}
    Under the preceding Ising graphical model selection problem with adaptive sampling and a given parameter $\lambda > 0$, in order to achieve $\Mc_{\nnode}(\Gforest,\lambda) \le \delta$, it is necessary that
    \begin{equation}
        \nnode \ge \max \Bigg\{ \frac{p\log p}{\log 2}, \frac{2 p\log p}{ \lambda \tanh\lambda } \Bigg\} (1-\delta-o(1)) \label{eq:GraphAdaptive}
    \end{equation}
    as $p \to \infty$.
\end{thm}
\begin{proof}
    We prove the result using Ensemble 1 and Ensemble 2a above.  We let $N$ denote the number of rounds; while this quantity is allowed to vary, we can assume without loss of generality that $N = \nnode$ by adding or removing rounds where no nodes are queried.  For any subset $\Gc \subseteq \Gforest$, applying Fano's inequality ({\em cf.}, Theorem \ref{thm:Fano}) and tensorization  ({\em cf.}, first part of Theorem \ref{lem:mi_adaptive}) yields
    \begin{equation}
        \sum_{i=1}^N I(G;Y_i | X_i) \ge (1-\delta)\log |\Gc| - \log 2, \label{eq:graph_init_A}
    \end{equation}
    where $G$ is uniform on $\Gc$.
    
    \underline{Restricted ensemble 1}: We again let $\Gc_1$ be the set of all trees, for which we know that $|\Gc| = p^{p-2}$.  Since the $n(X_i)$ entries of $Y_i$ differing from $\missing$ are binary, and those equaling $\missing$ are deterministic given $X_i$, we have $I(G;Y_i | X_i = x_i) \le n(x_i) \log 2$.  Averaging over $X_i$ and summing over $i$ yields $\sum_{i=1}^N I(G;Y_i | X_i) \le \sum_{i=1}^N \EE[n(X_i)] \log 2 \le \nnode \log 2$, and substitution into \eqref{eq:graph_init_A} yields the first bound in \eqref{eq:GraphAdaptive}.
    
    \underline{Restricted ensemble 2a}: We again use the above-defined ensemble $\Gc_{2a}$ of graphs with $\frac{p}{2}$ isolated edges, for which we know that $|\Gc_{2a}| \ge \big( p \log p \big)(1+o(1))$.  In this case, when we observe $n(X_i)$ nodes, the sub-graph corresponding to these observed nodes has at most $\frac{n(X_i)}{2}$ edges, all of which are isolated.  Hence, using Lemma \ref{lem:mi_kl}, the above-established fact that the KL divergence from a single-edge graph to the empty graph is at most $\lambda \tanh \lambda$, and the additivity of KL divergence for product distributions, we deduce that $I(G;Y_i | X_i = x_i) \le \frac{n(x_i)}{2} \lambda\tanh\lambda$.  Averaging over $X_i$ and summing over $i$ yields $\sum_{i=1}^N I(G;Y_i | X_i) \le \frac{1}{2}\nnode \lambda\tanh\lambda$, and substitution into \eqref{eq:graph_init_A} yields the second bound in \eqref{eq:GraphAdaptive}.
\end{proof}

The threshold in Theorem \ref{thm:GraphAdaptive} matches that of Theorem \ref{thm:GraphExact}, and in fact, a similar analysis under approximate recovery also recovers the threshold in Theorem \ref{thm:GraphPartial}.  This suggests that adaptivity is of limited help in the minimax sense for the Ising model and forest graph class.  There are, however, other instances of graphical model selection where adaptivity provably helps \cite{Das16,Sca16d}.

\subsubsection{Discussion: Other Graph Classes}

\underline{Degree and edge constraints}: While the class $\Gforest$ is a relatively easy class to handle, similar techniques have also been used for more difficult classes, notably including those that place restrictions on the maximal degree $d$ and/or the number of edges $k$.  Ensembles 2 and 2a above can again be used, and the resulting bounds are tight in certain scaling regimes where $\lambda \to 0$, but loose in other regimes due to their lack of dependence on $d$ and $k$.  To obtain bounds with such a dependence, alternative ensembles have been proposed consisting of sub-graphs with highly correlated nodes \cite{San12,Sha14,Sca16}.  

For instance, suppose that a group of $d+1$ nodes has all possible edges connected except one.  Unless $d$ or the edge strength $\lambda$ are small, the high connectivity makes the nodes very highly correlated, and the sub-graph is {difficult to distinguish from a fully-connected sub-graph}.  This is in contrast with Ensembles 2 and 2a above, whose graphs are difficult to distinguish from the {empty graph}.

%\begin{figure}
%  
%    \begin{centering}
%        \subfloat[Exact recovery] {
%            \includegraphics[height=0.15\columnwidth]{figs/Ensemble1.pdf}
%        }
%        \qquad
%        \subfloat[Approximate recovery] {
%            \includegraphics[height=0.15\columnwidth]{figs/Ensemble2.pdf}
%        }
%        \par
%    \end{centering}
%    
%    \caption{Example building blocks for exact and approximate recovery of Ising models with edge or degree constraints.  For instance, in the exact recovery example, the edges present may make the two right-most edges highly correlated, thus making it difficult to determine whether the additional dashed edge is present. \label{fig:BuildingBlocks}} \vspace*{-3ex}
%\end{figure}

\underline{Bayesian setting}: Beyond minimax estimation, it is also of interest to understand the fundamental limits of random graphs.  A particularly prominent example is the Erd\"os-R\'enyi random graph, in which each edge is independently included with some probability $q \in (0,1)$.  This is a case where the conditional form of Fano's inequality has proved useful; specifically, one can apply Theorem \ref{thm:Conditional} with $A = G$, and $\Ac$ equal to the following {\em typical set} of graphs:
\begin{equation}
    \Tc = \bigg\{ G \,:\, (1-\epsilon)q{p \choose 2} \le |E| \le (1 + \epsilon)q{p \choose 2}\bigg\},
\end{equation}
where $\epsilon > 0$ is a constant.  Standard properties of typical sets \cite{Cov01} yield that $\PP[\GER \in \Tc] \to 1$, $|\Tc| = e^{\big(H_2(q){p \choose 2}\big)(1+O(\epsilon))}$, and $H(V|V\in\Tc) = \big(H_2(q){p \choose 2}\big)(1+O(\epsilon))$ whenever $q{p \choose 2} \to \infty$, and once these facts are established, Theorem \ref{thm:Conditional} yields the following following necessary condition for $\pe \le \delta$:
\begin{equation}
    n \ge \frac{p H_2(q)}{ 2 \log 2 } (1-\delta - o(1)). \label{eq:ErdosRenyi}
\end{equation}
For instance, in the case that $q = O\big(\frac{1}{p}\big)$ (i.e., there are $O(p)$ edges on average), we have $H_2(q) = \Theta\big(\frac{\log p}{p}\big)$, and we find that $n = \Omega(\log p)$ samples are necessary. This scaling is tight when $\lambda$ is constant \cite{Ana12}, whereas improved bounds for other scalings can be found in \cite{Sha14}.

% While the bound in \eqref{eq:ErdosRenyi} provides the correct scaling laws in some cases, it is generally somewhat crude in the sense that it does not depend on the edge strength $\lambda$.  Refined bounds with such a dependence have also been obtained for certain scalings of $q$ using the same typicality argument as the first step \cite{Sha14}, albeit with much more involved subsequent steps.

\section{From Discrete to Continuous} \label{sec:continuous}

Thus far, we have focused on using Fano's inequality to provide converse bounds for the estimation of discrete quantities.  In many, if not most, statistical applications, one is instead interested in estimating {continuous} quantities; examples include linear regression, covariance estimation, density estimation, and so on.  It turns out that the discrete form of Fano's inequality is still broadly applicable in such settings.  The idea, as outlined in Section \ref{sec:intro}, is to choose a {finite subset} that still captures the inherent difficulty in the problem.  In this section, we present several tools used for this purpose.  % We first continue with estimation problems, and then move on to optimization problems.

\subsection{Minimax Estimation Setup} \label{sec:minimax_est}

Recall the setup described in Section \ref{sec:overview}: A parameter $\theta$ is known to lie in some subset $\Theta$ of a continuous domain (e.g., $\RR^p$), the samples $\Yv = (Y_1,\dotsc,Y_n)$ are drawn from a joint distribution $\Ptn(\yv)$, an estimate $\thetahat$ is formed, and the loss incurred is $\ell(\theta,\thetahat)$.  For clarity of exposition, we focus primarily on the case that there is no input, i.e., $\Xv$ in Figure \ref{fig:Reduction} is absent or deterministic.  However, the main results ({\em cf.}, Theorems \ref{thm:Packing} and \ref{thm:PackingPartial} below) extend to settings with inputs as described in Section \ref{sec:overview}; the mutual information $I(V;\Yv)$ is replaced by $I(V;\Yv|\Xv)$ in the non-adaptive setting, or $I(V;\Xv,\Yv)$ in the adaptive setting.

In continuous settings, the reduction to multiple hypothesis testing ({\em cf.}, Figure \ref{fig:Reduction}) requires that the loss function is sufficiently well-behaved.  We focus on a widely-considered class of functions that can be written as
\begin{equation}
    \ell(\theta,\thetahat) = \Phi\big( \rho(\theta,\thetahat) \big), \label{eq:ellPhiRho}
\end{equation}
where $\rho(\theta,\theta')$ is a metric, and $\Phi(\cdot)$ is an increasing function from $\RR_+$ to $\RR_+$.  For instance, the squared-$\ell_2$ loss $\ell(\theta,\theta') = \|\theta - \theta'\|_2^2$ clearly takes this form.
    
We focus on the minimax setting, defining the {\em minimax risk} as follows:
\begin{equation}
    \Mc_n(\Theta,\ell) = \inf_{\thetahat} \sup_{\theta \in \Theta} \EE_{\theta}\big[ \ell(\theta,\thetahat) \big],
\end{equation}
where the infimum is over all estimators $\thetahat = \thetahat(\Yv)$, and $\EE_{\theta}$ denotes expectation when the underlying parameter is $\theta$.  We subsequently define $\PP_{\theta}$ analogously.

% As will be exemplified in Section \ref{sec:apps_cont}, this general setting captures a wide range of problems including regression and density estimation.

\subsection{Reduction to the Discrete Case}

We present two related approaches to reducing the continuous estimation problem to a discrete one.  The first, based on the standard form of Fano's inequality in Theorem \ref{thm:Fano}, was discovered much earlier \cite{Ibr77}, and accordingly, it has been used in a much wider range of applications.  However, the second approach, based on the approximate recovery version of Fano's inequality in Theorem \ref{thm:Partial}, has recently been shown to provide added flexibility in the reduction \cite{Duc13}.

\subsubsection{Reduction with Exact Recovery}

As we discussed in Section \ref{sec:intro}, we seek to reduce the continuous problem to multiple hypothesis testing in such a way that successful minimax estimation implies success in the hypothesis test with high probability.  To this end, we choose a {\em hard subset} $\theta_1,\dotsc,\theta_M$, for which the elements are sufficiently well-separated so that the index $v \in \{1,\dotsc,M\}$ can be identified from the estimate $\thetahat$ ({\em cf.}, Figure \ref{fig:Reduction}).  This is formalized in the proof of the following result.

%\begin{defn} \label{def:packing}
%    Given a set $\Theta$ and an associated function $\rho(\theta,\theta')$, a finite subset $\{\theta_1,\dotsc,\theta_M\} \subseteq \Theta$ is said to be an {\em $\epsilon$-packing of $\Theta$ with respect to $\rho$} if
%\end{defn}

% With this definition in place, a general minimax lower bound based on the standard form of Fano's inequality is given as follows.

\begin{thm} \label{thm:Packing}
    {\em (Minimax bound via reduction to exact recovery)}
    Under the preceding minimax estimation setup, fix $\epsilon > 0$, and let $\{\theta_1,\dotsc,\theta_M\}$ be a finite subset of $\Theta$ such that   
    \begin{equation}
        \rho(\theta_v,\theta_{v'}) \ge \epsilon, \quad \forall v,v' \in \{1,\dotsc,M\}, v \ne v'. \label{eq:PackingAssump}
    \end{equation}
    Then, we have
    \begin{equation}
        \Mc_n(\Theta,\ell) \ge \Phi\Big(\frac{\epsilon}{2}\Big) \bigg( 1 - \frac{I(V;\Yv) + \log 2}{ \log M }\bigg), \label{eq:PackingResult}
    \end{equation}
    where $V$ is uniform on $\{1,\dotsc,M\}$, and the mutual information is with respect to $V \to \theta_V \to \Yv$.  Moreover, in the special case $M = 2$, we have
    \begin{equation}
        \Mc_n(\Theta,\ell) \ge \Phi\Big(\frac{\epsilon}{2}\Big) H_2^{-1}\big( \log 2 - I(V;\Yv) \big), \label{eq:PackingResultM2}
    \end{equation}
    where $H_2^{-1}(\cdot) \in [0,0.5]$ is the inverse binary entropy function.
\end{thm}
\begin{proof}
    As illustrated in Figure \ref{fig:Reduction}, the idea is to reduce the estimation problem to a multiple hypothesis testing problem.  As an initial step, we note from Markov's inequality that, for any $\epsilon_0 > 0$,
    \begin{align}
        \sup_{\theta \in \Theta} \EE_{\theta}\big[ \ell(\theta,\thetahat) \big] 
            &\ge \sup_{\theta \in \Theta} \Phi(\epsilon_0) \PP_{\theta}[ \ell(\theta,\thetahat) \ge \Phi(\epsilon_0)] \\
            &= \Phi(\epsilon_0) \sup_{\theta \in \Theta} \PP_{\theta}[ \rho(\theta,\thetahat) \ge \epsilon_0], \label{eq:ExcessDist}
    \end{align}
    where \eqref{eq:ExcessDist} uses \eqref{eq:ellPhiRho} and the assumption that $\Phi(\cdot)$ is increasing.  % Hence, the average loss is lower bounded in terms of the probability of $\rho(\theta,\thetahat)$ exceeding $\epsilon_0$.
    
    Suppose that a random index $V$ is drawn uniformly from $\{1,\dotsc,M\}$, the samples $\Yv$ are drawn from the distribution $\Ptn$ corresponding to $\theta = \theta_V$, and the estimator is applied to produce $\thetahat$.  Let $\hat{V}$ correspond to the closest $\theta_j$ according to the metric $\rho$, i.e., $\hat{V} = \argmin_{v=1,\dotsc,M} \rho(\theta_v,\thetahat)$.  Using the triangle inequality and the assumption \eqref{eq:PackingAssump}, if $\rho(\theta_v,\thetahat) < \frac{\epsilon}{2}$ then we must have $\hat{V} = v$; hence, 
    \begin{equation}
        \PP_v\bigg[ \rho(\theta_v,\thetahat) \ge \frac{\epsilon}{2} \bigg] \ge \PP_v[\hat{V} \ne v], \label{eq:ExcessDistProb}
    \end{equation}
    where $\PP_v$ is a shorthand for $\PP_{\theta_v}$.
    
    With the above tools in place, we proceed as follows:
    \begin{align}
        \sup_{\theta \in \Theta} \PP_{\theta}\bigg[ \rho(\theta,\thetahat) \ge \frac{\epsilon}{2}\bigg]
            &\ge \max_{v=1,\dotsc,M} \PP_{v}\bigg[ \rho(\theta_v,\thetahat) \ge \frac{\epsilon}{2}\bigg] \label{eq:MinimaxEnd1} \\
            &\ge \max_{v=1,\dotsc,M} \PP_{v}[\hat{V} \ne v] \label{eq:MinimaxEnd2}  \\
            &\ge \frac{1}{M} \sum_{v=1,\dotsc,M} \PP_{v}[\hat{V} \ne v] \label{eq:MinimaxEnd3}  \\
            &\ge 1 - \frac{ I(V;\Yv) + \log 2 }{ \log M }, \label{eq:MinimaxEnd4} 
    \end{align}
    where \eqref{eq:MinimaxEnd1} follows by maximizing over a smaller set, \eqref{eq:MinimaxEnd2} follows from \eqref{eq:ExcessDistProb}, \eqref{eq:MinimaxEnd3} lower bounds the maximum by the average, and \eqref{eq:MinimaxEnd4} follows from Fano's inequality (\emph{cf.}, Theorem \ref{thm:Fano}) and the fact that $I(V;\hat{V}) \le I(V;\Yv)$ by the data processing inequality (\emph{cf.}, Lemma \ref{lem:dpi}).
    
    The proof of \eqref{eq:PackingResult} is concluded by substituting \eqref{eq:MinimaxEnd4} into \eqref{eq:ExcessDist} with $\epsilon_0 = \frac{\epsilon}{2}$, and taking the infimum over all estimators $\thetahat$.  For $M = 2$, we obtain \eqref{eq:PackingResultM2} in the same way upon replacing \eqref{eq:MinimaxEnd4} by the version of Fano's inequality for $M = 2$ given in Remark \ref{rem:FanoWeaken}.
\end{proof}

We return to this result in Section \ref{sec:global_local}, where we introduce and compare some of the most widely-used approaches to choosing the set $\{\theta_1,\dotsc,\theta_M\}$ and bounding the mutual information.

%\begin{rem}
%    The statement and proof of Theorem \ref{thm:Packing} extend to the more general case that $\rho$ only satisfies a weaker version of the triangle inequality: $\rho(\theta_1,\theta_2) \le C\big( \rho(\theta_1,\theta_0) + \rho(\theta_0,\theta_2) \big)$ for some $C \ge 1$.  In this case, $\frac{\epsilon}{2}$ in \eqref{eq:PackingResult} should be replaced by $\frac{\epsilon}{2C}$ \cite{Yan99}.  % As argued in \cite{Yan99}, the arguments can also be adapted to the case that this generalized triangle inequality only holds when $\min\{\rho(\theta_1,\theta_0), \rho(\theta_0,\theta_2)\}$ is sufficiently small.
%\end{rem}

\subsubsection{Reduction with Approximate Recovery}

The following generalization of Theorem \ref{thm:Packing}, based on Fano's inequality with approximate recovery ({\em cf.}, Theorem \ref{thm:Partial}), provides added flexibility in the reduction. An example comparing the two approaches will be given in Section \ref{sec:apps_cont} for the sparse linear regression problem.

\begin{thm} \label{thm:PackingPartial}
    {\em (Minimax bound via reduction to approximate recovery)}
    Under the preceding minimax estimation setup, fix $\epsilon > 0$, $t \in \RR$, a finite set $\Vc$ of cardinality $M$, and an arbitrary real-valued function $d(v,v')$ on $\Vc \times \Vc$, and let $\{\theta_v\}_{v \in \Vc}$ be a finite subset of $\Theta$ such that
    \begin{equation}
        d(v,v') > t \implies \rho(\theta_v,\theta_{v'}) \ge \epsilon, \quad \forall v,v' \in \Vc. \label{eq:PackingAssumpPartial}
    \end{equation}
    Then we have for any $\epsilon \ge 0$ that
    \begin{equation}
        \Mc_n(\Theta,\ell) \ge \Phi\Big(\frac{\epsilon}{2}\Big) \bigg( 1 - \frac{I(V;\Yv) + \log 2}{ \log \frac{M}{\Nmax(t)} }\bigg), \label{eq:PackingResultApprox}
    \end{equation}
    where $V$ is uniform on $\{1,\dotsc,M\}$, the mutual information is with respect to $V \to \theta_V \to \Yv$, and $\Nmax(t) = \max_{v' \in \Vc} \sum_{v \in \Vc} \openone\{ d(v,v') \le t \}$.
\end{thm}
%\begin{proof}
%
%\end{proof}
\noindent The proof is analogous to that of Theorem \ref{thm:Packing}, and can be found in \cite{Duc13}.

\subsection{Local vs.~Global Approaches} \label{sec:global_local}

Here we highlight two distinct approaches to applying the reduction to exact recovery as per Theorem \ref{thm:Packing}, termed the {\em local} and {\em global} approaches.  We do not make such a distinction for the approximate recovery variant in Theorem \ref{thm:PackingPartial}, since we are not aware of a global approach being used previously for this variant.

{\bf Local approach.} The most common approach to applying Theorem \ref{thm:Packing} is to construct a set $\{\theta_1,\dotsc,\theta_M\}$ of elements that are ``close'' in KL divergence.  Specifically, upper bounding the mutual information via Lemma \ref{lem:mi_kl} (with the vector $\Yv$ playing the role of $Y$ therein), one can weaken \eqref{eq:PackingResult} as follows.

\begin{cor} \label{cor:local}
    {\em (Local approach to minimax estimation)}
    Under the setup of Theorem \ref{thm:Packing} with a given set $\{\theta_1,\dotsc,\theta_M\}$ satisfying \eqref{eq:PackingAssump}, it holds for any auxiliary distribution $Q^n(\yv)$ that
    \begin{equation}
        \Mc_n(\Theta,\ell) \ge \Phi\Big(\frac{\epsilon}{2}\Big) \bigg( 1 - \frac{\min_{v=1,\dotsc,M} D(\Ptvn \| Q^n) + \log 2}{ \log M }\bigg). \label{eq:local}
    \end{equation}
    Moreover, the same bound holds true when $\min_{v} D(\Ptvn \| Q^n)$ is replaced by any of $\frac{1}{M}\sum_{v} D(\Ptvn \| Q^n)$, $\frac{1}{M^2} \sum_{v,v'} D(\Ptvn \| \Ptvvn)$, or $\max_{v,v'} D(\Ptvn \| \Ptvvn)$.
\end{cor}

Attaining a good bound in \eqref{eq:local} requires choosing $\{\theta_1,\dotsc,\theta_M\}$ to trade off two competing objectives: (i) A larger value of $M$ means that more hypotheses need to be distinguished; and (ii) A smaller value of $\min_{v} D(\Ptvn \| Q^n)$ means that the hypotheses are more similar.  Generally speaking, there is no single best approach to optimizing this trade-off, and the size and structure of the set can vary significantly from problem to problem.  Moreover, the construction need not be explicit; one can instead use probabilistic arguments to prove the existence of a set satisfying the desired properties.  Examples are given in Section \ref{sec:apps_cont}.  Naturally, an analog of Corollary \ref{cor:local} holds for $M = 2$ as per Theorem \ref{thm:Packing}, and a counterpart for approximate recovery holds as per Theorem \ref{thm:PackingPartial}.

We briefly mention that Corollary \ref{cor:local} has interesting connections with the popular Assouad method from the statistics literature, as detailed in \cite{Yu97}.  In addition, the counterpart of Corollary \ref{cor:local} with $M = 2$ (using \eqref{eq:FanoM2b} in its proof) is similarly related to an analogous technique known as Le Cam's method.

{\bf Global approach.} An alternative approach to applying Theorem \ref{thm:Packing} is the {\em global} approach, which performs the following: (i) Construct a subset of $\Theta$ with as many elements as possible subject to the assumption \eqref{eq:PackingAssump}; (ii) Construct a set that {\em covers} $\Theta$, in the sense of Lemma \ref{lem:Covering}, with as few elements as possible.  The following definitions formalize the notions of forming ``as many'' and ``as few'' elements as possible.  We write these in terms of a general real-valued function $\rho_0(\theta,\theta')$ that need not be a metric.

\begin{defn} \label{def:PackingNumber}
    A set $\{\theta_1,\dotsc,\theta_M\} \subseteq \Theta$ is said to be an {\em $\epsp$-packing set of $\Theta$ with respect to a measure $\rho_0 \,:\,\Theta \times \Theta \to \RR$} if $\rho_0(\theta_v,\theta_{v'}) \ge \epsp$ for all $v,v' \in \{1,\dotsc,M\}$ with $v' \ne v$.  The {\em $\epsp$-packing number} $M_{\rho_0}^*(\Theta,\epsp)$ is defined to be the maximum cardinality of any $\epsp$-packing. 
\end{defn}

\begin{defn} \label{def:CoveringNumber}
    A set $\{\theta_1,\dotsc,\theta_N\} \subseteq \Theta$ is said to be an {\em $\epsc$-covering set of $\Theta$ with respect to $\rho_0 \,:\,\Theta \times \Theta \to \RR$} if, for any $\theta \in \Theta$, there exists some $v \in \{1,\dotsc,N\}$ such that $\rho_0(\theta,\theta_v) \le \epsc$.  The {\em $\epsc$-covering number} $N_{\rho_0}^*(\Theta,\epsc)$ is defined to be the minimum cardinality of any $\epsc$-covering.
\end{defn}

Observe that assumption \eqref{eq:PackingAssump} of Theorem \ref{thm:Packing} precisely states that $\{\theta_1,\dotsc,\theta_M\}$ is an $\epsilon$-packing set, though the result is often applied with $M$ far smaller than the $\epsilon$-packing number.  The logarithm of the covering number is often referred to as the {\em metric entropy}.

\begin{figure}
    \begin{centering}
        \includegraphics[width=0.7\columnwidth]{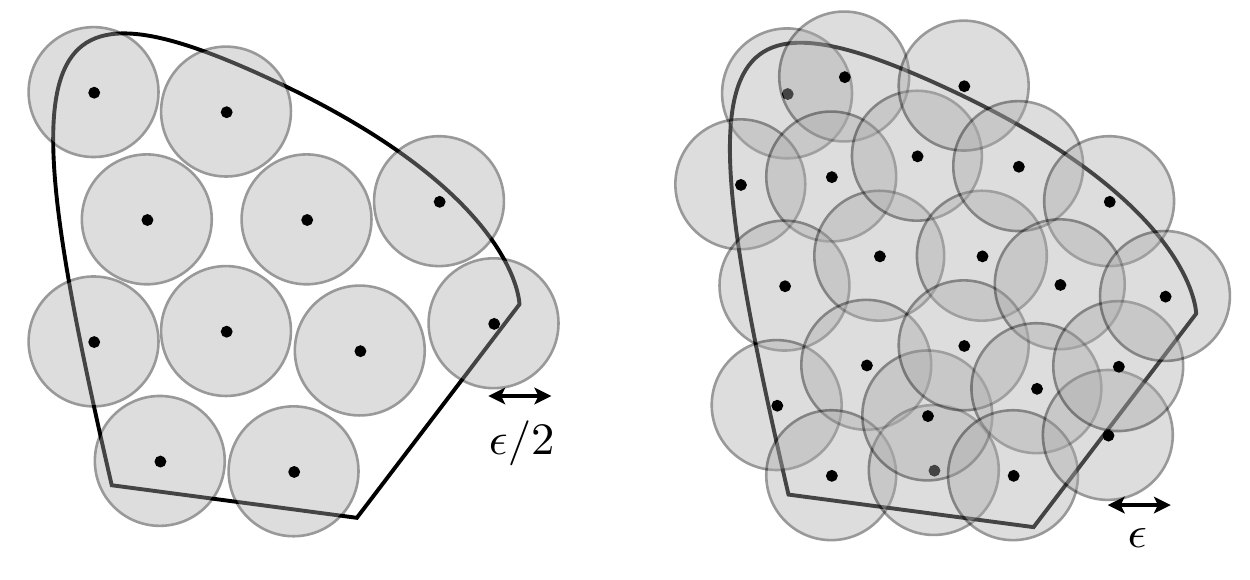}
        \par
    \end{centering}
    
    \caption{Examples of $\epsilon$-packing (Left) and $\epsilon$-covering (Right) sets in the case that $\rho_0$ is the Euclidean distance in $\RR^2$.  Since $\rho_0$ is a metric, a set of points is an $\epsilon$-packing if and only if their corresponding $\frac{\epsilon}{2}$-balls do not intersect. \label{fig:PackingCovering}}
\end{figure}

The notions of packing and covering are illustrated in Figure \ref{fig:PackingCovering}. We do not explore the properties of packing and covering numbers in detail in this chapter; the interested reader is referred to \cite{DuchiNotes,WuNotes} for a more detailed treatment.  We briefly state the following useful property, showing that the two definitions are closely related in the case that $\rho_0$ is a metric.

\begin{lem} {\em (Packing vs.~covering numbers)} \label{lem:pack_cov}
    If $\rho_0$ is a metric, then $M_{\rho_0}^*(\Theta,2\epsilon) \le N_{\rho_0}^*(\Theta,\epsilon) \le M_{\rho_0}^*(\Theta,\epsilon)$.
\end{lem}

We now show how to use Theorem \ref{thm:Packing} to construct a lower bound on the minimax risk in terms of certain packing and covering numbers.  For the packing number, we will directly consider the metric $\rho$ used in Theorem \ref{thm:Packing}.  On the other hand, for the covering number, we consider the density $\Ptvn(\yv)$ associated with each $\theta \in \Theta$, and use the associated KL divergence measure:
\begin{equation}
    \NKLn^*(\Theta,\epsilon) = N_{\rhoKLn}^*(\Theta,\epsilon), \quad \rhoKLn(\theta,\theta') = D(P_{\theta}^n \| P_{\theta'}^n).
\end{equation}

\begin{cor} \label{cor:global}
    {\em (Global approach to minimax estimation)}
    Under the minimax estimation setup of Section \ref{sec:minimax_est}, we have for any $\epsp > 0$ and $\epscn > 0$ that
    \begin{equation}
        \Mc_n(\Theta,\ell) \ge \Phi\Big(\frac{\epsp}{2}\Big) \bigg( 1 - \frac{\log \NKLn^*( \Theta, \epscn ) + \epscn + \log 2}{ \log M_\rho^*(\Theta,\epsp) }\bigg). \label{eq:PackingCovering}
    \end{equation}
    In particular, if $\Ptn(\yv)$ is the $n$-fold product of some single-measurement distribution $P_{\theta}(y)$ for each $\theta \in \Theta$, then we have for any $\epsp > 0$ and $\epsc > 0$ that
    \begin{equation}
        \Mc_n(\Theta,\ell) \ge \Phi\Big(\frac{\epsp}{2}\Big) \bigg( 1 - \frac{\log \NKL^*( \Theta, \epsc ) + n\epsc + \log 2}{ \log M_\rho^*(\Theta,\epsp) }\bigg), \label{eq:PackingCovering2}
    \end{equation}
    where $\NKL^*(\Theta,\epsilon) = N_{\rhoKL}^*(\Theta,\epsilon)$ with $\rhoKL(\theta,\theta') = D(P_{\theta} \| P_{\theta'})$.
\end{cor}
\begin{proof}
    Since Theorem \ref{thm:Packing} holds for any packing set, it holds for the maximal packing set.  Moreover, using Lemma \ref{lem:Covering}, we have $I(V;\Yv) \le \log \NKLn^*( \Theta, \epscn ) + \epscn$ in \eqref{eq:PackingResult}, since covering the entire space $\Theta$ is certainly enough to cover the elements in the packing set.  Combining these, we obtain the first part of the corollary.  The second part follows directly from the first part by choosing $\epscn = n\epsc$ and noting that the KL divergence is additive for product distributions.
\end{proof}

Corollary \ref{cor:global} has been used as the starting point to derive minimax lower bounds for a wide range of problems \cite{Yan99}; see Section \ref{sec:apps_cont} for an example.
% often using bounds on the packing and covering numbers known from function approximation theory.  After substituting such bounds, obtaining a minimax lower bound amounts to balancing $\epsp$ and $\epsc$, i.e., choosing them to maximize the lower bound on the minimax risk.
It has been observed that the global approach is mainly useful for infinite-dimensional problems such as density estimation and non-parametric regression, with the local approach typically being superior for finite-dimensional problems such as vector or matrix estimation.

\subsection{Beyond Estimation -- Fano's Inequality for Optimization} \label{sec:opt}

While the minimax estimation framework captures a diverse range of problems of interest, there are also interesting problems that it does not capture.  A notable example, which we consider in this section, is {\em stochastic optimization}. We provide a brief treatment, and refer the reader to \cite{Rag11} for further details and results.
% Although this problem does not directly fall under the estimation framework considered thus far, similar ideas can still be applied in order to obtain converse bounds.  

We consider the following setup:
\begin{itemize}
    \item We seek to minimize an unknown function $f \,:\,\Xc \to \RR$ on some input domain $\Xc$, i.e., to find a point $x \in \Xc$ such that $f(x)$ is as low as possible.
    \item The algorithm proceeds in iterations: At the $i$-th iteration, a point $x_i \in \Xc$ is queried, and an {\em oracle} returns a sample $y_i$ depending on the function, e.g., a noisy function value, a noisy gradient, or a tuple containing both.  The selected point $x_i$ can depend on the past queries and samples.
    \item After iteratively sampling $n$ points, the optimization algorithm returns a final point $\hat{x}$, and the {\em loss} incurred is $\ell_f(\hat{x}) = f(\hat{x}) - \min_{x \in \Xc} f(x)$, i.e., the gap to the optimal function value.
    \item For a given class of functions $\Fc$, the {\em minimax risk} is given by
    \begin{equation}
        \Mc_n(\Fc) = \inf_{\hat{X}} \sup_{f \in \Fc} \EE_f[ \ell_f(\hat{X}) ],
    \end{equation}
    where the infimum is over all optimization algorithms that iteratively query the function $n$ times and return a final point $\hat{x}$ as above, and $\EE_f$ denotes expectation when the underlying function is $f$.
\end{itemize}

\noindent In the following, we let $\Xv = (X_1,\dotsc,X_n)$ and $\Yv = (Y_1,\dotsc,Y_n)$ denote the queried locations and samples across the $n$ rounds. %, using capital letters to indicate that these are random in noisy scenarios.

\begin{thm} \label{thm:PackingConvex}
    {\em (Minimax bound for noisy optimization)}
    Fix $\epsilon > 0$, and let $\{f_1,\dotsc,f_M\} \subseteq \Fc$ be a finite subset of $\Fc$ such that for each $x \in \Xc$, we have $\ell_{f_v}(x) \le \epsilon$ for at most one value of $v \in \{1,\dotsc,M\}$.  Then we have
    \begin{equation}
        \Mc_n(\Fc) \ge \epsilon \cdot \bigg( 1 - \frac{I(V;\Xv,\Yv) + \log 2}{ \log M }\bigg), \label{eq:PackingConvex}
    \end{equation}
    where $V$ is uniform on $\{1,\dotsc,M\}$, and the mutual information is with respect to $V \to f_V \to (\Xv,\Yv)$.  Moreover, in the special case $M = 2$, we have
    \begin{equation}
        \Mc_n(\Fc) \ge \epsilon \cdot H_2^{-1}\big( \log 2 - I(V;\Xv,\Yv) \big), \label{eq:PackingConvex2}
    \end{equation}
    where $H_2^{-1}(\cdot) \in [0,0.5]$ is the inverse binary entropy function.
\end{thm}
\begin{proof}
    By Markov's inequality, we have
    \begin{equation}
    \sup_{f \in \Fc} \EE_f[ \ell_f(\hat{X}) ] \ge \sup_{f \in \Fc} \epsilon \cdot  \PP_f[ \ell_f(\hat{X}) \ge \epsilon ]. \label{eq:ConvExcessDist_main} 
    \end{equation}
    Suppose that a random index $V$ is drawn uniformly from $\{1,\dotsc,M\}$, and the triplet $(\Xv,\Yv,\hat{X})$ is generated by running the optimization algorithm on $f_V$.  Given $\hat{X} = \hat{x}$, let $\hat{V}$ index the function among $\{f_1,\dotsc,f_M\}$ with the lowest corresponding value: $\hat{V} = \argmin_{v=1,\dotsc,M} f_{v}(\hat{x})$.
    
    By the assumption that any $x$ satisfies $\ell_{f_v}(x) \le \epsilon$ for at most one of the $M$ functions, we find that the condition $\ell_{f_v}(\hat{x}) \le \epsilon$ implies $\hat{V} = v$.  Hence, we have
    \begin{equation}
    \PP_{v}\big[ \ell_{f_v}(\hat{X}) > \epsilon \big] \ge \PP_{f_v}[\hat{V} \ne v]. \label{eq:ConvExcessDistProb_main} 
    \end{equation}
    The remainder of the proof follows \eqref{eq:MinimaxEnd1}--\eqref{eq:MinimaxEnd4} in the proof of Theorem \ref{thm:Packing}: We lower bound the minimax risk $\sup_{f \in \Fc} \PP_{f}\big[ \ell_f(\hat{X}) \ge \epsilon \big]$ by the average over $V$, and apply Fano's inequality ({\em cf.}, Theorem \ref{thm:Fano} and Remark \ref{rem:FanoWeaken}) and the data processing inequality (\emph{cf.}, third part of Lemma \ref{lem:mi_adaptive}).
\end{proof}
% \noindent The proof is again analogous to that of Theorem \ref{thm:Packing}, and can be found in \cite{Rag11}. % is given in Appendix \ref{sec:pf_packing_p}.

\begin{rem}
    Theorem \ref{thm:PackingPartial} is based on reducing the optimization problem to a multiple hypothesis testing problem with exact recovery.  One can derive an analogous result reducing to approximate recovery, but we are unaware of any works making use of such a result for optimization.
\end{rem}

\section{Applications -- Continuous Settings} \label{sec:apps_cont}

In this section, we present three applications of the tools introduced in Section \ref{sec:continuous}: sparse linear regression, density estimation, and convex optimization.  Similarly to the discrete case, our examples are chosen to permit a relatively simple analysis, while still effectively exemplifying the key concepts and tools.

% All of the converse bounds given in this section are tight in terms of scaling laws, but not in terms of the constants.  Characterizing the constants precisely in continuous settings is typically a formidable task, though we briefly discuss some potential improvements in Section \ref{sec:discussion}.

\subsection{Sparse Linear Regression} \label{sec:sparse}

% The linear regression problem has recently regained considerable attention following the development of methods that exploit {\em structure}, thereby allowing accurate estimation even when there are far fewer measurements than parameters.  In this example, we consider the ubiquitous structure of {\em sparsity}, which also played a role in the group testing example of Section \ref{sec:apps_discrete}.  

In this example, we extend the $1$-sparse linear regression example of Section \ref{sec:overview} to the more general scenario of $k$-sparsity.  The setup is described as follows:
\begin{itemize}
    \item We wish to estimate a high-dimensional vector $\theta \in \RR^p$ that is {\em $k$-sparse}: $\|\theta\|_0 \le k$, where $\|\theta\|_0$ is the number of non-zero entries in $\theta$.
    \item The vector of $n$ measurements is given by $\Yv = \Xv\theta + \Zv$, where $\Xv \in \RR^{n \times p}$ is a known deterministic matrix, and $\Zv \sim \Ndist(\bzero,\sigma^2 \Iv_n)$ is additive Gaussian noise.
    \item Given knowledge of $\Xv$ and $\Yv$, an estimate $\thetahat$ is formed, and the loss is given by the squared $\ell_2$-error, $\ell(\theta,\thetahat) = \|\theta - \thetahat\|_2^2$, corresponding to \eqref{eq:ellPhiRho} with $\rho(\theta,\thetahat) = \|\theta - \thetahat\|_2$ and $\Phi(\cdot) = (\cdot)^2$.  Overloading the general notation $\Mc_n(\Theta,\ell)$, we write the minimax risk as
    \begin{equation}
        \Mc_n(k,\Xv) = \inf_{\thetahat} \sup_{\theta \in \RR^p\,:\,\|\theta\|_0 \le k} \EE_{\theta}[ \|\theta - \thetahat\|_2^2 ],
    \end{equation}
    where $\EE_{\theta}$ denotes expectation when the underlying vector is $\theta$.  
\end{itemize}

\subsubsection{Minimax Bound}

The lower bound on the minimax risk is formally stated as follows.  To simplify the analysis slightly, we state the result in an asymptotic form for the sparse regime $k = o(p)$; with only minor changes, one can attain a non-asymptotic variant attaining the same scaling laws for more general choices of $k$ \cite{Duc13}.

\begin{thm} \label{thm:SparseLinReg}
    {\em (Sparse linear regression)}
    Under the preceding sparse linear regression problem with $k = o(p)$ and a fixed regression matrix $\Xv$, we have
    \begin{equation}
        \Mc_n(k,\Xv) \ge \frac{\sigma^2 kp \log\frac{p}{k}}{ 32 \|\Xv\|_F^2 } (1+o(1))
    \end{equation}
    as $p \to \infty$.  In particular, under the constraint $\|\Xv\|_F^2 \le np \Gamma$ for some $\Gamma > 0$, achieving $\Mc_n(k,\Xv) \le \delta$ requires $n \ge \frac{\sigma^2 k \log\frac{p}{k}}{ 32\delta \Gamma } (1+o(1))$.
\end{thm}
\begin{proof}
    We present a simple proof based on a reduction to approximate recovery (\emph{cf.}, Theorem \ref{thm:PackingPartial}). In Section \ref{sec:alt_proof}, we discuss an alternative proof based on a reduction to exact recovery (\emph{cf.}, Theorem \ref{thm:Packing}).  
    
    We define the set 
    \begin{equation}
        \Vc = \big\{ v \in \{-1,0,1\}^p \,:\, \|v\|_0 = k \big\}, \label{eq:SparseChoiceV}
    \end{equation}
    and to each $v \in \Vc$, we associate a vector $\theta_v = \epsilon' v$ for some $\epsilon' > 0$.  Letting $d(v,v')$ denote the Hamming distance, we have the following properties:
    \begin{itemize}
        \item For $v,v' \in \Vc$, if $d(v,v') > t$, then $\|\theta_v - \theta_{v'}\|_2 > \epsilon'\sqrt{t}$;
        \item The cardinality of $\Vc$ is $|\Vc| = 2^k {p \choose k}$, yielding $\log|\Vc| \ge \log{p \choose k} \ge k\log\frac{p}{k}$;
        \item The quantity $\Nmax(t)$ in Theorem \ref{thm:PackingPartial} is the maximum possible number of $v' \in \Vc$ such that $d(v, v') \le t$ for a fixed $v$.  Setting $t = \frac{k}{2}$, a simple counting argument gives $\Nmax(t) \le \sum_{j=0}^{\lceil k/2 \rceil} 2^j {p \choose j} \le \big( \lceil\frac{k}{2}\rceil + 1\big) \cdot 2^{\lceil k/2 \rceil } \cdot {p \choose \lceil k/2 \rceil}$, which simplifies to $\log \Nmax(t) \le \big(\frac{k}{2}\log\frac{p}{k}\big)(1+o(1))$ due to the assumption $k = o(p)$.
    \end{itemize}
    From these observations, applying Theorem \ref{thm:PackingPartial} with $t = \frac{k}{2}$ and $\epsilon = \epsilon'\sqrt{\frac{k}{2}}$ yields
    \begin{equation}
        \Mc_n(k,\Xv) \ge \frac{k \cdot (\epsilon')^2}{8} \bigg( 1 - \frac{I(V;\Yv) + \log 2}{ \big(\frac{k}{2}\log\frac{p}{k}\big)(1+o(1))  }\bigg). \label{eq:SparseFano1}
    \end{equation}
    Note that we do not condition on $\Xv$ in the mutual information, since we have assumed that $\Xv$ is deterministic.
    
     To bound the mutual information, we first apply tensorization ({\em cf.}, first part of Lemma \ref{lem:mi_tensorization}) to obtain $I(V;\Yv) \le \sum_{i=1}^n I(V;Y_i)$, and then bound each $I(V;Y_i)$ using equation \eqref{eq:Aux2} in Lemma \ref{lem:mi_kl}.  We let $Q_Y$ be the $\Ndist(0,\sigma^2)$ density function, and we let $P_{v,i}$ denote the density function of $\Ndist( X_i^T \theta_v, \sigma^2)$, where $X_i$ is the transpose of the $i$-th row of $\Xv$.   Since the KL divergence between the $\Ndist(\mu_0,\sigma^2)$ and $\Ndist(\mu_1,\sigma^2)$ density functions is $\frac{(\mu_1 - \mu_0)^2}{2\sigma^2}$, we have $D(P_{v,i} \| Q_Y) = \frac{ | X_i^T \theta_v |^2 }{ 2\sigma^2 }$.  As a result, Lemma \ref{lem:mi_kl} yields $I(V;Y_i) \le \frac{1}{|\Vc|}\sum_{v} D(P_{v,i} \| Q_Y) = \frac{1}{2\sigma^2}\EE\big[ | X_i^T \theta_V |^2  \big]$ for uniform $V$.  Summing over $i$ and recalling that $\theta_v = \epsilon' v$, we deduce that
    \begin{equation}
        I(V;\Yv) \le \frac{(\epsilon')^2}{2\sigma^2} \EE[ \|\Xv V\|_2^2 ]. \label{eq:SparseInitMI}
    \end{equation}
    From the choice of $\Vc$ in \eqref{eq:SparseChoiceV}, we can easily compute $\cov[V] = \frac{k}{p} \Iv_p$, which implies that $\EE[ \|\Xv V\|_2^2 ] = \frac{k}{p} \|\Xv\|_F^2$.  Substitution into \eqref{eq:SparseInitMI} yields $I(V;\Yv) \le \frac{(\epsilon')^2}{2\sigma^2} \cdot \frac{k}{p} \|\Xv\|_F^2$, and we conclude from \eqref{eq:SparseFano1} that 
    \begin{equation}
        \Mc_n(k,\Xv) \ge \frac{k \cdot (\epsilon')^2}{8} \bigg( 1 - \frac{\frac{(\epsilon')^2}{2\sigma^2} \cdot \frac{k}{p} \|\Xv\|_F^2 + \log 2}{ \big(\frac{k}{2}\log\frac{p}{k}\big)(1+o(1))  }\bigg).
    \end{equation}
    The proof is concluded by setting $(\epsilon')^2 = \frac{\sigma^2 p \log \frac{p}{k}}{ 2\|\Xv\|_F^2 }$, which is chosen to make the bracketed term tend to $\frac{1}{2}$.
\end{proof}

Up to constant factors, the lower bound in Theorem \ref{thm:SparseLinReg} cannot be improved without additional knowledge of $\Xv$ beyond its Frobenius norm \cite{Can13}.  For instance, in the case that $\Xv$ has i.i.d.~Gaussian entries, a matching upper bound holds with high probability under maximum-likelihood decoding.

\subsubsection{Alternative Proof: Reduction with Exact Recovery} \label{sec:alt_proof}

In contrast to the proof given above (adapted from \cite{Duc13}), the first known proof of Theorem \ref{thm:SparseLinReg} was based on packing with {\em exact} recovery (\emph{cf.}, Theorem \ref{thm:Packing}) \cite{Can13}.  For the sake of comparison, we briefly outline this alternative approach, which turns out to be more complicated.

The main step is to prove the existence of a set $\{\theta_1,\dotsc,\theta_M\}$ satisfying the following properties:
\begin{itemize}
    \item The number of elements satisfies $M = \Omega\big( k\log\frac{p}{k} \big)$;
    \item Each element is $k$-sparse with non-zero entries equal to $\pm 1$;
    \item The elements are well-separated in the sense that $\|\theta_v - \theta_{v'}\|_2^2 = \Omega(k)$ for $v \ne v'$;
    \item The empirical covariance matrix is close to a scaled identity matrix in the following sense: $\big\| \frac{1}{M} \sum_{v=1}^M \theta_v \theta_v^T - \frac{k}{p} \cdot \Iv_p \|_{2\to2} = o\big(\frac{k}{p}\big)$, where $\|\cdot\|_{2\to2}$ denotes the $\ell_2/\ell_2$-operator norm, i.e., the largest singular value.
\end{itemize}
Once this is established, the proof proceeds along the same lines as the proof we gave above, scaling the vectors down by some $\epsilon' > 0$ and using Theorem \ref{thm:Packing} in place of Theorem \ref{thm:PackingPartial}.

The existence of the packing set is proved via a probabilistic argument: If one generates $\Omega\big( k\log\frac{p}{k} \big)$ uniformly random $k$-sparse sequences with non-zero entries equaling $\pm 1$, then these will satisfy the remaining two properties with positive probability.  While it is straightforward to establish the condition of being well-separated, the proof of the condition on the empirical covariance matrix requires a careful application of the non-elementary matrix Bernstein inequality. 

Overall, while the two approaches yield the same result up to constant factors in this example, the approach based on approximate recovery is entirely elementary and avoids the preceding difficulties.

\subsection{Density Estimation}

In this subsection, we consider the problem of estimating an entire probability density function given samples from its distribution, commonly known as {\em density estimation}.  We consider a non-parametric view, meaning that the density does not take any specific parametric form.  As a result, the problem is inherently {infinite-dimensional}, and lends itself to the {global} packing and covering approach introduced in Section \ref{sec:global_local}. 

While many classes of density functions have been considered in the literature \cite{Yan99}, we focus our attention on a specific setting for clarity of exposition:
\begin{itemize}
    \item The density function $f$ that we seek to estimate is defined on the domain $[0,1]$, i.e., $f(y) \ge 0$ for all $y \in [0,1]$, and $\int_{0}^1 f(y)dy = 1$.
    \item We assume that $f$ satisfies the following conditions:
    \begin{equation}
        f(y) \ge \eta, \forall y \in [0,1], \qquad \|f\|_{\TV} \le \Gamma
    \end{equation}
    for some $\eta \in (0,1)$ and $\Gamma > 0$, where the {\em total variation} (TV) norm is defined as $\|f\|_{\TV} = \sup_{L} \sup_{0 \le x_1 \le \dotsc \le x_L \le 1} \sum_{l=2}^L \big( f(x_l) - f(x_{l-1}) \big)$.  The set of all density functions satisfying these constraints is denoted by $\Fc_{\eta,\Gamma}$.
    \item Given $n$ independent samples $\Yv = (Y_1,\dotsc,Y_n)$ from $f$, an estimate $\hat{f}$ is formed, and the loss is given by $\ell(f,\hat{f}) = \|f - \hat{f}\|_2^2 = \int_0^1 ( f(x) - \hat{f}(x))^2 dx$.  Hence, the minimax risk is given by
    \begin{equation}
        \Mc_n(\eta,\Gamma) = \inf_{\hat{f}} \sup_{f \in \Fc_{\eta,\Gamma}} \EE_f\big[\|f - \hat{f}\|_2^2\big],
    \end{equation}
    where $\EE_f$ denotes expectation when the underlying density is $f$.
\end{itemize}

\subsubsection{Minimax Bound}

The minimax lower bound is given is follows.

\begin{thm} \label{thm:DensityEst}
    {\em (Density estimation)}
    Consider the preceding density estimation setup with some $\eta \in (0,1)$ and $\Gamma > 0$ not depending on $n$.  There exists a constant $c > 0$ (depending on $\eta$ and $\Gamma$) such that in order to achieve $\Mc_n(\eta,\Gamma) \le \delta$, it is necessary that
    \begin{equation}
        n \ge c \cdot \Big(\frac{1}{\delta}\Big)^{3/2}
    \end{equation}
    when $\delta$ is sufficiently small.  In other words, $\Mc_n(\eta,\Gamma) = \Omega\big(n^{-2/3}\big)$.
\end{thm}
\begin{proof}
    We specialize the general analysis of \cite{Yan99} to the class $\Fc_{\eta,\Gamma}$.
    Recalling the packing and covering numbers from Definitions \ref{def:PackingNumber} and \ref{def:CoveringNumber}, we adopt the shorthand notation $M_2^*(\epsp) = M_{\rho}^*(\Fc_{\eta,\Gamma},\epsp)$ with $\rho(f,f') = \|f - f'\|_2$, and similarly $N_2^*(\epsc) = N_{\rho}^*(\Fc_{\eta,\Gamma},\epsc)$.  % By applying Corollary \ref{cor:global}, we will implicitly perform the steps of Procedure \ref{alg:steps} as follows: (i) Form a multiple hypothesis test according to a maximal $\epsp$-packing; (ii) Apply Fano's inequality in the form given in Theorem \ref{thm:Packing}; (iii) Bound the mutual information via Lemma \ref{lem:Covering} with a minimal $\epsc$-covering.
    We first show that $\NKL^*$ ({\em cf.}, Corollary \ref{cor:global}) can be upper bounded in terms of $M_2^*$, which will lead to a minimax lower bound that depends only on the packing number $M_2^*$.  For $f_1,f_2 \in \Fc_{\eta,\Gamma}$, we have
    \begin{align}
        D(f_1 \| f_2) 
            &\le \int_0^1 \frac{(f_1(x)-f_2(x))^2}{f_2(x)} dx \label{eq:de_DivBound1} \\
            &\le \frac{1}{\eta}  \int_0^1 (f_1(x)-f_2(x))^2 dx \label{eq:de_DivBound2} \\
            &= \frac{1}{\eta} \|f_1 - f_2\|^2_2, \label{eq:de_DivBound3}
    \end{align}
    where \eqref{eq:de_DivBound1} follows since the KL divergence is upper bounded by the $\chi^2$-divergence ({\em cf.}, Lemma \ref{lem:relations}), and  \eqref{eq:de_DivBound2} follows from the assumption that the density is lower bounded by $\eta$.  From the definition of $\NKL^*$ in Corollary \ref{cor:global}, we deduce the following for any $\epsc > 0$:
    \begin{equation}
        \NKL^*(\epsc) \le N_2^*( \sqrt{\eta\epsc} ) \le M_2^*( \sqrt{\eta\epsc} ), \label{eq:de_DivBound4}
    \end{equation}
    where the first inequality holds because any $\sqrt{\eta\epsc}$-covering the the $\ell_2$-norm is also a $\epsc$-covering in the KL divergence due to \eqref{eq:de_DivBound3}, and the second inequality follows from Lemma \ref{lem:pack_cov}.
    
    Combining \eqref{eq:de_DivBound4} with Corollary \ref{cor:global} and the choice $\Phi(\cdot) = (\cdot)^2$ gives
    \begin{equation}
        \Mc_n(\eta,\Gamma) \ge \Big(\frac{\epsp}{2}\Big)^2 \bigg( 1 - \frac{\log M_2^*( \sqrt{\eta\epsc} ) + n\epsc + \log 2}{ \log M_2^*( \epsp ) }\bigg).
    \end{equation}
    We now apply the following bounds on the packing number of $\Fc_{\eta,\Gamma}$, which we state from \cite{Yan99} without proof:
    \begin{equation}
        \cunder \cdot \epsilon^{-1} \le \log M_2^*(\epsilon) \le \cbar \cdot \epsilon^{-1},
    \end{equation}
    for some constants $\cunder,\cbar > 0$ and sufficiently small $\epsilon > 0$. It follows that
    \begin{equation}
        \Mc_n(\eta,\Gamma) \ge \Big(\frac{\epsp}{2}\Big)^2 \bigg( 1 - \frac{ \cbar \cdot (\eta\epsc)^{-1/2}  + n\epsc + \log 2}{ \cunder \cdot \epsp^{-1} }\bigg). \label{eq:de_Minimax3}
    \end{equation}
    The remainder of the proof amounts to choosing $\epsp$ and $\epsc$ to balance the terms appearing in this expression.
    
    First, choosing $\epsc$ to equate the terms $\cbar \cdot (\eta\epsc)^{-1/2}$ and $n\epsc$ leads to $\epsc = \big(\frac{c'}{n}\big)^{2/3}$ with $c' = \cbar \eta^{-1/2}$, yielding $\frac{ \cbar \cdot (\eta\epsc)^{-1/2}  + n\epsc + \log 2}{ \cunder \cdot \epsp^{-1} } = \frac{2n\big(\frac{c'}{n}\big)^{2/3} + \log 2}{ \cbar\cdot\epsp^{-1} }$.  Next, choosing $\epsp$ to make this fraction equal to $\frac{1}{2}$ yields $\epsp^{-1} = \frac{2}{\cbar}\big( 2(c')^{2/3} n^{ 1/3 } + \log 2 \big)$, which means that $\epsp \ge c'' \cdot n^{-1/3}$ for suitable $c'' > 0$ and sufficiently large $n$.  Finally, since we made the fraction equal  to $\frac{1}{2}$, \eqref{eq:de_Minimax3} yields $\Mc_n(\eta,\Gamma) \ge \frac{\epsp^2}{8} \ge \frac{(c'')^2 n^{-2/3}}{8}$.  Setting $\Mc_n(\eta,\Gamma) = \delta$ and solving for $n$ yields the desired result.
\end{proof}

The scaling given in Theorem \ref{thm:DensityEst} cannot be improved; a matching upper bound is given in \cite{Yan99}, and can be achieved even when $\eta = 0$. % i.e., the density need not be lower bounded.

\subsection{Convex Optimization}

In our final example, we consider the optimization setting introduced in Section \ref{sec:opt}.  We provide an example that is rather simple, yet has interesting features not present in the previous examples: (i) an example departing from estimation; (ii) a continuous example with adaptivity; and (iii) a case where Fano's inequality with $|\Vc| = 2$ is used.

We consider the following special case of the general setup of Section \ref{sec:opt}:
\begin{itemize}
    \item We let $\Fc$ be the set of differentiable and {\em strongly convex} functions on $\Xc = [0,1]$, with strong convexity parameter equal to one:
    \begin{equation}
        \Fscv = \bigg\{ f \,:\, f\text{ is differentiable} \,\cap\, f(x)-\frac{1}{2}x^2\text{ is convex} \bigg\}.
    \end{equation}
    The analysis that we present can easily be extended to functions on an arbitrary closed interval with an arbitrary strong convexity parameter.
    \item When we query a point $x \in \Xc$, we observe a noisy sample of the function value and its gradient:
    \begin{equation}
        Y = ( f(x) + Z, f'(x) + Z' ),
    \end{equation}
      where $Z$ and $Z'$ are independent $\Ndist(0,\sigma^2)$ random variables, for some $\sigma^2 > 0$.  This is commonly referred to as the {\em noisy first-order oracle}.
\end{itemize}

\subsubsection{Minimax Bound}

\begin{figure}
    \begin{centering}
        \includegraphics[width=0.45\columnwidth]{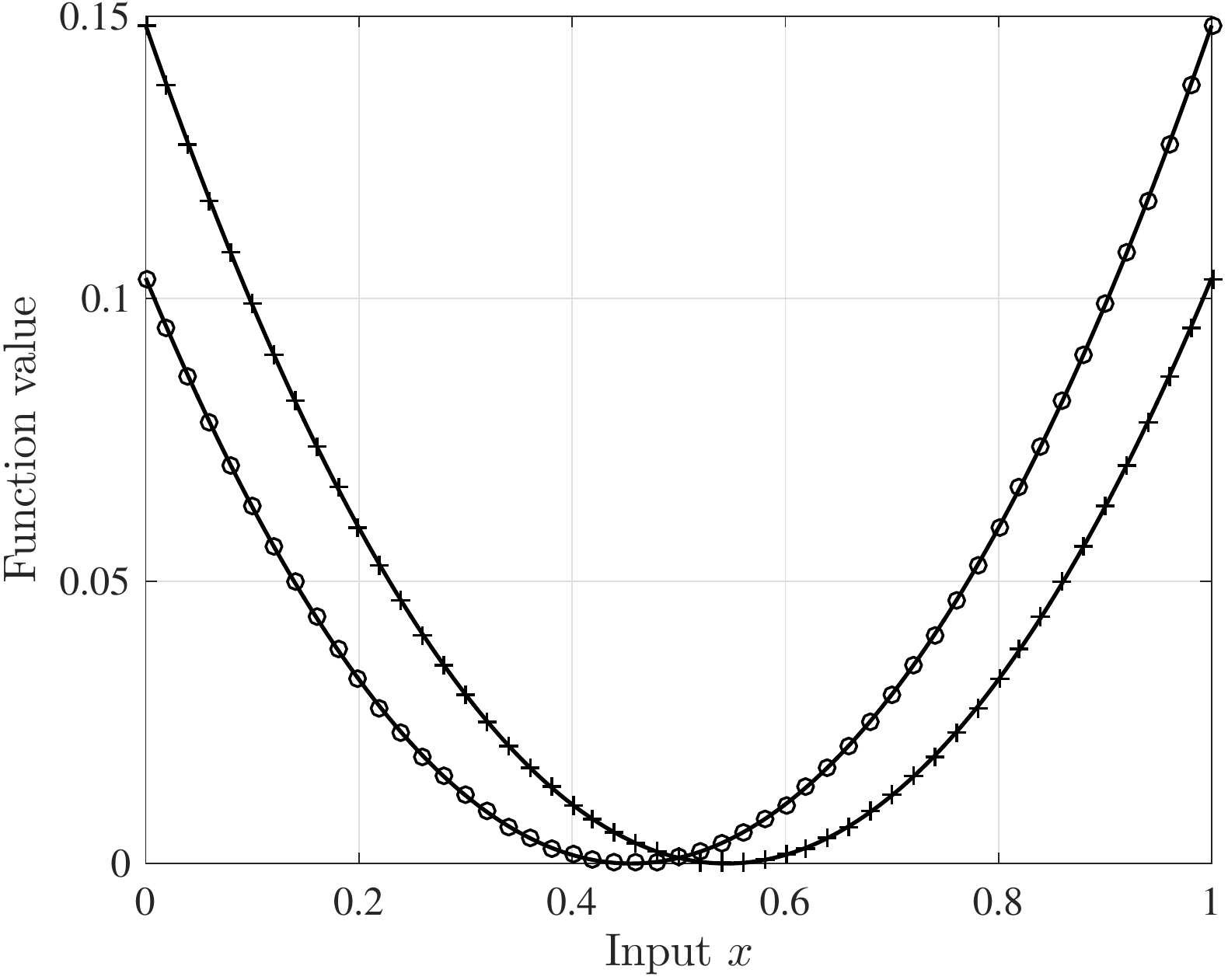}
        \par
    \end{centering}
    
    \caption{Construction of two functions in $\Fscv$ that are difficult to distinguish, and such that any point $x \in [0,1]$ can be $\epsilon$-optimal for only one of the two functions. \label{fig:ConvexOpt}}
\end{figure}

The following theorem lower bounds the number of queries required to achieve $\delta$-optimality.  The proof is taken from \cite{Rag11} with only minor modifications.

\begin{thm} \label{thm:StronglyConvexOpt}
    {\em (Stochastic optimization of strongly convex functions)}
    Under the preceding convex optimization setting with noisy first-order oracle information, in order to achieve $\Mc_n(\Fscv) \le \delta$, it is necessary  that
    \begin{equation}
         n \ge \frac{\sigma^2 \log 2}{40 \delta} \label{eq:StrongConvexResult}
    \end{equation}
    when $\delta$ is sufficiently small.
\end{thm}
\begin{proof}
    We construct a set of two functions satisfying the assumptions of Theorem \ref{thm:PackingConvex}.  Specifically, we fix $(\epsilon,\epsilon')$ such that $0 < \epsilon < \epsilon' < \frac{1}{8}$, define $x_{1}^* = \frac{1}{2} - \sqrt{2\epsilon'}$ and $x_{2}^* = \frac{1}{2} + \sqrt{2\epsilon'}$, and set 
    \begin{equation}
        f_v(x) = \frac{1}{2}(x - x_v^*)^2, \quad v=1,2.
   \end{equation}
   These functions are illustrated in Figure \ref{fig:ConvexOpt}.
   
    Since $\epsilon' \in \big(0,\frac{1}{8}\big)$, both $x_1^*$ and $x_2^*$ lie in $(0,1)$, and hence $\min_{x \in [0,1]} f_1(x) = \min_{x \in [0,1]} f_2(x) = 0$.  Moreover, a direct evaluation reveals that $f_1(x) + f_2(x) = \big(x - \frac{1}{2}\big)^2 + 2\epsilon' > 2\epsilon$, which implies that any $\epsilon$-optimal point for one function cannot be $\epsilon$-optimal for the other function.  This is the condition needed to apply Theorem \ref{thm:PackingConvex}, yielding from \eqref{eq:PackingConvex2} that
    \begin{equation}
        \Mc_n(\Fscv) \ge \epsilon \cdot H_2^{-1}( \log 2 - I(V;\Xv,\Yv) ). \label{eq:conv_ex_minimax}
    \end{equation}
    To bound the mutual information, we first apply tensorization ({\em cf.}, first part of Lemma \ref{lem:mi_adaptive}) to obtain $I(V;\Xv,\Yv) \le \sum_{i=1}^n I(V;Y_i | X_i)$.  We proceed by bounding $I(V;Y_i | X_i)$ for any given $i$.  Fix $x \in [0,1]$, let $P_{Y_{x}}$ and $P_{Y'_{x}}$ be the density functions of the noisy samples of $f_1(x)$ and $f'_1(x)$, and let $Q_{Y_x}$ and $Q_{Y'_x}$ be defined similarly for $f_0(x) = \frac{1}{2}\big(x - \frac{1}{2}\big)^2$. 
     We have
    \begin{align}
        D(P_{Y_{x}} \times P_{Y'_{x}} \| Q_{Y_x} \times Q_{Y'_x}) 
            &= D(P_{Y_{x}} \| Q_{Y_{x}}) + D(P_{Y'_{x}} \| Q_{Y'_{x}}) \label{eq:ConvDiv1} \\
            &= \frac{(f_1(x) - f_0(x))^2}{2\sigma^2} + \frac{(f'_1(x) - f'_0(x))^2}{2\sigma^2}, \label{eq:ConvDiv2}
    \end{align}
    where \eqref{eq:ConvDiv1} holds since the KL divergence is additive for product distributions, and \eqref{eq:ConvDiv2} uses the fact that the divergence between the $\Ndist(\mu_0,\sigma^2)$ and $\Ndist(\mu_1,\sigma^2)$ density functions is $\frac{(\mu_1 - \mu_0)^2}{2\sigma^2}$.
    
    Recalling that $f_1(x) = \frac{1}{2} \big(x - \frac{1}{2} + \sqrt{2\epsilon'}\big)^2$ and $f_0(x) = \frac{1}{2}\big(x - \frac{1}{2}\big)^2$, we have
    \begin{equation}
        (f_1(x) - f_0(x))^2 = \frac{1}{4}\bigg( 2\epsilon' + 2\Big(x - \frac{1}{2}\Big) \sqrt{2\epsilon'} \bigg)^2 \le \Big(\epsilon' + \sqrt{\frac{\epsilon'}{2}}\Big)^2 \le 2\epsilon',
    \end{equation}
    where the first inequality uses the fact that $x \in [0,1]$, and the second inequality follows since $\epsilon' < \frac{1}{8}$ and hence $\epsilon' = \sqrt{\epsilon'}\cdot\sqrt{\epsilon'} \le \sqrt{\frac{\epsilon'}{8}}$ (note that $\big(\frac{1}{\sqrt{8}}+\frac{1}{\sqrt{2}}\big)^2 \le 2$).  Moreover, taking the derivatives of $f_0$ and $f_1$ gives $(f'_1(x) - f'_0(x))^2 = 2\epsilon'$, and substitution into \eqref{eq:ConvDiv2} yields $D(P_{Y_{x}} \times P_{Y'_{x}} \| Q_{Y_x} \times Q_{Y'_x}) \le \frac{2\epsilon'}{\sigma^2}$.
    
    The preceding analysis applies in a near-identical manner when $f_2$ is used in place of $f_1$, and yields the same KL divergence bound when $(P_{Y_x},P_{Y'_x})$ is defined with respect to $f_2$.  As a result, for any $x \in [0,1]$, we obtain from \eqref{eq:Aux3} in Lemma \ref{lem:mi_kl} that $I(V;Y_i|X_i=x) \le \frac{2\epsilon'}{\sigma^2}$.  Averaging over $X$, we obtain $I(V;Y_i|X_i) \le \frac{2\epsilon'}{\sigma^2}$, and substitution into the above-established bound  $I(V;\Xv,\Yv) \le \sum_{i=1}^n I(V;Y_i | X_i)$ yields $I(V;\Xv,\Yv) \le \frac{2n\epsilon'}{\sigma^2}$.  Hence, \eqref{eq:conv_ex_minimax} yields
    \begin{equation}
        \Mc_n(\Fscv) \ge \epsilon \cdot H_2^{-1}\bigg( \log 2 - \frac{2n\epsilon'}{\sigma^2} \bigg). \label{eq:conv_ex_minimax2}
    \end{equation}
    Now observe that if $n \le \frac{\sigma^2 \log 2}{4\epsilon'}$ then the argument to $H_2^{-1}(\cdot)$ is at least $\frac{\log 2}{2}$.  It is easy to verify that $H_2^{-1}\big( \frac{\log 2}{2} \big) > \frac{1}{10}$, from which it follows that $\Mc_n(\Fscv) > \frac{\epsilon}{10}$.  Setting $\epsilon = 10\delta$ and noting that $\epsilon'$ can be chosen arbitrarily close to $\epsilon$, we conclude that the required number of samples $\frac{\sigma^2 \log 2}{4\epsilon'}$ recovers \eqref{eq:StrongConvexResult}.
\end{proof}

Theorem \ref{thm:StronglyConvexOpt} provides tight scaling laws, since stochastic gradient descent is known to achieve $\delta$-optimality for strongly convex functions using $O\big( \frac{\sigma^2}{\delta} \big)$ queries.  Analogous results for the multi-dimensional setting can be found in \cite{Rag11}.

%\subsection{Sparse Fourier transform}
%
%\subsubsection{Problem description}
%
%\subsubsection{Non-adaptive sampling}
%
%\subsubsection{Adaptive sampling}

%\subsection{Gaussian mean estimation}
%
%We present a simple example for the purpose of illustrating an application of the continuum Fano inequality (Theorem \ref{thm:Continuum}):
%\begin{itemize}
%    \item We are given i.i.d.~samples $\Yv$ from some distribution $P_{\theta}$ of the form $\Ndist(\theta,\sigma^2 \Iv_p)$ with mean $\theta \in \RR^p$, where $\Iv_p$ is the $p \times p$ identity matrix. 
%    \item An estimate $\hat{\theta}$ of $\theta$ is formed, and the error incurred is $\ell(\theta,\hat{\theta}) = \|\theta - \hat{\theta}\|_2^2$.
%    \item Hence, the minimax risk is
%    \begin{equation}
%        \Mc_n(p) = \inf_{\hat{\theta}} \sup_{\theta \in \RR^p} \|\theta - \hat{\theta}\|_2^2.
%    \end{equation}
%\end{itemize}
%
%\begin{thm}
%    {\em (Gaussian mean estimation)}
%    Under the preceding Gaussian mean estimation problem, we have
%    \begin{equation}
%        \Mc_n(p) \ge \frac{\sigma^2 p}{n} \cdot \frac{\log 2}{4}\Big(1 - \frac{1}{p}\Big)^2.
%    \end{equation}
%\end{thm}
%\begin{proof}
%    We fix $r > 0$ and lower bound the worst-case error by the average with $V$ uniformly distributed in $\BB_2(0,r)$, defined to be the $\ell_2$-ball of radius $r$ centered at zero.  Since the volume is proportional to the dimension, we have
%    \begin{equation}
%        \frac{\Vol(\BB_2(0,r))}{\BB_2(v',r)} = \Big(\frac{r}{t}\Big)^d
%    \end{equation}
%    for any $v' \in \RR^p$ and $t > 0$. 
%\end{proof}

\section{Discussion} \label{sec:discussion}

\subsection{Limitations of Fano's Inequality} \label{sec:limitations}

While Fano's inequality is a highly versatile method with successes in a wide range of statistical applications ({\em cf.}, Table \ref{tbl:applications}), it is worth pointing out some of its main limitations.  We briefly mention some alternative methods below, as well as discussing some suitable generalizations of Fano's inequality in Section \ref{sec:generalizations}.

{\bf Non-asymptotic weakness.} Even in scenarios where Fano's inequality provides converse bounds with the correct asymptotics including constants, these bounds can be inferior to alternative methods in the non-asymptotic sense \cite{Pol10,Joh15}.   Related to this issue is the distinction between the {\em weak converse} and {\em strong converse}: We have seen that Fano's inequality typically provides necessary conditions of the form $n \ge n^* (1 - \delta - o(1))$ for achieving $\pe \le \delta$, in contrast with strong converse results of the form $n \ge n^* (1 - o(1))$ for {\em any} $\delta \in (0,1)$.  Alternative techniques addressing these limitations are discussed in the context of communication in \cite{Pol10}, and in the context of statistical estimation in \cite{Ven17,Loh17}.

%One widespread method for obtaining strong converse results is {\em multiplicative change of measure}, which relates the probability of a success event $\Ac$ under two different distributions $P^n(\yv),Q^n(\yv)$ as follows (e.g., see \cite{Tsy09,Han03,Pol10}):
%\begin{gather}
%    \PP_P[\Ac] \le \PP_P\Big[ \frac{P^n(\Yv)}{Q^n(\Yv)} > \gamma \Big] + \gamma\PP_Q[\Ac], \label{eq:mult_chg}
%\end{gather}
%where $\gamma$ is an arbitrary threshold.  This approach can also be used to improve the constant factors arising in continuous estimation settings \cite{Ven17}.

{\bf Difficulties in adaptive settings.} While we have provided examples where Fano's inequality provides tight bounds in adaptive settings, there are several applications where alternative methods have proved to be more suitable.  One reason for this is that the conditional mutual information terms $I(V;Y_i|X_i)$ ({\em cf.}, Lemma \ref{lem:mi_adaptive}) often involve complicated conditional distributions that are difficult to analyze.  We refer the reader to \cite{Lai85,Aue95a,Ari13} for examples in which alternative techniques proved to be more suitable for adaptive settings.

%It is common to instead use change-of-measure techniques, including both variations of \eqref{eq:mult_chg} (e.g., see \cite{Lai85}) and alternatives based on {\em additive change of measure}: For any random variable $A \in [0,\Amax]$, one has
%\begin{align}
%    \EE_P[A] 
%    &\le \EE_Q[A] + {\Amax} \dTV(P, Q), \label{eq:chg_add_TV} 
%    % &\le \EE_Q[A] + \Amax \sqrt{\frac{1}{2}\min\{D(P\|Q),D(Q\|P)\}}, \label{eq:chg_add_TV2}
%\end{align}
%where $\dTV(P,Q) = \frac{1}{2}\EE_Q\big[ \big| \frac{P(Y)}{Q(Y)} - 1 \big| \big]$ is the total variation distance.  In particular, if $A$ is the indicator function of an event $\Ac$, then \eqref{eq:chg_add_TV} reduces to $\PP_P[\Ac] \le \PP_Q[\Ac] +  \dTV(P, Q)$.  

% We list two examples here: (i) For the problem of recovering sparse vectors with adaptive measurements (i.e., an adaptive counterpart to the sparse linear regression example of Section \ref{sec:sparse}), Fano's inequality has been successfully applied only in the $1$-sparse case \cite{Pri13}.  Converse bounds for higher sparsity levels were developed via an entirely different approach \cite{Ari13}. (ii) Among the extensive literature on bandit algorithms, the vast majority of lower bounds have been derived using alternative techniques.

{\bf Restriction to KL divergence.} When applying Fano's inequality, one invariably needs to bound a mutual information term, which is an instance of the KL divergence.  While the KL divergence satisfies a number of convenient properties that can help in this process, it is sometimes the case that other divergence measures are more convenient to work with, or can be used to derive tighter results.  Generalizations of Fano's inequality have been proposed specifically for this purpose, as we discuss in the following subsection.

\subsection{Generalizations of Fano's Inequality} \label{sec:generalizations}

Several variations and generalizations of Fano's inequality have been proposed in the literature \cite{Han94,Bir05,Gus03,Gun11,Pol10a,Bra15}.  Most of these are not derived based on the most well-known proof of Theorem \ref{thm:Fano}, but are instead based on an alternative proof via the {data processing inequality for KL divergence}: For any event $E$, one has
\begin{equation}
    I(V;\hat{V})  = D(P_{\fV\hat{V}} \| P_{\fV}\times P_{\hat{V}}) \ge D_2\big( P_{\fV\hat{V}}[E] \,\|\, (P_{\fV}\times P_{\hat{V}})[E] \big), \label{eq:alt_proof}
\end{equation}
where $D_2(p\|q) = p\log\frac{p}{q} + (1-p)\log\frac{1-p}{1-q}$ is the binary KL divergence function.  Observe that if $V$ is uniform and $E$ is the event that $V \ne \hat{V}$, then we have $P_{\fV\hat{V}}[E] = \pe$ and $(P_{\fV}\times P_{\hat{V}})[E] = 1 - \frac{1}{|\Vc|}$, and Fano's inequality ({\em cf.}, Theorem \ref{thm:Fano}) follows by substituting the definition of $D_2(\cdot\|\cdot)$ in \eqref{eq:alt_proof} and re-arranging.  This proof lends itself to interesting generalizations, including the following.

{\bf Continuum version.} Consider a continuous random variable $V$ taking values on $\Vc \subseteq \RR^p$ for some $p \ge 1$, and an error probability of the form $\pe(t) = \PP\big[ d(\fV,\hat{V}) > t \big]$ for some real-valued function $d$ on $\RR^p \times \RR^p$.  This is the same formula as \eqref{eq:pe_partial}, which we previously introduced for the discrete setting.   Defining the ``ball'' $\BB_d(\hat{v},t) = \{v \in \RR^p \,:\, d(v,\hat{v}) \le t\}$ centered at $\hat{v}$, \eqref{eq:alt_proof} leads to the following for $V$ uniform on $\Vc$: 
\begin{equation}
    \pe(t) \ge 1 - \frac{I(V;\hat{V}) + \log 2}{ \log\frac{\Vol(\Vc)}{ \sup_{\hat{v} \in \RR^p} \Vol(\Vc \cap \BB_d(\hat{v},t)) } },
\end{equation}
where $\Vol(\cdot)$ denotes the volume of a set.  This result provides a continuous counterpart to the final part of Theorem \ref{thm:Partial}, in which the cardinality ratio is replaced by a volume ratio.  We refer the reader to \cite{Duc13} for example applications, and to \cite{Bra15} for the simple proof outlined above.

{\bf Beyond KL divergence.} The key step \eqref{eq:alt_proof} extends immediately to other measures that satisfy the data processing inequality.  A useful class of such measures is the class of {\em $f$-divergences}: $D_f(P\|Q) = \EE_Q\big[ f\big(\frac{P(\Yv)}{Q(\Yv)}\big) \big]$ for some convex function $f$ satisfying $f(1) = 0$.  Special cases include KL divergence ($f(z) = z\log z$), total variation ($f(z) = \frac{1}{2}|z-1|$), squared Hellinger distance ($f(z) = (\sqrt{z}-1)^2$), and $\chi^2$-divergence ($f(z) = (z-1)^2$). It was shown in \cite{Gun11} that alternative choices beyond the KL divergence can provide improved bounds in some cases.  Generalizations of Fano's inequality beyond $f$-divergences can be found in \cite{Pol10a}.

{\bf Non-uniform priors.} The first form of Fano's inequality in Theorem \ref{thm:Fano} does not require $V$ to be uniform.  However, in highly non-uniform cases where $H(V) \ll \log |\Vc|$, the term $\pe \log( |\Vc| - 1 )$ may be too large for the bound to be useful.  In such cases, it is often useful to use different Fano-like bounds based on the alternative proof above.  In particular, the step \eqref{eq:alt_proof} makes no use of uniformity, and continues to hold even in the non-uniform case.  In \cite{Han94}, this bound was further weakened to provide simpler lower bounds for non-uniform settings with discrete alphabets.  Fano-type lower bounds in {\em continuous} Bayesian settings with non-uniform priors arose more recently, and are typically more technically challenging; the interested reader is referred to \cite{Xu17,Che16a}.

\appendix

\section{Appendix}

Here we provide the omitted proofs from the main body.  Throughout the proofs, the random variables $V$ and $\hat{V}$ are assumed to be discrete, whereas the other random variables involved, including the inputs $\Xv = (X_1,\dotsc,X_n)$ and samples $\Yv = (Y_1,\dotsc,Y_n)$, may be continuous.  In such cases, entropy quantities such as $H(Y_i)$ should be interpreted as being the {\em differential entropy} \cite[Ch.~8]{Cov01}, and probability functions such as $P_Y(y)$ should be interpreted as being a probability density function (PDF).

\subsection{Preliminary Information-Theoretic Results} \label{sec:it_prelim}

The following lemma states some useful results from information theory.  The proofs can be found in standard references such as \cite{Cov01}.

\begin{lem} \label{lem:it_prelim}
    {\em (Standard information-theoretic results)}
    We have the following:
    \begin{itemize}
        \item {\em (Chain rule for entropy)} $H(Y_1,\dotsc,Y_n) = \sum_{i=1}^n H(Y_i|Y_1,\dotsc,Y_{i-1})$.
        \item  {\em (Chain rule for mutual information)} $I(X;Y_1,\dotsc,Y_n) = \sum_{i=1}^n I(X;Y_i|Y_1,\dotsc,Y_{i-1})$. 
        \item {\em (Sub-additivity of entropy)} $H(Y_1,\dotsc,Y_n) \le \sum_{i=1}^n H(Y_i)$.
        \item {\em (Conditioning reduces entropy)} $H(Y|X) \le H(Y)$.
        % \item {\em (Discrete entropy bound)} For $X$ on a finite alphabet $\Xc$, $H(X) \le \log|\Xc|$ with equality for uniform $X$.
        % \item {\em (Data processing inequality for mutual information)} If $X \to Y \to Z$, then $I(X;Z) \le I(X;Y)$.  
        % \item {\em (Data processing inequality for KL divergence)} Fix any $P_X$ and $Q_X$, any let $P_Y$ and $Q_Y$ be the marginals of $P_X \times P_{Y|X}$ and $Q_X \times P_{Y|X}$.  Then $D(P_Y\|Q_Y) \le D(P_X \| Q_X)$.
        % \item {\em (Alternative form of mutual information)} $I(X;Y) = \min_{Q_Y} D(P_{XY} \| P_X \times Q_Y)$, and the minimum is achieved by $Q_Y = P_Y$.
        % \item  {\em (Conditional vs.~unconditional mutual information)} If $X \to Y \to Z$, then $I(X;Y|Z) \le I(X;Y)$. {\bf [TODO: Consider removing]}
        \item {\em (Information-preserving transform)} If $Y$ depends on $X$ only through $f(X)$, then $H(Y|X,f(X)) = H(Y|f(X))$, and $I(X;Y) = I( f(X); Y )$.
        \item {\em (Capacity of binary symmetric channel)} If $X,Y$ are binary with $Y = X \oplus Z$ for $Z \sim \Bernoulli(\epsilon)$ (where $\oplus$ denotes modulo-2 addition), then $I(X;Y) \le \log 2 - H_2(\epsilon)$.
        \item {\em (Divergence between independent pairs)} $D(P_X \times P_Y \| Q_X \times Q_Y) = D(P_X \| Q_X) + D(P_Y \| Q_Y)$.
        \item {\em (Divergence between equal-variance univariate Gaussians)} For $X \sim \Ndist(\mu_1,\sigma^2)$ and $Y \sim \Ndist(\mu_2,\sigma^2)$, it holds that $D(P_X\|P_Y) = \frac{(\mu_1 - \mu_2)^2}{2\sigma^2}$.
        % \item {\em (Divergence between zero-mean multivariate Gaussians)} For $n$-dimensional vectors $\Xv \sim \Ndist(\bzero,\bSigma_1)$ and $\Yv \sim \Ndist(\bzero,\bSigma_2)$, it holds that $D(P_{\Xv} \| P_{\Yv}) = \frac{1}{2}\big( \Tr(\bSigma_2^{-1}\bSigma_1) - n + \log\frac{\det \bSigma_2}{ \det\bSigma_1 } \big)$.
        % \item {\em (Taylor expansion of binary entropy)} For any $p \in (0,1)$, it holds that $H_2(p + \delta) = H_2(p) + \delta\log\frac{1-p}{p} + o(\delta)$ as $\delta \to 0$.
        % \item {\bf \color{red} [TODO: Possible additions -- $H(X) \le \log|\Xc|$, $I(X;Y) \ge 0$, $H(X) \ge 0$ (discrete only), Gaussian maximizes entropy (continuous only), $I(X;Y) = D(P_{XY}\|P_X \times P_Y)$, $I(X;Y) = I(Y;X)$, $I(X;Y) \le I(X;Y,Z)$]}
    \end{itemize}
\end{lem}

\noindent We will make use of these results without necessarily referencing the lemma.

% \subsection{Proof of Theorem \ref{thm:Fano} (Fano's inequality)} \label{sec:pf_Fano}

\subsection{Proof of Theorem \ref{thm:Fano} (Fano's Inequality)}

Defining the error indicator random variable $E = \openone\{\fV \ne \hat{V}\}$, we have
\begin{align}
    H(\fV|\hat{V}) 
        &= H(\fV,E|\hat{V}) \label{eq:Fano_pf1} \\
        &= H(E|\hat{V}) + H(\fV|\hat{V},E) \label{eq:Fano_pf2} \\
        &\le  H(E) + H(\fV|\hat{V},E) \label{eq:Fano_pf3} \\
        &=  H_2(\pe) + \pe H(\fV|\hat{V},E = 1) + (1-\pe) H(\fV|\hat{V},E = 0) \label{eq:Fano_pf4} \\
        &= H_2(\pe) + \pe \log\big(|\Vc| - 1\big), \label{eq:Fano_pf5}
\end{align}
where \eqref{eq:Fano_pf1} holds since $E$ is a deterministic function of $(\fV,\hat{V})$, \eqref{eq:Fano_pf2} follows from the chain rule, \eqref{eq:Fano_pf3} holds since conditioning reduces entropy, \eqref{eq:Fano_pf4} uses $H(E) = H_2(\pe)$, and \eqref{eq:Fano_pf5} follows since $\fV$ has no uncertainty given $\hat{V}$ when $E=0$, and takes one of $|\Vc| - 1$ values given $\hat{V}$ when $E=1$.

In case that $\fV$ is uniform, we obtain \eqref{eq:Fano2} by upper bounding $|\Vc| - 1 \le |\Vc|$ and $H_2(\pe) \le \log 2$ in \eqref{eq:Fano1}, subtracting $H(V) = \log|\Vc|$ on both sides, and taking the negative on both sides.

\subsection{Proof of Theorem \ref{thm:Partial} (Fano's Inequality with Approximate Recovery)} \label{sec:pf_Partial}

    Define the error event $E_{t} = \{ d(V,\hat{V}) > t \}$.  Following the steps \eqref{eq:Fano_pf1}--\eqref{eq:Fano_pf4} with $E_t$ in place of $E$, we obtain
\begin{align}
    H(\fV|\hat{V}) 
    &\le  H_2(\pe(t)) + \pe(t) H(\fV|\hat{V},E_{t} = 1) + (1-\pe(t)) H(\fV|\hat{V},E_{t} = 0) \label{eq:Partial_pf1} \\
    &\le H_2(\pe(t)) + \pe(t) \log\big( |\Vc| - \Nmin(t)\big) + (1-\pe(t)) \log \Nmax(t) \label{eq:Partial_pf2} \\
    &= H_2( \pe(t) ) + \pe(t)\log\frac{|\Vc| - \Nmin(t)}{\Nmax(t)} + \log \Nmax(t),
\end{align}
where \eqref{eq:Partial_pf2} follows since when $\hat{V}$ is given and $E_t = 0$, $\fV$ takes one of at most  $\Nmax(t)$ values, whereas if $\hat{V}$ is given and $E_t = 1$, $\fV$ takes one of at most $|\Vc| - \Nmin(t)$ values.  We have thus proved \eqref{eq:Partial1}.

In case that $\fV$ is uniform, we obtain \eqref{eq:Partial2} by upper bounding $|\Vc| - \Nmin(t) \le |\Vc|$ and $H_2(\pe(t)) \le \log 2$ in \eqref{eq:Partial1}, subtracting $H(V) = \log|\Vc|$ on both sides, and taking the negative on both sides.

\subsection{Proof of Lemma \ref{lem:dpi} (Data Processing Inequality)} \label{sec:pf_dpi}

We focus on the first part, since the second and third parts follow as special cases.  We have
\begin{align}
    I(V;\hat{V}) &= H(V) - H(V|\hat{V}) \label{eq:dpi_pf1} \\
        &\le H(V) - H(V|\hat{V},\Yv) \label{eq:dpi_pf2} \\
        &= H(V) - H(V|\Yv) \label{eq:dpi_pf3} \\
        &= I(V;\Yv), \label{eq:dpi_pf4}
\end{align}
where \eqref{eq:dpi_pf2} follows since conditioning reduces entropy, and \eqref{eq:dpi_pf3} holds because $V$ and $\hat{V}$ are conditionally independent given $\Yv$.

\subsection{Proof of Lemma \ref{lem:mi_tensorization} (Tensorization)} \label{sec:pf_tens}

    We start with the second claim, since the first claim then follows by letting each $X_i$ deterministically equal an arbitrary fixed value (e.g., zero).  To prove the second claim, we write
\begin{align}
    I(V;\Yv|\Xv) 
    &= H(\Yv|\Xv) - H(\Yv | V,\Xv) \label{eq:tensor1} \\
    &\le \sum_{i=1}^n H(Y_i|X_i) - H(\Yv|V,\Xv) \label{eq:tensor2} \\
    &= \sum_{i=1}^n \big( H(Y_i|X_i) - H(Y_i|V,\Xv) \big) \label{eq:tensor3} \\
    &= \sum_{i=1}^n \big( H(Y_i|X_i) - H(Y_i|V,X_i) \big) \label{eq:tensor4} \\
    &= \sum_{i=1}^n I(V;Y_i|X_i), \label{eq:tensor5}
\end{align}
where \eqref{eq:tensor2} follows from the sub-additivity of entropy and the fact that conditioning reduces entropy, \eqref{eq:tensor3} follows from the conditional independence of the $Y_i$ given $(V,\Xv)$, and \eqref{eq:tensor4} follows from the assumption that $Y_i$ depends on $(V,\Xv)$ only on through $(V,X_i)$.  

The third claim follows from the second claim by writing
\begin{align}
    I(\fV;Y_i|X_i) \le I(V,X_i;Y_i) = I(U_i;Y_i)
\end{align}
by the assumption that $Y_i$ depends on $(V,X_i)$ only through $U_i$.

\subsection{Proof of Lemma \ref{lem:mi_adaptive} (Tensorization with Adaptivity)} \label{sec:pf_adaptive}

    We have the following:
\begin{align} 
    I(\fV;\Xv,\Yv) 
    &= \sum_{i=1}^n I(X_i, Y_i; \fV \,|\, X_1^{i-1}, Y_1^{i-1}) \label{eq:Adaptive1} \\
    &= \sum_{i=1}^n I(Y_i; \fV \,|\, X_1^{i-1},Y_1^{i-1},X_i) \label{eq:Adaptive2}\\
    &=  \sum_{i=1}^n  \Big( H(Y_i \,|\, X_1^{i-1},Y_1^{i-1},X_i) - H(Y_i \,|\, X_1^{i-1},Y_1^{i-1},X_i,V) \Big) \label{eq:Adaptive3} \\
    &=  \sum_{i=1}^n  \Big( H(Y_i \,|\, X_1^{i-1},Y_1^{i-1},X_i)  - H(Y_i \,|\, V,X_i) \Big) \label{eq:Adaptive4} \\
    &\le  \sum_{i=1}^n  \Big( H(Y_i \,|\, X_i) - H(Y_i | \fV,X_i) \Big) \label{eq:Adaptive5} \\
    &= \sum_{i=1}^n I(\fV;Y_i|X_i), \label{eq:Adaptive6}
\end{align}
where \eqref{eq:Adaptive1} follows from the chain rule, \eqref{eq:Adaptive2} follows since $X_i$ is a function of $(X_1^{i-1},Y_1^{i-1})$, \eqref{eq:Adaptive4} follows since $Y_i$ is conditionally independent of $(X_1^{i-1},Y_1^{i-1})$ given $(\fV,X_i)$, and \eqref{eq:Adaptive5} follows since conditioning reduces entropy.   This completes the proof of the first part.

To prove the second part, we note that
\begin{align}
    I(\fV;Y_i|X_i) \le I(V,X_i;Y_i) = I(U_i;Y_i)
\end{align}
by the assumption that $Y_i$ depends on $(X_i,V)$ only through $U_i$.

%\subsection{Proof of Lemma \ref{lem:mi_kl} (KL divergence based bounds)} \label{sec:pf_kl}
%
%We obtain \eqref{eq:Aux1} from the definition of mutual information, and \eqref{eq:Aux2} from the fact that $I(V;Y) = \EE\big[ \log\frac{P_{Y|V}(Y|V)}{P_{Y}(Y)}\big] = \EE\big[ \log\frac{P_{Y|V}(Y|V)}{Q_{Y}(Y)}\big] - \EE\big[ \log\frac{P_{Y}(Y)}{Q_{Y}(Y)}\big]$; the second term here is a KL divergence, and is therefore non-negative.  We obtain \eqref{eq:TwoThetas1} from \eqref{eq:Aux2} by noting that $Q_{Y}$ can be chosen to be any of the $P_{Y}(\cdot\,|\,v')$, and the remaining inequalities \eqref{eq:Aux3} and \eqref{eq:TwoThetas2} are trivial.

\subsection{Proof of Lemma \ref{lem:Covering} (Covering-Based Mutual Information Bound)} \label{sec:pf_covering}

Applying \eqref{eq:Aux3} in Lemma \ref{lem:mi_kl} with the choice $Q_{Y}(y) = \frac{1}{N}\sum_{j=1}^N Q_j(y)$, and letting $\EE_{v}$ denote expectation with respect to $P_{Y|V}(\cdot\,|v)$, we have
\begin{align}
    I(V;Y) 
    &\le \max_{v} D\bigg( P_{Y|V}(\cdot\,|\,v) \,\bigg\|\, \frac{1}{N} \sum_{j=1}^NQ_{j}\bigg) \label{eq:Covering1} \\
    &= \max_{v} \EE_{v}\bigg[ \log\frac{ P_{Y|V}(Y\,|\,v) }{ \frac{1}{N} \sum_{j=1}^NQ_{j}(Y) } \bigg] \label{eq:Covering2} \\
    &\le \max_{v} \EE_{v}\bigg[ \log\frac{ P_{Y|V}(Y\,|\,v) }{ \frac{1}{N} Q_{j^*(v)}(Y) } \bigg] \label{eq:Covering3} \\
    &= \log N + \max_{v} D\big( P_{Y|\fV}(\cdot\,|\,v) \,\big\|\, Q_{j^*(v)}\big)  \label{eq:Covering4} \\
    &\le \log N + \epsilon, \label{eq:Covering5} 
\end{align}
where \eqref{eq:Covering2} applies the definition of KL divergence, \eqref{eq:Covering3} lower bounds the summation by the single term $j^*(v)$ achieving the minimum in \eqref{eq:CoveringCond}, and \eqref{eq:Covering5} applies the upper bound in \eqref{eq:CoveringCond}.

\subsection{Omitted Details in Discrete Examples with Approximate Recovery}

\subsubsection{Group Testing}

Here we characterize the asymptotic behavior of the logarithm in \eqref{eq:gt_partial_init}.  The main step is to upper bound the summation in the denominator, which is given by $\sum_{j=0}^{\lfloor \alpha k \rfloor} {p-L \choose j}{L \choose k -j}$.  By the assumption $L = o(p)$, the value $j = \lfloor \alpha k \rfloor$ must yield the highest value of ${p-L \choose j}{L \choose k -j}$ when $p$ is sufficiently large.  Hence, upper bounding the summation by $\alpha k + 1$ times the maximum yields $\sum_{j=0}^{\lfloor \alpha k \rfloor} {p-L \choose j}{L \choose k -j} \le (\alpha k + 1) {p-L \choose \lfloor \alpha k \rfloor}{L \choose k - \lfloor \alpha k \rfloor}$.  Applying $\log {a \choose b} \le b \log \frac{ae}{b}$, we deduce that
\begin{equation}    
    \log \sum_{j=0}^{\lfloor \alpha k \rfloor} {p-L \choose j}{L \choose k -j} \le \log(\alpha k + 1) + \lfloor \alpha k \rfloor \log \frac{p e}{\lfloor \alpha k \rfloor} + (k - \lfloor \alpha k\rfloor) \log \frac{Le}{ k - \lfloor \alpha k\rfloor}.
\end{equation}
Since $\log\frac{p}{k} \to \infty$ and $\alpha \in (0,1)$ does not depend on $p$, a simple asymptotic analysis yields $\log \sum_{j=0}^{\lfloor \alpha k \rfloor} {p-L \choose j}{L \choose k-j} \le \big( \alpha k \log \frac{p}{k} + (1-\alpha) k \log \frac{L}{k} \big)(1+o(1))$.  The logarithm in \eqref{eq:gt_partial_init} therefore simplifies to $\big(k \log \frac{p}{L}\big)(1+o(1))$, as desired.

% The proof is concluded by substituting into \eqref{eq:gt_partial_init} and applying $\log {p \choose k} \ge k\log\frac{p}{k}$ and $(1-\alpha)k \log \frac{p}{k} - (1-\alpha) k \log \frac{L}{k} = (1-\alpha)k \log \frac{p}{L}$, as well as upper bounding the conditional mutual information using \eqref{eq:gt_exact1_mi2}.

\subsubsection{Graphical Model Selection}

Here we upper bound the quantity $\Nmax(\alpha p)$ in \eqref{eq:GraphNmax_init2}.
For all $j$, the first combinatorial term is upper bounded by $2^p$, the second is maximized by $j = \alpha p$ (for sufficiently large $p$), and further upper bounding ${p \choose 2} - p + 1 \le p^2$ yields $\Nmax(\alpha p) \le (\alpha p + 1) \cdot 2^p \cdot {p^2 \choose \alpha p}$.  Taking the logarithm and applying $\log{a \choose b} \le a\log\frac{ae}{b}$ along with asymptotic simplifications, we find that $\Nmax(\alpha p) \le \big(\alpha p \log p\big) (1+o(1))$, as desired.

\subsection{Proof of Theorem \ref{thm:PackingPartial}  (Reduction to Approximate Recovery)} \label{sec:pf_packing_p}

We adopt the same general approach as Theorem \ref{thm:Packing}, but instead of the error probability $\PP_v[\hat{V} \ne v]$, we consider an approximate recovery version of the form $\PP_v[d(v,\hat{V}) > t]$.  We again start with \eqref{eq:ExcessDist}, which we repeat here:
\begin{align}
    \sup_{\theta \in \Theta} \EE_{\theta}\big[ \ell(\theta,\thetahat) \big] 
    = \Phi(\epsilon_0) \sup_{\theta \in \Theta} \PP_{\theta}[ \rho(\theta,\thetahat) \ge \epsilon_0], \label{eq:ExcessDist_p}
\end{align}
for any $\epsilon_0 > 0$. Consider the following minimum-distance rule for $\hat{V}$:
\begin{equation}
    \hat{V} = \argmin_{v=1,\dotsc,M} \rho(\theta_v,\thetahat), \label{eq:m_hat}
\end{equation}
and suppose that we have $\rho(\theta_v,\thetahat) < \frac{\epsilon}{2}$ for the correct index $v$.  Then for any $v' \in \Vc$ such that $d(v,v') > t$, we have
\begin{align}
    \rho(\theta_{v'},\thetahat) 
    &\ge \rho(\theta_v,\theta_{v'}) - \rho(\theta_v,\thetahat) \label{eq:PartialRhoBound1}\\
    &> \epsilon - \frac{\epsilon}{2} = \frac{\epsilon}{2}, \label{eq:PartialRhoBound2} 
\end{align}
where \eqref{eq:PartialRhoBound1} follows from the triangle inequality, and \eqref{eq:PartialRhoBound2} follows from \eqref{eq:PackingAssumpPartial} and the assumption $\rho(\theta_v,\thetahat) < \frac{\epsilon}{2}$.  As a result, when $\rho(\theta_v,\thetahat) < \frac{\epsilon}{2}$, the minimum-distance rule \eqref{eq:m_hat} must output some $\hat{v}$ satisfying $d(v,\hat{v}) \le t$, yielding
\begin{equation}
\PP_v\bigg[ \rho(\theta_v,\thetahat) \ge \frac{\epsilon}{2} \bigg] \ge \PP_v[d(v,\hat{V}) > t]. \label{eq:ExcessDistProbPartial}
\end{equation}
% The remainder of the proof follows \eqref{eq:MinimaxEnd1}--\eqref{eq:MinimaxEnd4}, with the approximate recovery variant of Fano's inequality (\emph{cf.}, Theorem \ref{thm:Partial}) in place of the exact recovery version.
With the above tools in place, we proceed as follows:
\begin{align}
    \sup_{\theta \in \Theta} \PP_{\theta}\bigg[ \rho(\theta,\thetahat) \ge \frac{\epsilon}{2}\bigg]
    &\ge \max_{v=1,\dotsc,M} \PP_{v}\bigg[ \rho(\theta_v,\thetahat) \ge \frac{\epsilon}{2}\bigg] \label{eq:MinimaxEnd1p} \\
    &\ge \max_{v=1,\dotsc,M} \PP_v[d(v,\hat{V}) > t] \label{eq:MinimaxEnd2p}  \\
    &\ge \frac{1}{M} \sum_{v=1,\dotsc,M} \PP_v[d(v,\hat{V}) > t] \label{eq:MinimaxEnd3p}  \\
    &\ge 1 - \frac{ I(V;\Yv) + \log 2 }{ \log \frac{M}{\Nmax(t)} }, \label{eq:MinimaxEnd4p} 
\end{align}
where \eqref{eq:MinimaxEnd1p} follows by maximizing over a smaller set, \eqref{eq:MinimaxEnd2p} follows from \eqref{eq:ExcessDistProbPartial}, \eqref{eq:MinimaxEnd3p} lower bounds the maximum by the average, and \eqref{eq:MinimaxEnd4p} follows from Fano's inequality for approximate recovery (\emph{cf.}, Theorem \ref{thm:Partial}) and the fact that $I(V;\hat{V}) \le I(V;\Yv)$ by the data processing inequality (\emph{cf.}, Lemma \ref{lem:dpi}).  The proof of \eqref{eq:PackingResultApprox} is concluded by substituting \eqref{eq:MinimaxEnd4p} into \eqref{eq:ExcessDist_p} with $\epsilon_0 = \frac{\epsilon}{2}$, and taking the infimum over all estimators $\thetahat$. 

\subsection{Proof of Theorem \ref{thm:PackingConvex} (Reduction for Noisy Optimization)} \label{sec:pf_packing_opt}

We follow a similar proof to that of Theorem \ref{thm:Packing}, which gave an analogous result for estimation.  First, by Markov's inequality, we have
\begin{equation}
\sup_{f \in \Fc} \EE_f[ \ell_f(\hat{X}) ] \ge \sup_{f \in \Fc} \epsilon \cdot  \PP_f[ \ell_f(\hat{X}) \ge \epsilon ]. \label{eq:ConvExcessDist} 
\end{equation}
We proceed by analyzing the probability on the right-hand side.

Suppose that a random index $V$ is drawn uniformly from $\{1,\dotsc,M\}$, and the triplet $(\Xv,\Yv,\hat{X})$ is generated by running the optimization algorithm on $f_V$.  Moreover, given $\hat{X} = \hat{x}$, let $\hat{V}$ index the function among $\{f_1,\dotsc,f_M\}$ with the lowest corresponding value: $\hat{V} = \argmin_{v=1,\dotsc,M} f_{v}(\hat{x})$.  
% By the assumption that any point is $\epsilon$-optimal for at most one of the functions, we find that if $\ell_{f_v}(\hat{x}) \le \epsilon$ (i.e., $\hat{x}$ is $\epsilon$-optimal for $f_v$), it must hold that $\hat{V} = v$.  Hence, we have
By the assumption that any point $x$ satisfies $\ell_{f_v}(x) \le \epsilon$ at most one of the functions, we find that if $\ell_{f_v}(\hat{x}) \le \epsilon$, then it must hold that $\hat{V} = v$.  Hence, we have
\begin{equation}
\PP_{v}\big[ \ell_{f_v}(\hat{X}) > \epsilon \big] \ge \PP_v[\hat{V} \ne v], \label{eq:ConvExcessDistProb} 
\end{equation}
where $\PP_v$ is a shorthand for $\PP_{f_v}$.

With the above tools in place, we proceed as follows:
\begin{align}
    \sup_{f \in \Fc} \PP_{f}\big[ \ell_f(\hat{X}) \ge \epsilon \big]
    &\ge \max_{v=1,\dotsc,M} \PP_{v}\big[  \ell_{f_v}(\hat{X}) \ge \epsilon \big] \label{eq:ConvexEnd1} \\
    &\ge \max_{v=1,\dotsc,M} \PP_{v}[\hat{V} \ne v] \label{eq:ConvexEnd2}  \\
    &\ge \frac{1}{M} \sum_{v=1,\dotsc,M} \PP_{v}[\hat{V} \ne v] \label{eq:ConvexEnd3}  \\
    &\ge 1 - \frac{ I(V;\Xv,\Yv) + \log 2 }{ \log M }, \label{eq:ConvexEnd4} 
\end{align}
where \eqref{eq:ConvexEnd1} follows by maximizing over a smaller set, \eqref{eq:ConvexEnd2} follows from \eqref{eq:ConvExcessDistProb}, \eqref{eq:ConvexEnd3} lower bounds the maximum by the average, and \eqref{eq:ConvexEnd4} follows from Fano's inequality (\emph{cf.}, \eqref{eq:Fano3} in Theorem \ref{thm:Fano}) and the fact that $I(V;\hat{V}) \le I(V;\Xv,\Yv)$ by the data processing inequality (\emph{cf.}, third part of Lemma \ref{lem:mi_adaptive}).  The proof of \eqref{eq:PackingConvex} is concluded by substituting \eqref{eq:ConvexEnd4} into \eqref{eq:ConvExcessDist} and taking the infimum over all $\hat{X}$.  For $M = 2$, we obtain \eqref{eq:PackingResultM2} in the same way upon replacing \eqref{eq:ConvexEnd4} by the version of Fano's inequality for $M = 2$ given in Remark \ref{rem:FanoWeaken}.

\section*{Acknowledgments}

J.~Scarlett was supported by an NUS startup grant.  V.~Cevher was supported by the European Research Council (ERC) under the European Union's Horizon 2020 research and innovation programme (grant agreement 725594 -- time-data).

\bibliographystyle{IEEEtran}
\bibliography{JS_References}

\end{document}